\documentclass{lmcs}
\pdfoutput=1

\usepackage{lastpage}
\lmcsdoi{16}{3}{4}
\lmcsheading{}{\pageref{LastPage}}{}{}%
{Feb.~24,~2016}{Jul.~13,~2020}{}

\usepackage[english]{babel}
\usepackage[utf8]{inputenc}
\usepackage[T1]{fontenc}
\usepackage[matrix,arrow,curve]{xy}
\usepackage{amssymb}
\usepackage{tikz}
\usetikzlibrary{calc}
\usepackage{macros}
\usepackage{longtable}
\usepackage{xspace}

\def\strid1#1{\vcenter{\hbox{\includegraphics{#1.pdf}}}}
\def\strie1#1#2{\vcenter{\hbox{\includegraphics[scale=#1]{#2.pdf}}}}

\tikzset{gedge/.style={->,shorten <=2,shorten >=2}}

\newcommand{\svxym}[1]{\vcenter{\xymatrix@R=3ex@C=3ex{#1}}}

\renewcommand{\paragraph}[1]{\bigskip\noindent\emph{#1}}
\renewcommand{\leq}{\leqslant}
\renewcommand{\geq}{\geqslant}

\theoremstyle{definition}

\theoremstyle{plain}
\newtheorem{conjecture}[thm]{Conjecture}

\usepackage{etoolbox}

\allowdisplaybreaks

\begin{document}
\title{Directed Homotopy in Non-Positively Curved Spaces}
\author{Éric Goubault}
\author{Samuel Mimram}


\address{LIX, CNRS, École Polytechnique, Institut Polytechnique de Paris, 91128 Palaiseau Cedex, France}

\email{\href{mailto:eric.goubault@polytechnique.edu}{\texttt{eric.goubault@polytechnique.edu}}}
\email{\href{mailto:samuel.mimram@lix.polytechnique.fr}{\texttt{samuel.mimram@lix.polytechnique.fr}}}

\subjclass{
Mathematics of computing~Algebraic topology ; 
Theory of computation~Parallel computing models
}



\begin{abstract}
  A semantics of concurrent programs can be given using precubical sets, in
  order to study (higher) commutations between the actions, thus encoding the
  ``geometry'' of the space of possible executions of the program. Here, we
  study the particular case of programs using only mutexes, which are the most
  widely used synchronization primitive. We show that in this case, the
  resulting programs have \emph{non-positive curvature}, a notion that we
  introduce and study here for precubical sets, and can be thought of as an
  algebraic analogue of the well-known one for metric spaces. Using this it, as
  well as categorical rewriting techniques, we are then able to show that
  \emph{directed and non-directed homotopy coincide} for directed paths in these
  precubical sets. Finally, we study the geometric realization of precubical
  sets in metric spaces, to show that our conditions on precubical sets actually
  coincide with those for metric spaces. Since the category of metric spaces is
  not cocomplete, we are lead to work with generalized metric spaces and study
  some of their properties.
  %
\end{abstract}

\maketitle

\setcounter{tocdepth}{3}
\tableofcontents


Deep and fruitful links have been unraveled over the recent years between
concurrent programs and topological spaces. The starting point of those was the definition of the
so-called ``geometric semantics'', which to a concurrent program associates a
space whose points correspond to states of the program, paths to executions, and
homotopies between paths to equivalence of execution traces up to permutation of
independent actions. To be precise, in order to take in account the orientation of time,
which makes the execution of actions irreversible, one has to actually consider
spaces equipped with a notion of direction, which specifies directed paths and directed
homotopies. This point of view has lead to both theoretical and practical
applications, which are based on the fact that the study of the geometrical
features of those spaces brings us information about the possible executions of
the program, without having to consider all the possible interleavings of
actions of the various threads constituting the program. We refer to reader
to~\cite{goubault2000geometry}, as well as the recent book about the
subject~\cite{datc}, for more details.

All these techniques have thus been developed with the aim of using
mathematical tools for computing geometric invariants in order to ease the
verification of concurrent programs. However, because of the aforementioned
presence of time direction, most of those tools (such as homology groups,
homotopy groups, etc.)\ cannot directly be used in order to compute
the information which is relevant from a computer
scientific point of view. In this article, we focus on a particular
restricted class of programs: concurrent programs using only binary mutexes as
a synchronization primitive. While being constrained enough to enjoy interesting
properties, this class is still realistic from a practical point of view and can
be used to express many of the classical synchronization algorithms. One of the
main results of this article is that, in the geometrical models associated to
those programs, \emph{homotopy coincides with directed homotopy}, thus allowing one to
use all tools from ordinary homotopy theory 
to characterize these execution models. 

\paragraph{The cube property.}
Apart from the general motivations stated above, the origins of this article lie
in the observation that a similar property, that we call here the \emph{cube
  property}, occurs in different forms in different contexts, and it turns out
that there are good reasons for this to be the case. This property is namely
\begin{enumerate}
\item one of the main properties satisfied by semantics in asynchronous
  transition systems (or precubical sets) of our class of programs, which
  roughly says that when the transition system contains transitions forming half
  of a cube then it also contains the other half:
  \[
    \svxym{
      &\sst{+++}&\\
      \sst{++-}\ar[ur]^{++0}\ar@{}[r]|{+00}&\sst{+-+}\ar[u]|{+0+}\ar@{}[r]|{00+}&\sst{-++}\ar[ul]_{0++}\\
      \sst{+--}\ar[u]^{+0-}\ar[ur]|{+-0}\ar@{}[rr]|{0-0}&&\sst{--+}\ar[ul]|{0-+}\ar[u]_{-0+}\\
      &\sst{---}\ar[ul]^{0--}\ar[ur]_{--0}&
    }
    \qquad\quad\Leftrightarrow\qquad\quad
    \svxym{
      &\sst{+++}&\\
      \sst{++-}\ar[ur]^{++0}\ar@{}[rr]|{0+0}&&\sst{-++}\ar[ul]_{0++}\\
      \sst{+--}\ar[u]^{+0-}\ar@{}[r]|{00-}&\sst{-+-}\ar[ul]|{0+-}\ar[ur]|{-+0}\ar@{}[r]|{-00}&\sst{--+}\ar[u]_{-0+}\\
      &\sst{---}\ar[ul]^{0--}\ar[u]|{-0-}\ar[ur]_{--0}&
    }
  \]
\item the so-called Yang-Baxter axiom satisfied by symmetries in monoidal
  categories, which can be depicted using string diagrams as
  \[
    \strid1{yb_nl_l}
    \qquad\quad=\qquad\quad
    \strid1{yb_nl_r}
  \]
\item the Gromov condition~\cite{gromov1987hyperbolic} which characterizes
  cubical complexes of non-positive curvature (also called \NPC or CAT(0)
  spaces).
\end{enumerate}
The correspondence between the first two will be one of the main technical tools
used in order to show the coincidence between homotopy and its directed variant
mentioned above. The correspondence with the third point, reveals the main
characteristics of the geometry of the spaces obtained as semantics of programs:
if we realize them as metric spaces, \emph{they are non-positively curved}. This
means that if we draw a triangle in such a space, whose sides are geodesics, it
will appear to be thinner than usual. This is typically the case for hyperbolic
spaces (on the left) where a triangle (also pictured in the middle) is typically
thinner than an usual triangle in a flat space (on the right):
\[
  \includegraphics[scale=.5]{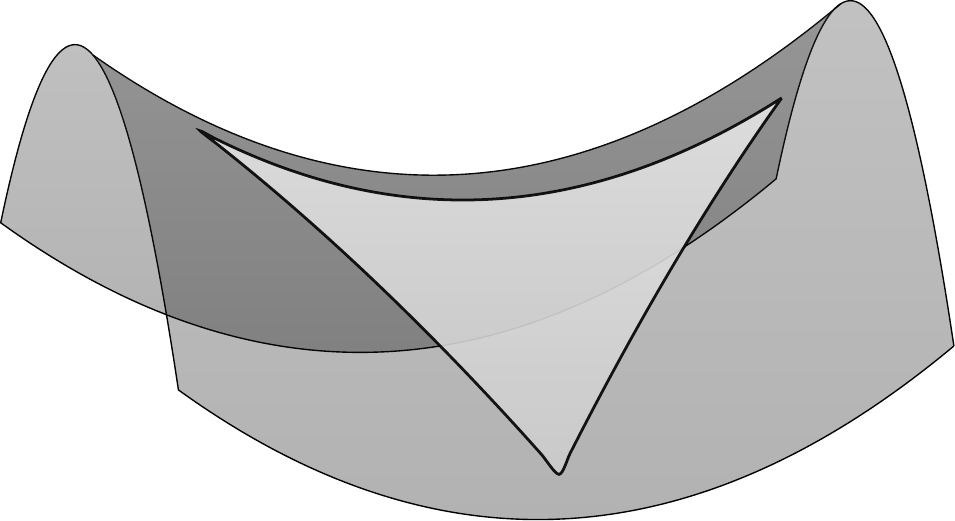}
  \qquad\qquad
  \begin{tikzpicture}[scale=.7]
    \path (0,0) coordinate (A) (4,1) coordinate (B) (2,-2) coordinate (C);
    \draw (A) to[bend right=20] (B) to[bend right=10] (C) to[bend right=10] (A);
  \end{tikzpicture}
  \qquad\qquad
  \begin{tikzpicture}[scale=.7]
    \path (0,0) coordinate (A) (4,1) coordinate (B) (2,-2) coordinate (C);
    \draw (A) to (B) to (C) to (A);
  \end{tikzpicture}
\]

\paragraph{Related work.}
Since the topics covered in the various sections of the article are quite
distinct (but of course closely related), we have mentioned in each section the
related work specific to it, and only briefly expose here transversal
references.
In particular, the cube property that we use for defining \NPC precubical sets
and characterizing precubical semantics of concurrent programs and with binary
mutexes has already been used in various works, as we now briefly recall.
%
On the one hand this condition characterizes domains arising from event
structures, which were introduced by Winskel~\cite{winskel1980events,
  winskel1993models}. In particular, an equivalence of categories between
deterministic labeled event structures and generalized trace languages (which
satisfy three cube properties!) is given by Sassone, Nielsen and Winskel
in~\cite{sassone1993deterministic}.
This line of work is quite relevant here because event structures can be thought
of as ``unfolded'' variants of the programs we consider: in particular, our use
of binary mutexes corresponds to the binary conflict relation in event
structures.
However, the uses of cube conditions in order to model
concurrent processes dates back to Stark's concurrent transition
systems~\cite{stark1989concurrent}, to the work of Panangaden and
Shanbhogue~\cite{panangaden1990stability} (where it is linked with a stability
property in domains), and to Droste~\cite{droste1994kleene} and
Kuske~\cite{kuske1994nondeterministic} who showed that concurrency automata with
two cube axioms correspond to dI-domains.
On the other hand, this condition was used by Gromov to characterize \NPC cubical
complexes~\cite{gromov1987hyperbolic} as mentioned above. The link between the
two was first explicitly made by Chepoi~\cite{chepoi2012nice} (who noticed that
event structures are in bijection with median
graphs~\cite{barthelemy1993median}, which in turn are in bijection with CAT(0)
cube complexes~\cite{roller1998poc, chepoi2000graphs}) and was recently
rediscovered~\cite{ardila2012geodesics} (reinventing the notion of event
structure under the prosaic name ``poset with inconsistent pairs''). These
relationships can be summarized in the following diagram, where the dashed arrows
schematize some of our contributions.
\[
\vxym{
  &\txt{\NPC\\cube\\complexes}\ar@{<-->}@/_3ex/[dr]\ar@{<->}[dd]|-{\txt{\cite{chepoi2012nice,ardila2012geodesics}}}\ar@{<->}[rr]|-{\txt{\cite{roller1998poc,chepoi2000graphs}}}\ar@{<->}[dr]|-{\txt{\cite{gromov1987hyperbolic}}}&&\ar@{<->}@/^8ex/[ddll]|-<<<<<<<<<{\txt{\cite{barthelemy1993median}}}\txt{median\\graphs}\\
  \txt{concurrent\\programs\\with\\mutexes}\ar@{<-->}[ur]\ar@{<-->}@/^3ex/[rr]&&\txt{cube\\property}&\\
  &\txt{event\\structures}\ar@{<->}[rr]|-{\txt{\cite{nielsen1981petri}}}\ar@{<->}[ur]|-{\txt{\cite{nielsen1994relationships, sassone1996models}}}&&\ar@{<->}[ul]|-{\txt{\cite{panangaden1990stability,kuske1994nondeterministic}}}\txt{dI domains}\\
}
\]
Other uses of the cube property can be found in Lévy~\cite{levy1978reductions} (this
seems to be the first occurrence of the condition) and
Melliès'~\cite{gonthier1992abstract,mellies2002axiomatic} study of
standardization in rewriting systems, see also~\cite{bezem2003term}, and
Dehornoy's work on Garside monoids~\cite{dehornoy2000completeness,
  dehornoy2013foundations} where it is used in order to ensure that a presented
category embeds in its enveloping groupoid (this is generalized to localizations
in~\cite{clerc2015presenting}), which is closely related to the situation
studied in Section~\ref{sec:homotopy-npc-cc}.



\paragraph{Contents of the article.}
We begin by recalling the notion of precubical set (\sect{pcs}), associate to
each program a precubical set (\sect{pcs-sem}) describing its execution traces
as well as (higher) commutations between their actions, and characterize their
main properties (\sect{pcs-npc}): they are geometric (\sect{geom-pcs}),
non-positively curved and satisfy in particular the cube condition
(\sect{pcs-npc-def}), which can be reformulated by a condition on links
(\sect{pcs-link}). We then introduce the notions of homotopy and dihomotopy
between paths in precubical sets (\sect{pcs-hom-dihom}) and study those in a
2-categorical context (\sect{2-cat-gpd}), which leads us to
Theorem~\ref{thm:fundt-embedding} which shows that homotopy and dihomotopy
relations coincide (\sect{hom-to-dihom}). We then recall and study in detail
generalized metric spaces (\sect{gms}), allowing us to realize a
precubical set as a metric space (\sect{pcs-greal}) which is shown in
Theorem~\ref{thm:pcs-cat0} to be non-positively curved when obtained as
the semantics of a concurrent program (\sect{npc}). 


\bigskip

We would like to particularly thank Emmanuel Haucourt and Jean Goubault-Larrecq
for many comments and discussions about the ideas developed in this article, as
well as the anonymous referees for helpful suggestions.


\section{Precubical semantics of concurrent programs}
\subsection{Precubical sets}
\label{sec:pcs}
We first recall the definition of precubical sets which will serve as the
primary algebraic structure for defining semantics for concurrent
programs. Those can be thought of as a generalization of the notion of
graph, which incorporates the information of commutation between two actions, and
more generally between~$n$ actions, making them quite natural to model
concurrent programs in the spirit of ``true concurrency''.


\begin{defi}
  The \emph{precubical category} $\pcc$ is the free monoidal category containing
  one object $1$, and two morphisms $\varepsilon^-,\varepsilon^+:0\to 1$,
  where~$0$ denotes the unit of the monoidal category. It can also be presented
  as the category whose objects are integers and morphisms are generated by
  morphisms $\varepsilon^\epsilon_{i,n}:n\to n+1$ with $\epsilon\in\set{-,+}$,
  $n\in\N$ and $0\leq i\leq n$, subject to the relations
  \begin{equation}
    \label{eq:pcc-rel}
    \varepsilon^{\epsilon'}_{i,n+1}\varepsilon^\epsilon_{j,n}
    \qeq
    \varepsilon^\epsilon_{j+1,n+1}\varepsilon^{\epsilon'}_{i,n}
  \end{equation}
  whenever $i\leq j$ and $\epsilon,\epsilon'\in\set{-,+}$. The category of
  \textbf{precubical sets} is the category~$\hat\pcc$ of presheaves over this
  category.
\end{defi}

\noindent
A precubical set $C\in\hat\pcc$ thus consists of a family of sets
$(C(n))_{n\in\N}$ together with
maps $\partial^\epsilon_{i,n}:C(n+1)\to C(n)$, with $0\leq i\leq n$ and
$\epsilon\in\set{-,+}$, where the map~$\partial^\epsilon_{i,n}$ is a notation
for $C(\varepsilon^\epsilon_{i,n})$, satisfying relations which are dual to
\eqref{eq:pcc-rel}, and we sometimes simply write $\partial^\epsilon_i$ when $n$
is clear from the context. An element $c\in C(n)$ is called an \emph{$n$-cube of
  $C$}, $\partial^\epsilon_{i,n}(c)$ is called a \emph{face} of $c$, and given a
morphism $\phi:n\to p$ in $\pcc$, $C(\phi)(c)$ is called an \emph{iterated face}
of $c$. The \emph{dimension} of a precubical set $C$ is the smallest integer
$d\in\N\uplus\set{\infty}$ such that $C(n)=\emptyset$ for~$n>d$. A precubical
set~$C$ is \emph{finite} when it is finite-dimensional and each $C(n)$ is
finite.

\begin{exa}
  \label{ex:dui-pcs}
  We write $\dui$ for the precubical set of dimension~1, called the
  \emph{standard interval},
  \[
  \vxym{
    x\ar[r]^a&y
  }
  \]
  with $\dui(0)=\set{x,y}$, $\dui(1)=\set{a}$, $\dui(n)=\emptyset$ for
  $n\geq 2$, $\partial_0^-(a)=x$, $\partial_0^+(a)=y$.
\end{exa}

\begin{exa}
  \label{exprecub2}
  A simple example of a precubical set~$C$ of dimension~2 is
  \[
  \vxym{
    x\ar[r]^a\ar[d]_b\ar@{}[dr]|-{\alpha}&y\ar[d]^{b'}\\
    x\ar[r]_{a}&y
  }
  \]
  with~$C(0)=\set{x,y}$, $C(1)=\set{a,b,b'}$ and $C(2)=\set{\alpha}$. The cells
  of dimension 0 are pictured as points, cells of dimension 1 as segments, and
  the cell of dimension 2 as a square.
  This precubical set can thus abstractly be thought of as a cylinder. Face maps
  are given by $\partial^-_0(a)=x$, $\partial^+_0(a)=y$,
  $\partial^-_0(\alpha)=b$, $\partial^+_0(\alpha)=b'$,
  $\partial^-_1(\alpha)=\partial^+_1(\alpha)=a$, etc.
\end{exa}


We recall some operations available on precubical sets, which will be used in
the following in order to define the semantics for programs.
A tensor product can be defined as follows.

\begin{defi}
  Given two precubical sets~$C$ and~$D$, their \emph{tensor product}
  $C\otimes D$ is the precubical set defined by
  \[
  (C\otimes D)(n)\qeq\coprod_{i+j=n}C(i)\times D(j)
  \]
  for $n\in\N$, and
  \[
  \partial^\epsilon_{k,n}(C\otimes D)(c,d)
  \qeq
  \begin{cases}
    (\partial^\epsilon_{k,n}(c),d)&\text{if $0\leq k<i$}\\
    (c,\partial^\epsilon_{k-i,n}(d))&\text{if $i\leq k<n$}
  \end{cases}
  \]
  for $n\in\N$, $0\leq k<n$, $\epsilon\in\set{0,1}$ and $(c,d)\in C(i)\times
  D(j)\subseteq(C\otimes D)(n)$ with $i+j=n$.
  This tensor equips the category $\hat\pcc$ with a monoidal structure, whose
  unit is the precubical set~$C$ consisting of one single $0$-cube.
\end{defi}

\noindent
As a presheaf category, the category of precubical sets is cocomplete, with
colimits being computed pointwise~\cite{maclane:cwm}. In particular, we denote
by $C\sqcup D$ the coproduct of two cubical sets, and given a precubical set $C$
and two $0$-cubes $c,d\in C(0)$, we write $\vquot Ccd$ for the precubical set
obtained from $C$ by identifying $c$ and $d$ (this precubical set can be
obtained as the expected coequalizer). Finally, given a set $X\subseteq C(0)$ of
$0$-cubes of a precubical set~$C$, we write $C\setminus X$ for the
\emph{precubical subset} of~$C$ consisting of $n$-cubes of~$C$ which do not contain an element of $X$ as
iterated face.

Given $n\in\N$, we write $\pcc_n$ for the full subcategory of $\pcc$ whose
objects~$k$ satisfy~$k\leq n$, and $\hat\pcc_n$ is thus the full subcategory
of~$\hat\pcc$ consisting of precubical sets of dimension at most~$n$.

\begin{defi}
  \label{def:truncation}
  The presheaves in $\pcc_n$ are called \emph{$n$-precubical sets}. The functor
  $\hat\pcc\to\hat\pcc_n$ induced by precomposition with the inclusion functor
  $\pcc_n\into\pcc$ is called the \emph{$n$-truncation functor}.
\end{defi}

\noindent
In particular, the category $\hat\pcc_1$ is the category of graphs, and every
precubical set~$C$ has an underlying graph: we thus sometimes refer to the
elements of $C(0)$ (\resp $C(1)$) as \emph{vertices} (\resp \emph{edges} or
\emph{transitions}), a \emph{path} in a precubical set is a path in the
underlying directed graph, etc. Also, a 2-cube is sometimes called a
\emph{square}.

\begin{defi}
  \label{def:independence}
  We define a symmetric relation $\tile$ on paths of length 2 as the smallest
  symmetric relation such that, given two paths $a\cc b$ and $b'\cc a'$ of
  length 2 constituted of edges $a$, $b$, $a'$ and $b'$, we have
  $a\cc b\tile b'\cc a'$ whenever there exists a square $\alpha$ such that
  \[
  a=\partial_{0,1}^-(\alpha)
  \qquad\qquad
  b=\partial_{1,1}^+(\alpha)
  \qquad\qquad
  b'=\partial_{1,1}^-(\alpha)
  \qquad\qquad
  a'=\partial_{0,1}^+(\alpha)
  \]
  or symmetrically. Graphically, we have a square~$\alpha$ as on the left, which 
  we often schematically picture as on the right
  \begin{equation}
    \label{eq:tile}
    \vxym{
      &x_{11}&\\
      x_{10}\ar[ur]^{b=\partial_{1,1}^+(\alpha)}&\alpha&\ar[ul]_{\partial_{0,1}^+(\alpha)=a'}x_{01}\\
      &\ar[ul]^{a=\partial_{0,1}^-(\alpha)}x_{00}\ar[ur]_{\partial_{1,1}^-(\alpha)=b'}&\\
    }
    \qquad\qquad\qquad\qquad
    \vxym{
      &x_{11}&\\
      x_{10}\ar[ur]^{b}&\tile&\ar[ul]_{a'}x_{01}\\
      &\ar[ul]^{a}x_{00}\ar[ur]_{b'}&\\
    }
  \end{equation}
  In this situation, we say that the coinitial transitions $a$ and $b'$ (\resp
  $b'$ and $a$) are \emph{independent}.
\end{defi}

As for any presheaf category, there is a full and faithful embedding
$Y:\pcc\to\hat\pcc$ given by the Yoneda functor, which is defined on objects as
the set of morphisms $Ynm=\pcc(m,n)$. Given $n\in\N$, the representable
precubical set $Yn$ is called the \emph{standard $n$-cube}. It can be shown that
$Yn$ is isomorphic to $\dui^{\otimes n}$, the tensor product of $n$ copies of
$\dui$ (see Example~\ref{ex:dui-pcs}). For instance, we have
\[
\begin{array}{c@{\qquad\qquad\qquad}c@{\qquad\qquad\qquad}c@{\qquad\qquad\qquad}c}
  \xymatrix{x}
  &
  \xymatrix{x\ar[r]^f&y}
  &
  \xymatrix{
    x\ar[r]^f\ar[d]_g\ar@{}[dr]|\tile&y_1\ar[d]^{g'}\\
    y_2\ar[r]_{f'}&z
    }
  \\
  Y0&Y1&Y2
\end{array}
\]
An explicit description of standard $n$-cubes can be given as follows. We will
see that it is often quite useful in order to perform computations.

\begin{lem}
  Given $n\in\N$, the cubes in $Yn$ are in bijection with strings in
  $\set{-,0,+}^n$, the $k$-cubes in $Ynk$ being the strings containing the
  letter~$0$ exactly $k$ times, and given $u\in Ynk$, the face $\partial_i^-(u)$
  (\resp $\partial_i^+(u)$) is obtained from $u$ by replacing the $i$-th
  letter~$0$ by $-$ (\resp $+$).
\end{lem}

\noindent
For instance, the preceding standard cubes are
\[
\begin{array}{c@{\qquad\qquad\qquad}c@{\qquad\qquad\qquad}c@{\qquad\qquad\qquad}c}
  \xymatrix{\varepsilon}
  &
  \xymatrix{\sst{-}\ar[r]^0&\sst{+}}
  &
  \xymatrix{
    \sst{--}\ar[r]^{-0}\ar[d]_{0-}\ar@{}[dr]|{00}&\sst{-+}\ar[d]^{0+}\\
    \sst{+-}\ar[r]_{+0}&\sst{++}
    }
  \\
  Y0&Y1&Y2
\end{array}
\]
(here $\varepsilon$ denotes the empty word). One can notice that $Ynk=\emptyset$
for $k>n$ and there is only one element in~$Ynn$ (the string $0^n$). The
\emph{standard hollow $n$-cube}, denoted $\partial Yn$, is the precubical set
obtained from $Yn$ by removing the cell in $Ynn$. For instance,
\[
\begin{array}{c@{\qquad\qquad\qquad}c@{\qquad\qquad\qquad}c@{\qquad\qquad\qquad}c}
  \xymatrix{}
  &
  \xymatrix{\sst-&\sst+}
  &
  \xymatrix{
    \sst{--}\ar[r]^{-0}\ar[d]_{0-}&\sst{-+}\ar[d]^{0+}\\
    \sst{+-}\ar[r]_{+0}&\sst{++}
    }
  \\
  \partial Y0&\partial Y1&\partial Y2
\end{array}
\]
There is an obvious inclusion morphism of precubical sets
$\partial Yn\to Yn$ embedding the standard hollow $n$-cube as the ``border'' of
the standard $n$-cube. A fact that will sometimes be useful is
that the faces of the standard $n$-cubes naturally index morphisms computing
iterated faces of an $n$-cube: given a precubical set $C$, an $n$\nbd{}cube
$x\in C(n)$ and~$u\in\set{-,0,+}^n$, we write
\[
\partial^u(x)
\qeq
\partial_{0}^{u_0}\circ\partial_{1}^{u_1}\circ\ldots\circ\partial_{n-1}^{u_{n-1}}(x)
\]
where $u_i\in\set{-,0,+}$ is the $i$-th letter of~$u$ and by convention
$\partial_{i}^0(x)=x$.

Given a set~$\labels$ of \emph{labels}, we will sometimes consider labeled
variants of precubical sets.

\begin{defi}
  A \emph{labeled} precubical set $(C,\ell)$ consists of a precubical set $C$
  together with a labeling function $\ell:C(1)\to\labels$ on edges, such that 
we have
  $\ell\circ\partial^-_{i,1}=\ell\circ\partial^+_{i,1}$ for $i\in\set{0,1}$.
\end{defi}

\noindent
Graphically, the condition on the labeling function amounts to supposing that
parallel edges of a square have the same label.

\begin{rem}
  If we suppose that the set $\labels$ is totally ordered, we can define a
  precubical set $\lpcs\labels$ such that the elements of $\lpcs\labels(n)$ are
  strictly increasing sequences of $n$ elements of $\labels$, and given
  $c\in\lpcs\labels(n+1)$, $\partial^-_{i,n}(c)=\partial^+_{i,n}(c)$ is obtained
  by removing the $i$-th element of the sequence~$c$. The category of labeled
  precubical sets can then alternatively be defined as the comma category
  $\hat\pcc/\lpcs$, see~\cite{goubault2012formal} for details.
\end{rem}

We will also make some use of presimplicial sets, which are a well-known variant
of precubical sets, which is based on triangles instead of squares. Formally,
the category of \emph{(augmented) presimplicial sets} is the presheaf
category $\hat\psc$, where $\psc$ is the category whose objects are finite
ordinals~$\intset{n}$ (given $n\in\N$, we write $\intset{n}$ for the set
$\set{0,\ldots,n-1}$) and morphisms are strictly increasing functions. Given a
presimplicial set $S$, the elements of $S(n)$ are called
\emph{$n$\nbd{}simplices}, and, given~$i$ with $0\leq i\leq n$, we write
$\partial_{i,n}:S(n+1)\to S(n)$ for the image by~$S$ of the function
$\intset{n}\to\intset{n+1}$, whose image does not contain~$i$, associating to an
$(n+1)$-simplex its $i$-th face.

\subsection{Precubical semantics of concurrent programs}
\label{sec:pcs-sem}
We now introduce an idealized concurrent programming language and provide its
semantics in precubical sets. A detailed presentation of this language can be
found in~\cite{datc}, in a variant which includes more realistic features, such
as memory and manipulation of values.

We suppose fixed a set $\actions$ of \emph{actions}, containing a particular
action $\nop\in\actions$ (standing for the action which has no effect), and a
set $\mutexes$ of \emph{mutexes}. A \emph{program}~$p$ is an expression
generated by the grammar
\[
p
\gramdef
\pone
\gramor
A
\gramor
\P a
\gramor
\V a
\gramor
p\pseq p
\gramor
p\por p
\gramor
p\ppar p
\gramor
\pwhile p
\]
where $A\in\actions$ is an arbitrary action and $a\in\mutexes$ is an arbitrary
mutex. The intended meaning of these constructions is the following one: $\pone$
is the empty program, $A\in\actions$ is an arbitrary instruction (for instance,
assigning a value to a memory cell or printing a message on the screen),
$p\pseq q$ is $p$ followed by $q$, $p\por q$ is a (non-deterministic) choice
between $p$ and $q$, $p\ppar q$ is the program obtained by running $p$ and $q$
in parallel, $\pwhile p$ is the program $p$ executed an arbitrary number of
times. In practice, it is desirable to forbid some actions (such as accessing
memory) from running in parallel. This is generally done using mutexes which are
particular kind of resources on which a program can perform two operations: a
mutex~$a$ can either be \emph{locked} using the instruction $\P a$ or
\emph{released} using the instruction $\V a$. The operational semantics of
a mutex is such that at most one subprogram can have locked a mutex at a time:
if a subprogram tries to lock an already-locked mutex, it is frozen until the
mutex is released.

The class of programs defined above is very general and thus difficult to reason
about. In practice, the programs which are written are such that the state of
mutexes only depends on the position in the program and not on the execution
that lead to it. We will thus restrict to such programs, called
\emph{conservative}, which can be formally defined as follows.

\begin{defi}
  The \emph{resource consumption} function $\Delta(p):\mutexes\to\Z$ of a
  program~$p$ is defined inductively by
  \begin{gather*}
    \begin{align*}
      \Delta(\pone)&\qeq0&&&\Delta(A)&\qeq0\\
      \Delta(\P a)&\qeq-\delta_a&&&\Delta(\V a)&\qeq\delta_a\\
      \Delta(p\pseq q)&\qeq\Delta(p)+\Delta(q)&&&\Delta(p\ppar q)&\qeq\Delta(p)+\Delta(q)\\
      \Delta(p\por q)&\qeq\Delta(p)&&\text{if $\Delta(p)=\Delta(q)$}\\
      \Delta(\pwhile p)&\qeq 0&&\text{if $\Delta(p)=0$}
    \end{align*}
  \end{gather*}
  Above, $0$ denotes the constant function, the sum of functions is computed
  pointwise, and $\delta_a:\mutexes\to\Z$ is the function such that
  $\delta_a(a)=1$ and $\delta_a(b)=0$ for $b\neq a$. Note that, because of the
  side conditions in the cases for choice and loop, this function might not be
  defined. A program~$p$ such that $\Delta(p)$ is defined is called
  \emph{conservative}.
\end{defi}

\begin{exa}
  The program $\P a\pseq\pwhile{\pa{\V a\pseq\P a}}\pseq\V a$ is conservative,
  as well as $\P a\ppar\P b$, but not the program $\pwhile{\P a}$ (the number of
  times $a$ is taken depends on the number of loops executed) nor $\P a\por\P b$ (which
  mutex is taken in the end depends on the chosen branch).
\end{exa}

In the following, we will consider the following set of labels:
\[
\labels
\qeq
\actions\uplus\setof{\P a}{a\in\mutexes}\uplus\setof{\V a}{a\in\mutexes}
\]
To any program~$p$ can be inductively associated a precubical set~$\cs{p}$
labeled in~$\labels$, together with two vertices $\vbeg p$ and $\vend p$, as
follows:
\begin{itemize}
\item \emph{empty}: $\cs\pone$ is the precubical set reduced to one vertex (and
  $\vbeg\pone$ and $\vend\pone$ are both equal to this vertex):
  \[
  \cs\pone\qeq\vbeg\pone=\vend\pone
  \]
\item \emph{action}: $\cs A$ is the graph
  \[
  \cs A\qeq \vxym{\vbeg A\ar[r]^A&\vend A}
  \]
\item \emph{lock}: $\cs{\P a}$ is the graph
  \[
  \cs{\P a}\qeq\vxym{\vbeg{\P a}\ar[r]^{\P a}&\vend{\P a}}
  \]
\item \emph{release}: $\cs{\V a}$ is the graph
  \[
  \cs{\V a}\qeq \vxym{\vbeg{\V a}\ar[r]^{\V a}&\vend{\V a}}
  \]
\item \emph{sequence}:
  \[
  \cs{p\pseq q}\qeq(\cs p\sqcup\cs q)[\vend p=\vbeg q]
  \]
  with $\vbeg{p\pseq q}=\vbeg p$ and $\vend{p\pseq q}=\vend q$,
\item \emph{choice}:
  \[
  \cs{p\por q}
  \qeq
  (\cs{p'}\sqcup\cs{q'})[\vbeg{p'}=\vbeg{q'}][\vend{p'}=\vend{q'}]
  \]
  with $\vbeg{p\por q}=\vbeg{p'}=\vbeg{q'}$ and $\vend{p\por
    q}=\vend{p'}=\vend{q'}$, where $p'=\nop\pseq p$ and $q'=\nop\pseq q$,
\item \emph{parallel}:
  \[
  \cs{p\ppar q}
  \qeq
  \cs{p}\otimes\cs{q}
  \]
  with $\vbeg{p\ppar q}=(\vbeg p,\vbeg q)$ and $\vend{p\ppar q}=(\vend p,\vend
  q)$,
\item \emph{loop}:
  \[
  \cs{\pwhile p}
  \qeq
  (\cs{p'}\sqcup\cs\nop)[\vbeg{p'}=\vend{p'}][\vbeg{p'}=\vbeg\nop]
  \]
  with $\vbeg{\pwhile p}=\vbeg{p'}$ and $\vend{\pwhile p}=\vend\nop$, where
  $p'=\nop\pseq p$.
\end{itemize}
The last cases can be illustrated as follows:
\begin{align*}
  \cs{p\pseq q}&\qeq
  \begin{tikzpicture}[baseline]
    \gvert{(0,0)} node[left] {$\vbeg{p\pseq q}=\vbeg p$};
    \draw[dashed] (1,0) ellipse (1 and 0.5) node {$\cs p$};
    \gvert{(2,0)} node[left] {$\vend p$} node[right] {$\vbeg q$};
    \draw[dashed] (3,0) ellipse (1 and 0.5) node {$\cs q$};
    \gvert{(4,0)} node[right] {$\vend q=\vend{p\pseq q}$};
  \end{tikzpicture}
  \\
  \cs{p\por q}&\qeq
  \begin{tikzpicture}[baseline]
    \gvert{(0,0)} node[left] {$\vbeg{p\por q}$};
    \gvert{(1,1)} node[left] {$\vbeg{p}$};
    \gvert{(1,-1)} node[left] {$\vbeg{q}$};
    \draw (0,0) edge[gedge] node[above left] {$\nop$} (1,1);
    \draw (0,0) edge[gedge] node[below left] {$\nop$} (1,-1);
    \draw[fill=black] (3,0) circle (0.04) node[right] {$\vend{p}=\vend{q}=\vend{p\por q}$};
    \draw[dashed,rotate around={-22.5:(2,.5)}] (2,.5) ellipse (1.1 and 0.5) node {$\cs{p}$};
    \draw[dashed,rotate around={22.5:(2,-.5)}] (2,-.5) ellipse (1.1 and 0.5) node {$\cs{q}$};
  \end{tikzpicture}
  \\
  \pwhile p&\qeq
  \begin{tikzpicture}[baseline]
    \gvert{(0,0)} circle (0.04) node[left] {$\vbeg{\pwhile p}=\vend p$};
    \gvert{(2,0)} node[right] {$\vbeg p$};
    \draw[dashed] (1,0) ellipse (1 and 0.3) node {$\cs{p}$};
    \draw (0,0) edge[out=85,in=95,gedge] node[above] {$\nop$} (2,0);
    \draw[fill=black] (2,-1) circle (0.04) node[right] {$\vend{\pwhile{p}}$};
    \draw (0,0) edge[out=-75,in=180,gedge] node[below] {$\nop$} (2,-1);
  \end{tikzpicture}
\end{align*}

\begin{exa}
  \label{ex:prog-pcs}
  The precubical sets associated to the program $(\P a\ppar\P b)\por\P c$ and to
  the program $\pa{\P a\pseq \V a}\ppar\pa{\P a\pseq\V a}$ are respectively
  \[
  \vxym{
    &x\ar[dl]_\nop\ar[dr]^\nop\\
    y\ar[r]^{\P a}\ar[d]_{\P b}\ar@{}[dr]|-\tile&y'\ar[d]_{\P b}&z\ar[dl]^{\P c}\\
    y''\ar[r]_{\P b}&z'
  }
  \qquad\qquad\qquad
  \vxym{
    \ar[r]^{\P a}\ar[d]_{\P a}\ar@{}[dr]|\tile&\ar[d]|{\P a}\ar[r]^{\V a}\ar@{}[dr]|\tile&\ar[d]^{\P a}\\
    \ar[d]_{\V a}\ar[r]|{\P a}\ar@{}[dr]|\tile&\ar[d]|{\V a}\ar[r]|{\V a}\ar@{}[dr]|\tile&\ar[d]^{\V a}\\
    \ar[r]_{\P a}&\ar[r]_{\V a}&\\
  }
  \]
\end{exa}

\begin{rem}
  The transitions labeled by~$\nop$ abstractly represent the evaluation of
  conditions (depending on which the branch of a conditional branching is chosen
  or a loop is stopped). If our language contained boolean expressions, they
  should actually be such, see~\cite{datc}. Notice that these transitions cannot
  be contracted to a point in the semantics because it not be correct anymore:
  for instance the semantics of $\P a\por\pone$ would contain a loop.
\end{rem}

\noindent
A path~$t$ of length~$n$ in~$\cs{p}$, where the labels of its vertices are
$l_1,l_2,\ldots,l_n\in\labels$, can be seen as a program
$l_1\pseq l_2\pseqs l_n$ (\ie a sequence of actions and operations on mutexes),
and we write~$\Delta(t):\mutexes\to\Z$ for its resource consumption. When the
program is conservative, resource consumption only depends on the endpoints of
paths, which explains the terminology, by analogy with the situation in physics:

\begin{lem}
  Given a conservative program~$p$ and two paths $t$ and $u$ in $\cs{p}$ with
  the same source and the same target, we have $\Delta(t)=\Delta(u)$.
\end{lem}

\noindent
It can be observed that, for any program~$p$, there is a path in~$\cs{p}$ from
$\vbeg{p}$ to any vertex~$x$. This fact and preceding lemma ensure that the
following definition makes sense:

\begin{defi}
  The \emph{resource potential} $\rpot{p}:\cs{p}(0)\to(\mutexes\to\Z)$ of a
  conservative program~$p$ is the function which to a vertex~$x\in\cs{p}(0)$
  associates $\Delta(t)$ where $t$ is any path from $\vbeg{p}$ to~$x$.
\end{defi}

\noindent
A vertex~$x$ such that there exists $a \in \mutexes$ where 
$\rpot{p}(x)<-1$ or $\rpot{p}(x)>0$ is said to be
\emph{forbidden}: the operational semantics of mutexes enforces that such a
state of the program can never be reached, because a mutex can be taken at most
once and not released more than taken. The precubical semantics of a
conservative program~$p$ is thus obtained from~$\cs{p}$ by removing those
forbidden vertices. Given a precubical set~$C$ and a set~$X\subseteq C(0)$ of
vertices of~$C$, recall that we write $C\setminus X$ for the precubical set
obtained from~$C$ by removing all the cubes having an element of~$X$ as iterated
face.

\begin{defi}
\label{precubsem:defi}
  Given a conservative program~$p$, its \textbf{precubical semantics} is the
  precubical set~$\pcs{p}$ defined by
  \[
  \pcs{p}
  \qeq
  \cs{p}\setminus X
  \qquad\text{where}\qquad
  X\qeq\setof{x\in C(0)}{\rpot{p}(x)<-1 \lor\rpot{p}(x)>0}
  \]
\end{defi}

\begin{exa}
  The precubical semantics of the program
  $\pa{\P a\pseq \V a}\ppar\pa{\P a\pseq\V a}$
  \[
  \vxym{
    \ar[r]^{\P a}\ar[d]_{\P a}&\ar[r]^{\V a}&\ar[d]^{\P a}\\
    \ar[d]_{\V a}&&\ar[d]^{\V a}\\
    \ar[r]_{\P a}&\ar[r]_{\V a}&\\
  }
  \]
  is obtained by removing then central vertex (and associated transitions) from
  the graph of Example~\ref{ex:prog-pcs}.
\end{exa}

\noindent
An \emph{execution trace}~$t$ of a program~$p$ is a path in~$\pcs{p}$ starting
from the beginning vertex $\vbeg{p}$. It can be checked that those are in
bijection with execution traces as they would be defined in a reasonable
operational semantics~\cite{datc}. Moreover, $n$-cubes of dimension $2$ and
higher can be used to construct a useful notion of equivalence of execution
traces, called homotopy, as explained in Section~\ref{sec:homotopy-npc-cc}.

\subsection{Non-positively curved precubical sets}
\label{sec:pcs-npc}
In this section, we study the general properties of precubical sets arising as
semantics of programs. This will lead us to formulate a set of conditions
characterizing what we call \emph{non-positively curved precubical sets}: it
will be shown in Section~\ref{sec:Gromov-condition} that their geometric
realizations are precisely non-positively curved precubical complexes.

\subsubsection{Geometric precubical sets}
\label{sec:geom-pcs}
We begin by recalling the definition of ``geometric'' precubical
sets~\cite{fajstrup2005dipaths}. The reason for this terminology should become
apparent in Proposition~\ref{prop:geometric-pcs-real}: their geometric
realizations are cubical complexes. 

\begin{defi}
  \label{def:geometric-pcs}
  A precubical set~$C$ is \textbf{geometric} when it satisfies the following
  conditions:
  \begin{enumerate}
  \item \emph{no self-intersection}: two distinct iterated faces of a cube in
    $C$ are distinct, \ie given two morphisms $\phi,\psi:m\to n$ in $\pcc$, if there
    exists $c\in C(n)$ such that $C(\phi)(c)=C(\psi)(c)$ then $\phi=\psi$,
  \item \emph{maximal common faces}: two cubes admitting a common face admit a
    maximal common face.
  \end{enumerate}
\end{defi}

\begin{exa}
  \label{ex:pcs-geom}
  Consider the following precubical sets of dimension~1
  \[
  \begin{array}{c@{\qquad\qquad}c@{\qquad\qquad}c@{\qquad\qquad}c}
    \vxym{x\ar@/^/[r]^a\ar@/_/[r]_b&y}
    &
    \vxym{x\ar@(ul,ur)^a}
    &
    \vxym{x\ar@/^/[r]^a\ar@{<-}@/_/[r]_b&y}
    &
    \vxym{&y\ar[dr]^b&\\x\ar[ur]^a&&\ar[ll]^cz}
    \\
    (1)&(2)&(3)&(4)
  \end{array}
  \]
  (1) and (3) are not self-intersecting but fail to have maximal common faces,
  (2) is self-intersecting and (4) is geometric. Notice that (4) and (3) are
  obtained by subdividing (2). However, it is not always possible to subdivide
  enough a precubical set in order to obtain a geometric
  one~\cite{fajstrup2005dipaths}.
\end{exa}

\begin{exa}
  \label{ex:2-squares-non-geom}
  \label{ex:pinned-squares}
  Consider the precubical set obtained from two copies of the standard square
  $Y2$ by identifying both vertices $\sst{-+}$ and both vertices $\sst{+-}$ in
  the two squares:
  \[
  \vxym{
    &&y_1\ar[dr]_{g'}\ar@/^/[drr]^{b'}\\
    x\ar@/^/[urr]^a\ar@/_/[drr]_b&x'\ar[ur]_f\ar[dr]^g&&y'&y\\
    &&y_2\ar[ur]^{f'}\ar@/_/[urr]_{a'}
  }
  \]
  with $a\cc b'\tile a'\cc b$ and $f\cc g'\tile g\cc f'$. This precubical set is
  not self-intersecting but does not have maximal common faces.
\end{exa}

\noindent
Interestingly, the geometricity condition can be characterized as follows on the
category of elements of a precubical set (this classical notion is recalled in
Section~\ref{sec:El}):

\begin{lem}
  \label{lem:geom-El}
  A precubical set~$C$ is geometric if and only if its category of
  elements~$\El(C)$ is a poset such that any two elements with a common lower
  bound have a greatest lower bound, \ie for every $x\in\El(C)$ the slice
  category $x/\El(C)$ is a meet semilattice.
\end{lem}
\begin{proof}
  The assumption that~$C$ has no self-intersection implies that any two
  morphisms with the same source and the same target in~$\El(C)$ are equal, and
  since the non-trivial morphisms are from a cube to a cube of strictly higher
  dimension, there are no non-trivial endomorphisms. The category~$\El(C)$ is
  thus a poset. The condition on lower bounds is a direct reformulation of the
  ``maximal common face'' condition.
\end{proof}

\noindent
The precubical semantics of most programs is a geometric precubical set,
excepting for some degenerated ones. For instance, the semantics of
$\pone\por\pone$ is
\[
\vxym{
  \ar@/^/[r]^\nop\ar@/_/[r]_\nop&
}
\]
which is not self-intersecting. Not allowing the construction $\pone$ to occur
in programs is enough to ensure that their semantics is geometric and does not
remove any expressive power to the language in practice. We will implicitly
assume that this is the case in the rest of the paper.
%
By an easy induction on the structure of programs, the following can be shown:

\begin{lem}
  \label{lem:cs-geom}
  The precubical semantics $\pcs{p}$ of a program~$p$ without occurrences
  of~$\pone$ is a geometric precubical set.
\end{lem}


\subsubsection{Non-positively curved precubical sets}
\label{sec:pcs-npc-def}
In this section, we introduce the notion of non-positively curved precubical
set, which will be shown to be an algebraic counterpart of the usual notion of
non-positively curved space in \sect{Gromov-condition}. The conditions for being
such a precubical set are easily formulated as lifting properties, which are of
the following general form.

\begin{defi}
  Given a morphism $f:D\to E$ between precubical sets, we say that a precubical
  set~$C$ \emph{lifts at least once} (\resp \emph{lifts at most once}, \resp \emph{lifts
    uniquely}) the morphism $f$, when for every morphism $h:D\to C$, there
  exists a morphism (\resp at most one morphism, \resp a unique morphism)
  $g:E\to C$ such that $h=g\circ f$.
  \[
  \vxym{
    D\ar[d]_-f\ar[r]^-h&C\\
    E\ar@{.>}[ur]_-g
  }
  \]
\end{defi}

\noindent
We will be mainly interested in the following lifting properties, which all more
or less express the fact that, as soon as a precubical set contains a subset
which could be the border of a cube, it should actually contain a cube with this
border.

\begin{defi}
  Suppose given a precubical set~$C$. We say that $C$ has
  \begin{enumerate}
  \item \emph{no parallel edges} when it lifts at most once the canonical
    inclusion $\partial Y1\into Y1$,
  \item the \emph{at most one square closing property} when each of the four
    canonical inclusions of
    \begin{equation}
      \label{eq:square1}
      \svxym{
        &\sst{++}&\\
        \sst{-+}\ar[ur]^{0+}&&\\
        &\ar[ul]^{-0}\sst{--}&
      }
      \svxym{
        &\sst{++}&\\
        &&\ar[ul]_{+0}\sst{+-}\\
        &\sst{--}\ar[ur]_{0-}&
      }
      \qquad
      \svxym{
        &&\\
        \sst{-+}&&\sst{+-}\\
        &\ar[ul]^{-0}\sst{--}\ar[ur]_{0-}&
      }
      \qquad
      \svxym{
        &\sst{++}&\\
        \sst{-+}\ar[ur]^{0+}&&\ar[ul]_{+0}\sst{+-}\\
        &&
      }
    \end{equation}
    into
    \[
      Y2\qeq
      \svxym{
        &\sst{++}&\\
        \sst{-+}\ar[ur]^{0+}\ar@{}[rr]|{00}&&\ar[ul]_{+0}\sst{+-}\\
        &\ar[ul]^{-0}\sst{--}\ar[ur]_{0-}&
      }
    \]
    \ie
    \begin{align*}
      Y2\setminus\set{\sst{+-}}&\into Y2
      &
      Y2\setminus\set{\sst{-+}}&\into Y2
      \\
      Y2\setminus\set{\sst{++}}&\into Y2
      &
      Y2\setminus\set{\sst{--}}&\into Y2
    \end{align*}
    is lifted at most once,
  \item the \emph{cube property} when each of four canonical inclusions of
    \begin{gather}
      \label{eq:cube1}
      \svxym{
        &\sst{+++}&\\
        \sst{++-}\ar[ur]^{++0}\ar@{}[r]|{+00}&\sst{+-+}\ar[u]|{+0+}\ar@{}[r]|{00+}&\sst{-++}\ar[ul]_{0++}\\
        \sst{+--}\ar[u]^{+0-}\ar[ur]|{+-0}\ar@{}[rr]|{0-0}&&\sst{--+}\ar[ul]|{0-+}\ar[u]_{-0+}\\
        &\sst{---}\ar[ul]^{0--}\ar[ur]_{--0}&
      }
      \qquad\qquad
      \svxym{
        &\sst{+++}&\\
        \sst{++-}\ar[ur]^{++0}\ar@{}[rr]|{0+0}&&\sst{-++}\ar[ul]_{0++}\\
        \sst{+--}\ar[u]^{+0-}\ar@{}[r]|{00-}&\sst{-+-}\ar[ul]|{0+-}\ar[ur]|{-+0}\ar@{}[r]|{-00}&\sst{--+}\ar[u]_{-0+}\\
        &\sst{---}\ar[ul]^{0--}\ar[u]|{-0-}\ar[ur]_{--0}&
      }
      \\
      \label{eq:cube2}
      \svxym{
        &\sst{-+-}\ar[dl]_{-+0}\ar[dr]^{0+-}&\\
        \sst{-++}\ar@{}[r]|{-00}&\sst{---}\ar[dl]|{--0}\ar[u]|{-0-}\ar[dr]|{0--}&\sst{++-}\ar@{}[l]|{00-}\\
        \sst{--+}\ar[u]^{-0+}\ar[dr]_{0-+}\ar@{}[rr]|{0-0}&&\sst{+--}\ar[dl]^{+-0}\ar[u]_{+0-}\\
        &\sst{+-+}&\\
      }
      \qquad\qquad
      \svxym{
        &\sst{-+-}\ar[dl]_{-+0}\ar[dr]^{0+-}&\\
        \sst{-++}\ar[dr]|{0++}\ar@{}[rr]|{0+0}&&\sst{++-}\ar[dl]|{++0}\\
        \sst{--+}\ar[u]^{-0+}\ar[dr]_{0-+}&\sst{+++}\ar@{}[l]|{00+}\ar@{}[r]|{+00}&\sst{+--}\ar[dl]^{+-0}\ar[u]_{+0-}\\
        &\sst{+-+}\ar[u]|{+0+}&
      }
    \end{gather}
    into
    \[
      \partial Y3\qeq
      \svxym{
        &&\sst{+++}&&\\
        \sst{++-}\ar[urr]^{++0}&\sst{+-+}\ar[ur]|{+0+}&&&\sst{-++}\ar[ull]_{0++}\\
        \sst{+--}\ar[u]^{+0-}\ar[ur]|{+-0}&&&\sst{-+-}\ar[ulll]|{0+-}\ar[ur]|{-+0}&\sst{--+}\ar[ulll]|{0-+}\ar[u]_{-0+}\\
        &&\sst{---}\ar[ull]^{0--}\ar[ur]|{-0-}\ar[urr]_{--0}&
      }
    \]
    (we did not figure 2-cubes of $\partial Y3$ in order to keep the figure more
    or less readable), \ie
    \begin{align*}
      Y3\setminus\set{\sst{-+-}}&\into\partial Y3
      &
      Y3\setminus\set{\sst{+-+}}&\into\partial Y3
      \\
      Y3\setminus\set{\sst{+++}}&\into\partial Y3
      &
      Y3\setminus\set{\sst{---}}&\into\partial Y3
    \end{align*}
    is lifted at least once,
  \item the \emph{unique $n$-cube property} (\resp \emph{at most one $n$-cube
      property}) when it lifts uniquely (\resp at most once) the inclusion
    $\partial Yn\to Yn$.
  \end{enumerate}
\end{defi}

\begin{rem}
  \label{rem:unique-cube-property}
  When a precubical set satisfies the at most one square closing property, the
  liftings of \eqref{eq:cube1} and \eqref{eq:cube2} are necessarily unique.
\end{rem}

\begin{rem}
  \label{rem:missing-liftings}
  Notice that, in the cube property, we require lifting four inclusions of a
  ``half cube'' into the hollow cube, but there are four more natural inclusions
  that we could write (there are height in total, one corresponding to each
  vertex of the cube). However, those are superfluous, because the cube has
  non-trivial automorphisms. For instance, the lifting of the inclusion of
  \[
  \svxym{
    &\sst{+--}&\\
    \sst{+-+}\ar@{<-}[ur]^{+-0}\ar@{}[r]|{+00}&\sst{++-}\ar@{<-}[u]|{+0-}\ar@{}[r]|{00-}&\sst{---}\ar[ul]_{0--}\\
    \sst{+++}\ar@{<-}[u]^{+0+}\ar@{<-}[ur]|{++0}\ar@{}[rr]|{0+0}&&\sst{-+-}\ar[ul]|{0+-}\ar@{<-}[u]_{-0-}\\
    &\sst{-++}\ar[ul]^{0++}\ar@{<-}[ur]_{-+0}&
  }
  \]
  into $\partial Y3$ can be deduced from the lifting of the inclusion
  corresponding to the left of~\eqref{eq:cube1}.
\end{rem}

\noindent
By definition of the condition of being geometric
(Definition~\ref{def:geometric-pcs}), we have the following.

\begin{lem}
  \label{lem:geom-lifting}
  A geometric precubical set~$C$ lifts at most once the embedding
  $\partial Yn\into Yn$ for any $n\geq 1$. In particular, it has no parallel
  edges. Moreover, $C$ satisfies the at most one square filling property.
\end{lem}

\noindent
We introduce the following terminology, whose meaning should be explained in
Section~\ref{sec:Gromov-condition}.

\begin{defi}
  \label{def:pcs-npc}
  A precubical set which
  \begin{enumerate}
  \item is geometric,
  \item satisfies the cube property,
  \item and satisfies the unique $n$-cube property for $n\geq 3$,
  \end{enumerate}
  is said to be \textbf{non-positively curved} (or \emph{\NPC} for short).
\end{defi}

\begin{prop}
  \label{prop:sem-cube}
  The precubical semantics of a program is non-positively curved.
\end{prop}
\begin{proof}
  We have already seen in Lemma~\ref{lem:cs-geom} that the cubical semantics of
  a non-degenerated program is geometric. We show the other properties.
  %
  %
  Given a program~$p$, it is easy to show by induction on~$p$ that the
  precubical set~$\cs{p}$ is non-positively curved (essentially, one has to
  check that the properties (2) and (3) are preserved by gluing along vertices
  and taking tensor product). The precubical semantics is defined as
  $\pcs{p}=\cs{p}\setminus X$ where~$X$ is the set of forbidden
  vertices. Clearly, the unique $n$-cube property, with $n\geq 3$, is still
  satisfied for $\pcs{p}$. However, the cube property requires a little more
  care. Write $L$ (\resp $R$) for the precubical set on the left (\resp right)
  of~\eqref{eq:cube1}. Notice that the hollow 3-cube $\partial Y3$ can be
  obtained by gluing $L$ and~$R$ along their ``border'' (\ie by identifying the
  vertices and edges with the same names in both precubical sets). Suppose given
  a morphism~$L\to\pcs{p}$: we are going to show that it can be extended as a
  morphism $\partial Y3\to\pcs{p}$. By abuse of notation, we identify $L$ with
  its image in~$\pcs{p}$ and simply say that~$\pcs{p}$ ``contains''~$L$. Since
  $\cs{p}$ satisfies the cube property, it also contains~$R$. Thus, in
  $\pcs{p}$, $L$ can be completed as a hollow 3-cube unless the vertex
  $\sst{-+-}$ is forbidden. We write $A$, $B$ and $C$ for the actions
  respectively labeling the edges of the form $0\epsilon\epsilon'$,
  $\epsilon 0\epsilon'$ and $\epsilon\epsilon'0$, with
  $\epsilon,\epsilon'\in\set{-,+}$: $L$ and~$R$ in $\cs{p}$ are respectively
  labeled as
  \[
  \svxym{
    &\sst{+++}&\\
    \sst{++-}\ar[ur]^{C}\ar@{}[r]|{\tile}&\sst{+-+}\ar[u]|<<<{B}\ar@{}[r]|{\tile}&\sst{-++}\ar[ul]_{A}\\
    \sst{+--}\ar[u]^{B}\ar[ur]|{C}\ar@{}[rr]|{\tile}&&\sst{--+}\ar[ul]|{A}\ar[u]_{B}\\
    &\sst{---}\ar[ul]^{A}\ar[ur]_{C}&
  }
  \qquad\qquad\qquad
  \svxym{
    &\sst{+++}&\\
    \sst{++-}\ar[ur]^{C}\ar@{}[rr]|\tile&&\sst{-++}\ar[ul]_{A}\\
    \sst{+--}\ar[u]^{B}\ar@{}[r]|\tile&\sst{-+-}\ar[ul]|{A}\ar[ur]|{C}\ar@{}[r]|\tile&\sst{--+}\ar[u]_{B}\\
    &\sst{---}\ar[ul]^{A}\ar[u]|<<<{B}\ar[ur]_{C}&
  }
  \]
  Suppose that $\sst{-+-}$ is forbidden. Since none of the vertices in $L$ is
  forbidden, there exists a resource~$a$ such that either
  $\rpot{p}(\sst{-+-})(a)=1$ and $B$ is~$\V a$, or $\rpot{p}(\sst{-+-})=-2$ and
  $B$ is~$\P a$. Suppose that we are in the first case (the other case is
  similar), \ie $B=\V a$. Notice that we have
  $\rpot{p}(\sst{---})=0$. Necessarily, we have $A=\P a$ (otherwise
  $\rpot{p}(\sst{++-})=\rpot{p}(\sst{-+-})=1$ and the vertex~$\sst{++-}$ would
  be forbidden) and $C=\P a$ (otherwise $\sst{-++}$ would be forbidden). But in
  this case, we have $\rpot{p}(\sst{+-+})=-2$ and therefore $\sst{+-+}$ would be
  forbidden, contradicting the hypothesis that $\pcs{p}$ contains~$L$. All other
  cases can be handled by similar reasoning.
\end{proof}

\begin{rem}
  The precubical semantics of the program $A\por B$ is
  \[
  \svxym{
    &\ar[dr]^A\\
    \ar[ur]^\nop\ar[dr]_\nop&&\\
    &\ar[ur]_B
  }
  \]
  showing that the $2$-cube property is not generally satisfied for cubical
  semantics of programs.
\end{rem}

\subsubsection{The link condition}
\label{sec:pcs-link}
The Gromov condition for characterizing \NPC cubical complexes, which is
recalled in \sect{Gromov-condition}, is generally expressed in terms of a
``flagness'' condition on the ``link'' of vertices in the complex. We extend
here the definition of link to precubical sets and show that a similar
characterization of \NPC can be formulated in terms of those links.


\begin{defi}
  \label{def:pcs-link}
  Given a precubical set~$C$, the \textbf{link} of~$C$, denoted $\link(C)$, is
  the (augmented) presimplicial set whose $n$-simplices are pairs $(u,y)$ with
  $u\in\set{-,+}^n$ and $y\in C(n)$. Given $i$ with $i\in\intset{n}$, the $i$-th
  simplicial face of such an $n$-simplex $(u,y)$ is given by
  \[
  \partial_i(u,y)
  \qeq
  (u',\partial_i^{u_i}(y))
  \]
  where $u'$ is the word $u$ with the $i$-th letter removed. Given a
  vertex~$x\in C(0)$, the link of~$x$, denoted $\link(x)$, is the presimplicial
  subset of~$\link(C)$ whose simplices are those having~$(\varepsilon,x)$ as
  iterated $0$-face.
\end{defi}

\noindent
Note that, by definition, the link of a vertex has only one 0-simplex, and as
such can be seen as a non-augmented presimplicial set, shifting dimensions by
-1. This will implicitly be used in the following in order to draw pictures.

\begin{exa}
  \label{ex:link}
  Consider the precubical set~$C$ on the left. The link of the vertex~$x$ is
  shown on the right:
  \[
  \vxym{
    z'&y_1\ar[l]\ar[r]\ar@{}[dr]|\alpha&z\\
    y_3\ar[u]&x\ar[l]^a\ar[u]_b\ar[r]_c&y_2\ar[u]
  }
  \qquad\qquad\qquad\qquad\qquad
  \vxym{
    a&b\ar[r]^\alpha&c
  }
  \]
\end{exa}

\noindent
The link of a vertex~$x$ is constituted of the cubes having~$x$ as iterated
face. Its simplices can be represented as follows.

\begin{defi}
  Given $u\in\set{-,+}^n$, we write $\Lambda^u$ for the precubical subset of
  $\partial Yn$ whose cubes $v\in\set{-,0,+}^n$ are those having at least one
  letter in common with~$u$ (\ie there exists~$i$, with $0\leq i<n$, such that
  $v_i=u_i$).
\end{defi}

\begin{rem}
  \label{rem:Lambda-subcube}
  The precubical set $\Lambda^u$ can alternatively defined from $\partial Yn$ by
  removing all cubes having the vertex $\ol u$ as iterated face, where
  $\ol u=\ol{u_0}\ldots\ol{u_{n-1}}$ with $\ol -=+$ and $\ol +=-$.
\end{rem}

\begin{exa}
  The precubical sets \eqref{eq:square1} are respectively $\Lambda^{-+}$,
  $\Lambda^{+-}$, $\Lambda^{--}$ and $\Lambda^{++}$. Those of \eqref{eq:cube1}
  are respectively $\Lambda^{+-+}$ and $\Lambda^{-+-}$, and similarly for
  \eqref{eq:cube2}, see also Remark~\ref{rem:Lambda-cube}.
\end{exa}

\begin{rem}
  \label{rem:Lambda-cube}
  By definition, given $u\in\set{-,+}^n$, there is a canonical inclusion
  $\Lambda^u\into\partial Yn$ and thus a canonical inclusion $\Lambda^u\into Yn$
  by post-composing with the canonical inclusion $\partial Yn\into Yn$. By
  Remark~\ref{rem:missing-liftings}, the cube property is precisely the lifting
  property \wrt all possible such inclusions in the case $n=3$. In particular,
  the two precubical sets of \eqref{eq:cube1} are respectively $\Lambda^{+-+}$
  and $\Lambda^{-+-}$, and those of \eqref{eq:cube2} are respectively
  $\Lambda^{---}$ and $\Lambda^{+++}$.
\end{rem}




Geometric precubical sets which are non-positively curved can be characterized
by a lifting condition on the links. We write $Z:\psc\to\hat\psc$ for the Yoneda
embedding of the presimplicial category into presimplicial sets (the notation
$Y$ is reserved for the Yoneda embedding $Y:\pcc\to\hat\pcc$). Given $n\in\psc$,
the representable $Zn$ is called the \emph{standard $n$\nbd{}simplex}. As in the
case of precubical sets, it can be explicitly described as the presimplicial set
whose $k$-simplices are the sets $\set{j_0,\ldots,j_{k-1}}\subseteq\intset{n}$
such that $j_0<j_1<\ldots<j_{k-1}$ and, given $i$ with $0\leq i<k$ the $i$-th
face of such a simplex is given by removing the element $j_i$. We thus have
$Znk=\emptyset$ for $k>n$ and $Znn$ is reduced to one element (the simplex
$\intset{n}$). The \emph{standard hollow $n$-simplex} $\partial Zn$ is the
presimplicial set obtained from $Zn$ by removing the top-dimensional simplex
in~$Zn$, and there is a canonical inclusion $\partial Zn\into Zn$.

\begin{defi}
  \label{def:flag}
  A presimplicial~$S\in\hat\psc$ set is \textbf{flag} if, for every integer
  $n\geq 3$ and morphism $f:\partial Zn\to S$, there exists a unique morphism
  $g:Zn\to S$ making the diagram
  \[
  \vxym{
    \partial Zn\ar@{_{(}->}[d]\ar[r]^f&S\\
    Zn\ar@{.>}[ur]_g
  }
  \]
  commute, where the vertical arrow is the standard inclusion.
\end{defi}


\noindent
As previously, we sometimes say that a flag presimplicial set lifts uniquely the
inclusions $\partial Zn\into Zn$. In Theorem~\ref{thm:pcs-npc-flag}, we
characterize \NPC precubical sets in terms of their links, and this will follow
from the following lemmas.

\begin{lem}
  \label{lem:hsimplex-link}
  Suppose given a precubical set~$C$ and a vertex~$x\in C(0)$.
  \begin{itemize}
  \item There is a bijection between presimplicial morphisms
    $\phi:Zn\to\link(x)$ and pairs $(u,\psi)$ with $u\in\set{-,+}^n$ and
    $\psi:Y(n)\to C$ a precubical morphism such that $\psi(u)=x$.
  \item Similarly, there is a bijection between presimplicial morphisms
    $\phi:\partial Zn\to\link(x)$ and all precubical morphisms of the form
    $\psi:\Lambda^u\to C$, for some $u\in\set{-,+}^n$, such that $\psi(u)=x$.
  \end{itemize}
\end{lem}
\begin{proof}
  The first bijection follows easily from the Yoneda lemma: morphisms
  $\phi:Zn\to\link(x)$ are in bijection with $n$-simplices in $\link(x)$, \ie
  pairs $(u,y)$ with $u\in\set{-,+}^n$ and $y\in C(n)$ such that
  $\partial^u(y)=x$, which are in turn in bijection with pairs $(u,\psi)$ with
  $u\in\set{-,+}^n$ and $\psi:Y(n)\to C$ such that $\psi(u)=x$. The second
  bijection is a slight variant of the first one.
\end{proof}

\begin{rem}
  \label{rem:hsimplex-link-restr}
  It can moreover be checked that these bijections are compatible with
  restriction to the border, what we will generally leave implicit in the
  following. This means that given a presimplicial morphism $\phi:Zn\to\link(x)$
  corresponding to a precubical morphism $\psi:Y(n)\to C$ by the first
  bijection, the presimplicial morphism
  $\partial Zn\into Zn\overset\phi\to\link(x)$ corresponds the the precubical
  morphism $\partial Y(n)\into Y(n)\overset\psi\to C$ by the second bijection.
\end{rem}

\noindent
The following lemma shows that \NPC precubical sets, always satisfy a
higher-dimensional generalization of the cube property.

\begin{lem}
  \label{lem:higher-cube-prop}
  Suppose given an \NPC precubical set~$C$. Given $n\in\N$ with $n\geq 3$ and
  $u\in\set{-,+}^n$, $C$ lifts uniquely the canonical embedding
  $\Lambda^u\into Yn$.
\end{lem}
\begin{proof}
  Given $u\in\set{-,+}^n$ and $k$ with $0\leq k\leq n$, we write $\Lambda^u_k$
  for the subcomplex of $Yn$ whose cubes are the strings $v=\set{-,0,+}^n$ which
  either have one letter in common with $u$ (\ie $v_i=u_i$ for some
  $i\in\intset{n}$) or have at most $k$ occurrences of the letter ``$0$''. In
  particular, we have $\Lambda^u_0=\Lambda^u$, $\Lambda^u_{n-1}=\partial Yn$ and
  $\Lambda^u_n=Yn$. We show that~$C$ lifts uniquely the embedding
  $\Lambda^u\into\Lambda^u_2$, and the embeddings
  $\Lambda^u_k\into\Lambda^u_{k+1}$ for $2\leq k<n$, from which we can
  conclude. In order to illustrate our proof, we suppose that $u=-^n$ (other
  cases can be obtained by symmetry) and that $n=5$ (but the arguments go
  through similarly in the general case).

  Case $k=0$. Suppose fixed a morphism $\phi:\Lambda^u\to C$. Given any triple
  of indices $0\leq i_0<i_1<i_2<n$, there is an embedding
  $\lambda_{i_0,i_1,i_2}:\Lambda^{---}\into\Lambda^u$, sending a $k$-cube
  $v_0v_1v_2$ to the $k$-cube
  \[
  w\qeq {+^{i_0}}v_0{+^{(i_1-i_0-1)}}v_1{+^{(i_2-i_1-i_0-2)}}v_2{+^{(n-3-i_2)}}
  \]
  \ie $w$ is obtained by inserting $v_0$, $v_1$ and $v_2$ at positions $i_0$,
  $i_1$ and $i_2$ respectively in the word constituted of $n-3$ letters ``$+$'':
  \[
  w_i\qeq
  \begin{cases}
    v_{i_p}&\text{if $i=i_p$ for some $p\in\set{0,1,2}$}\\
    +&\text{otherwise}
  \end{cases}
  \]
  For instance, with $n=5$, we have $\lambda_{1,2,4}(v_0v_1v_2)={+}v_0v_1{+}v_2$
  and the image of $\lambda_{1,2,4}$ is shown on the left
  \[
  \xymatrix@R=3ex@C=6ex{
    &\sst{+-++-}\ar[dl]_{+-++0}\ar[dr]^{+0++-}&\\
    \sst{+-+++}\ar@{}[r]|{+-0+0}&\sst{+--+-}\ar[dl]|{+--+0}\ar[u]|{+-0+-}\ar[dr]|{0--}&\sst{++++-}\ar@{}[l]|{+00+-}\\
    \sst{+--++}\ar[u]^{+-0++}\ar[dr]_{+0-++}\ar@{}[rr]|{+0-+0}&&\sst{++-+-}\ar[dl]^{++-+0}\ar[u]_{++0+-}\\
    &\sst{++-++}&\\
  }
  \qquad
  \xymatrix@R=3ex@C=6ex{
    &\sst{+-++-}\ar[dl]_{+-++0}\ar[dr]^{+0++-}&\\
    \sst{+-+++}\ar[dr]|{+0+++}\ar@{}[rr]|{+0++0}&&\sst{++++-}\ar[dl]|{++++0}\\
    \sst{+--++}\ar[u]^{+-0++}\ar[dr]_{+0-++}&\sst{+++++}\ar@{}[l]|{+00++}\ar@{}[r]|{++0+0}&\sst{+--}\ar[dl]^{++-+0}\ar[u]_{++0+-}\\
    &\sst{++-++}\ar[u]|{++0++}&
  }
  \]
  By the cube property, the morphism
  $\Lambda^{---}\into\Lambda^u\overset\phi\to C$ extends to a morphism
  $\partial Y(n)\to C$ (\ie in the case of $\lambda_{1,2,4}$, $C$ also
  contains the other half-cube shown on the right above) and into~$Y3$ by the
  unique 3-cube property. By the at most one square closing property (which holds by
  Lemma~\ref{lem:geom-lifting}), the image of the image of the vertex $+^{n+1}$
  thus defined does not depend on the indices $i_0,i_1,i_2$. For instance, the
  embeddings $\lambda_{1,2,4}$ and $\lambda_{0,2,3}$ can both be completed and
  both define an image for the vertex $\sst{+++++}$: the one defined by
  $\lambda_{1,2,4}$ is the same as the one defined by $\lambda_{1,2,3}$ (the
  completions share the face $\sst{+00++}$), which in turn is the same as the
  one defined by $\lambda_{0,2,3}$ (the completions share the face
  $\sst{++00+}$). Similarly, the 1- and 2-cubes defined by multiples embeddings
  do not depend on the choice of the embedding. Therefore, the precubical
  complex~$C$ lifts the embedding $\Lambda^u\into\Lambda^u_2$. Finally, by the
  at most one square closing property again, the lifting is unique.

  Suppose given $k$ with $2\leq k<n$ and a morphism $\phi:\Lambda^u_k\to C$. For
  any $(k+1)$-uple of indices $0\leq i_0<\ldots<i_k<n$, there is an
  embedding $\lambda_{i_0,\ldots,i_k}:\partial Y(k+1)\into\Lambda^u_k$ sending a
  cube $v\in\set{-,0,+}^{k+1}$ to the cube $w\in\set{-,0,+}^n$ defined by
  \[
  w_i\qeq
  \begin{cases}
    v_{i_p}&\text{if $i=i_p$ for some $p\in\intset{k+1}$}\\
    +&\text{otherwise}
  \end{cases}
  \]
  By the unique $(k+1)$-cube property, the morphism
  $\partial Y(k+1)\into\Lambda^u_k\overset\phi\to C$ extends uniquely to a
  morphism $Y(k+1)\to C$. The precubical set $\Lambda^u_{k+1}$ can be obtained
  from $\Lambda^u_k$ by properly adding $(k+1)$-cubes: more precisely, one can
  easily check that $\Lambda^u_{k+1}$ is the colimit of the diagram of
  precubical sets
  \[
  \vxym{
    \Lambda^u_k&&Y(k+1)\\
    \partial Y(k+1)\ar[u]^{\lambda_{I_0}}\ar[urr]&\ldots&\ar[ull]^<<<<<<<<{\lambda_{I_p}}\ar[u]\partial Y(k+1)
  }
  \]
  where the $I_j$ range over subsets of cardinal $k+1$ of~$[n]$. Therefore, by
  universal property of the colimit, the morphism $\phi$ extends uniquely to a
  morphism $\Lambda^u_{k+1}\to C$.
\end{proof}

\noindent
Notice that previous lemma also shows that the canonical embedding
$\Lambda^u\into\partial Yn$ is lifted uniquely. When the precubical set is only
supposed to be geometric, the lifting does not generally exists, but the proof
can easily be adapted in order to show that it is necessarily unique (the
particular case with $n=3$ was already noticed in
Remark~\ref{rem:unique-cube-property}):

\begin{lem}
  \label{lem:higher-cube-prop-uniqueness}
  Suppose given a geometric precubical set~$C$. Given $n\in\N$ with $n\geq 3$
  and $u\in\set{-,+}^n$, $C$ lifts at most once the canonical embedding
  $\Lambda^u\into\partial Yn$.
\end{lem}

\begin{thm}
  \label{thm:pcs-npc-flag}
  A geometric precubical set~$C$ is non-positively curved if and only if for
  every vertex $x\in C(0)$, the presimplicial set $\link(x)$ is flag.
\end{thm}
\begin{proof}
  We first show the left-to-right implication and suppose that~$C$ is
  non-positively curved.
  We show that for every vertex~$x$ of~$C$ the presimplicial set $\link(x)$
  lifts the inclusion $\partial Zn\into Zn$.
  Suppose fixed a presimplicial morphism $\partial Zn\to\link(x)$. By
  Lemma~\ref{lem:hsimplex-link}, this amounts to fix a precubical morphism
  $\phi:\Lambda^u\to C$, for some $u\in\set{-,+}^{n}$, such that
  $\phi(u)=x$. By Lemma~\ref{lem:higher-cube-prop}, this morphism extends
  uniquely to a morphism $\psi:Y(n)\to C$, satisfying $\psi(u)=x$, which
  corresponds by Lemma~\ref{lem:hsimplex-link} to a morphism $Zn\to\link(x)$.

  We now show the right-to-left implication and suppose that $\link(x)$ is flag
  for every vertex~$x$ of~$C$. We first show that $C$ satisfies the cube
  property. Suppose given $u\in\set{-,+}^3$ and a precubical morphism
  $\phi:\Lambda^u\to C$, by Remark~\ref{rem:Lambda-cube} we want to show that it
  extends to a morphism $Y3\to C$ (the uniqueness is granted by
  Remark~\ref{rem:unique-cube-property}). We write $x=\phi(u)$. By
  Lemma~\ref{lem:hsimplex-link}, the precubical morphism $\phi:\Lambda^u\to C$
  corresponds to a presimplicial morphism $\partial Z2\to\link(x)$, which can be
  uniquely completed into a presimplicial morphism $Z2\to\link(x)$ since the
  link is flag, which corresponds to a morphism $Y3\to C$ by
  Lemma~\ref{lem:hsimplex-link}, and therefore to a morphism $\partial Y3\to C$
  by precomposition with the inclusion $\partial Y3\into Y3$. We now show
  that~$C$ satisfies the unique $(n)$-cube property for $n\geq 2$. Suppose
  given a morphism $\phi:\partial Y(n)\to C$ and write $x=\phi(u)$ for some
  vertex~$u$ of $\partial Y(n)$. By Lemma~\ref{lem:hsimplex-link}, the
  precubical morphism $\Lambda^u\into\partial Y(n)\overset\phi\to C$
  corresponds to a presimplicial morphism $\partial Zn\to\link(x)$, which
  extends to a presimplicial morphism $Zn\to\link(x)$ because the link is flag,
  which in turn corresponds to a morphism $\psi:Y(n)\to C$ by
  Lemma~\ref{lem:hsimplex-link}. By Remark~\ref{rem:hsimplex-link-restr}, we
  have that the morphisms $\Lambda^u\into\partial Y(n)\overset\phi\to C$ and
  $\Lambda^u\into Y(n)\overset\psi\to C$ coincide, therefore the morphisms
  $\phi$ and $\partial Y(n)\into Y(n)\overset\psi\to C$ also coincide by
  Lemma~\ref{lem:higher-cube-prop-uniqueness}. Finally, the lifting is
  necessarily unique by Lemma~\ref{lem:geom-lifting}.
\end{proof}


\subsubsection{2-dimensional non-positively curved precubical sets}
Interestingly, a precubical set~$C$ satisfying the unique $n$-cube property, for
every $n\geq 3$, is characterized by its underlying 2-precubical set~$C_2$ (\ie
its image under the 2-truncation functor $\hat\pcc\to\hat\pcc_2$, see
Definition~\ref{def:truncation}). Higher-dimensional cubes can be recovered by a
``completion'' process: the $3$-cubes can be recovered from~$C_2$ by adding a
$3$-cube whenever $C_2$ contains the border of a cube (\ie for every morphism
$\partial Y3\to C_2$), then the $4$-cubes can be obtained by adding a $4$-cube
whenever the precubical set contains the border of a $4$-cube, etc. More
precisely, by general results about presheaf categories, we have

\begin{prop}
  The 2-truncation functor $\hat\pcc\to\hat\pcc_2$ admits a right adjoint
  $\hat\pcc_2\to\hat\pcc$, and this adjunction restricts to an equivalence of
  categories between the category~$\hat\pcc_2$ and the full subcategory
  of~$\hat\pcc$ whose objects are precubical sets satisfying the unique $n$-cube
  property for every $n\geq 3$.
\end{prop}

\noindent
The image of a 2-precubical set under the right adjoint will be called its
\emph{completion}. All the constructions preformed here preserve the above
equivalence. For instance, we could easily define a ``2-precubical semantics''
by adapting the definitions of Section~\ref{sec:pcs-sem} in order to associate a
2-precubical set to every program in the obvious way. It could then easily be
shown that the completion of the 2-precubical semantics is isomorphic to the
precubical semantics.

For this reason, we will mostly restrict, without loss of generality, to the underlying
2-precubical sets in the rest of the article.
In this setting, the property of being geometric can be also reformulated by
lifting properties, \ie there is a reciprocal to Lemma~\ref{lem:geom-lifting}.

\begin{defi}
  Suppose given a 2-precubical set~$C$.
  \begin{itemize}
  \item We say that a precubical set~$C$ has \emph{no looping edge} when there
    is no edge~$x\in C(1)$ with the same source and target (\ie
    $\partial^-(x)=\partial^+(x)$).
  \item We say that it has \emph{no folded square} when every square
    $x\in C(2)$ is such that $\partial^{--}(x)\neq\partial^{++}(x)$ and
    $\partial^{-+}(x)\neq\partial^{+-}(x)$.
  \item We say that is has \emph{no pinned squares} when for all squares
    $x,y\in C(2)$, either
    \begin{itemize}
    \item $\partial^{--}(x)=\partial^{--}(y)$ and
      $\partial^{++}(x)=\partial^{++}(y)$, or
    \item $\partial^{+-}(x)=\partial^{+-}(y)$ and
      $\partial^{-+}(x)=\partial^{-+}(y)$
    \end{itemize}
    implies $x=y$.
  \end{itemize}
\end{defi}

\noindent
A typical example of pinned square is given in Example~\ref{ex:pinned-squares}.

\begin{lem}
  \label{lem:2-geometric}
  A 2-precubical set~$C$ is geometric if and only if it has
  \begin{enumerate}
  \item no looping edge,
  \item no folded square,
  \item no parallel edges,
  \item no pinned squares,
  \item and the at most one square closing property.
  \end{enumerate}
\end{lem}
\begin{proof}
  Suppose given a 2-precubical set~$C$ satisfying the five above conditions.
  The first condition says precisely that~$C$ has no self-intersecting edge. We
  now show that it has no self-intersecting square. Suppose given
  square~$x\in C(2)$. Suppose that~$x$ has two 0-faces which are equal:
  $\partial^{\epsilon_0\epsilon_1}(x)=\partial^{\epsilon_0'\epsilon_1'}(x)$. The
  case $\epsilon_0\neq\epsilon_0'$ and $\epsilon_1\neq\epsilon_1'$ is excluded
  because $C$ has no folded square. Therefore we can suppose
  $\epsilon_0=\epsilon_0'$ (the case $\epsilon_1=\epsilon_1'$ is similar). If
  $\epsilon_1\neq\epsilon_1'$ then the edge $y=\partial^{\epsilon_0}_0$ is
  looping, which is forbidden, therefore $\epsilon_1=\epsilon_1'$ and the faces
  are not distinct. Suppose that~$x$ has two 1-faces which are equal:
  $\partial^{\epsilon}_i(x)=\partial^{\epsilon'}_{i'}(x)$. If $i=i'$ then
  $\epsilon=\epsilon'$, otherwise the edges $\partial^-_{1-i}(x)$ and
  $\partial^+_{1-i}(x)$ would both be looping. Suppose that $i\neq i'$, then it
  can easily be checked that this common face is a looping edge.
  The fact that~$C$ satisfies the other conditions required for being geometric
  can be checked by case analysis in a similar way, and the converse implication
  is simple to check.
\end{proof}

\noindent
To sum up, in the rest of the paper, by a non-positively curved precubical set,
we can thus mean the following:

\begin{defi}
  \label{def:2-pcs-npc}
  A \textbf{non-positively curved} precubical set is a 2-precubical set which
  satisfies
  \begin{enumerate}
  \item the five conditions of Lemma~\ref{lem:2-geometric}
  \item the cube property.
  \end{enumerate}
\end{defi}



\section{When dihomotopy coincides with homotopy}
\label{sec:homotopy-npc-cc}

\subsection{Homotopy and dihomotopy in precubical sets}
\label{sec:pcs-hom-dihom}

As exposed in the introduction, one of the main contributions of this article is
to show that homotopy coincides with its directed variant for paths in \NPC
precubical sets. We begin by formally introducing these equivalences on paths,
which can be thought of as algebraic variants of the classical notion of
homotopy between paths in algebraic topology.
Since we will need to consider both directed and non-directed paths in
precubical sets~$C$, we will consistently use the following terminology in the
remainder of the article. A \emph{directed path} (or \emph{dipath}) is a path in
the underlying directed graph of~$C$ (\ie what we have been simply calling a
``path'' up to now), and a \emph{path} is a path in the underlying non-directed
graph~$C$ (\ie a sequence of composable transitions which might contain
transitions taken backwards, that we call \emph{reversed} transitions).
We sometimes write $s:x\pathto y$ to indicate that $s$ is a path with~$x$ as
source and~$y$ as target, and $\emptypath{x}:x\pathto x$ for the empty path
on~$x$.
The \emph{length} $\length s$ of a path~$s$ is the number of (possibly reversed)
edges occurring in it. The concatenation of two paths $s:x\pathto y$ and
$t:y\pathto z$ is written $s\cc t:x\pathto z$, and the reversal of a
path~$s:x\pathto y$ is written $\ol s:y\pathto x$.
%
%
By an edge (or a transition) $a:x\to y$, we will always denote a dipath of
length one, and write $\ol a:y\to x$ for the corresponding reversed transition,
unless we explicitly state that it can be possibly reversed, in which case we
denote a path of length one.

We conservatively extend the definition relation~$\tile$
(Definition~\ref{def:independence}) from dipaths to paths: given a precubical
set~$C$, this relation is now defined as the smallest symmetric relation on
paths of length~2 such that
\begin{equation}
  \label{eq:dihomotopy}
  a\cc b\tile b'\cc a'
  \qquad\qquad
  \ol{b}\cc\ol{a}\tile \ol{a'}\cc\ol{b'}
  \qquad\qquad
  \ol{a}\cc b'\tile b\cc\ol{a'}
  \qquad\qquad
  \ol{b'}\cc a\tile a'\cc\ol{b}
\end{equation}
whenever there is a square
\begin{equation}
  \label{eq:tile'}
  \vxym{
    &x_{11}&\\
    x_{10}\ar[ur]^{b}&\tile&\ar[ul]_{a'}x_{01}\\
    &\ar[ul]^{a}x_{00}\ar[ur]_{b'}&\\
  }
\end{equation}
in~$C$.

\begin{defi}
  \label{def:dihomotopy}
  Given a precubical set~$C$, the \textbf{dihomotopy} relation~$\dihom$ on paths
  is the smallest congruence \wrt concatenation containing~$\tile$. A
  \emph{dihomotopy step} is a dihomotopy between two paths using only one of the
  four above relations~\eqref{eq:dihomotopy} in context.
  The \textbf{homotopy} relation~$\hom$ on paths is the smallest congruence
  containing~$\dihom$ and such that for every edge $a:x\to y$ we have
  \begin{equation}
    \label{eq:homotopy}
    a\cc\ol{a}\hom\emptypath{x}
    \qquad\qquad
    \ol{a}\cc a\hom\emptypath{y}
  \end{equation}
  A \emph{homotopy step} is a homotopy reduced to one of the six relations
  \eqref{eq:dihomotopy} and \eqref{eq:homotopy} in context.
\end{defi}

\noindent
For geometric precubical sets, the notion of (di)homotopy coincides with the
usual topological one through the geometric realization, see
Section~\ref{sec:hom-ax} and \cite{fajstrup2005dipaths}.

By definition of the axioms defining dihomotopy, we have the following
properties.

\begin{lem}
  \label{lem:dihomotopy-length}
  Two dihomotopic paths have the same length.
\end{lem}

\begin{lem}
  \label{lem:dihomotopy-direction}
  A path which is dihomotopic to a dipath is necessarily a dipath.
\end{lem}

\noindent
These notions allow us to introduce variants of the well-known construction of
fundamental groupoid in algebraic topology.

\begin{defi}
  \label{def:fund-cat-gpd}
  Given a precubical set~$C$, its \textbf{fundamental groupoid} $\fundgpd(C)$ is
  the category whose objects are vertices and morphisms are paths up to
  homotopy, and its \textbf{fundamental category} $\fundcat(C)$ is the category
  whose objects are vertices and morphisms are dipaths up to dihomotopy.
\end{defi}

\noindent
Since the relations~\eqref{eq:homotopy} are precisely those required for every
transition $a$ to admit $\ol a$ as inverse, it is not hard to show that

\begin{lem}
  \label{fundisfreegroupoid}
  Given a precubical set~$C$, the category $\fundgpd(C)$ is the free groupoid
  over the category~$\fundcat(C)$.
\end{lem}

By definition, two dipaths which are dihomotopic are homotopic. The main goal of
Section~\ref{sec:dihom-hom} is to show that for precubical sets satisfying the
cube property, the converse is true and both relations thus coincide
(Theorem~\ref{thm:hom-dihom}).

\begin{exa}
  \label{ex:homotopy-vs-dihomotopy}
  \label{ex:cube-without-bottom}
  Consider the standard hollow cube $\partial Y3$ without the ``bottom'' face
  $\sst{0-0}$, first considered in~\cite{fahrenberghomology}:
  \[
  \svxym{
    \sst{--+}\ar[rrr]^{0-+}\ar[dr]|{-0+}&\ar@{}[dr]|{00+}&&\ar[dl]|{+0+}\sst{+-+}\\
    &\sst{-++}\ar[r]^{0++}&\sst{+++}&\\
    \ar@{}[ur]|{-00}&\sst{-+-}\ar[u]|{-+0}\ar[r]_{0+-}\ar@{}[ur]|{0+0}&\ar[u]|{++0}\sst{++-}\ar@{}[ur]|{+00}&\\
    \sst{---}\ar[uuu]^{--0}\ar[rrr]_{0--}\ar[ur]|{-0-}&\ar@{}[ur]|{00-}&&\ar[ul]|{+0-}\ar[uuu]_{+-0}\sst{+--}
  }
  \]
  The two external paths $\sst{--0}\cc\sst{0-+}$ and $\sst{0--}\cc\sst{+-0}$ are
  homotopic (as detailed in Example~\ref{ex:cube-without-bottom-dihom}), but not
  dihomotopic (each of those two paths is dihomotopic only to itself). Notice
  that this precubical set does not satisfy the cube property.
\end{exa}

We will also see in Section~\ref{sec:geo-metric} that our algebraic conditions
for \NPC precubical sets are closely related to those characterizing
non-positively curved metric spaces in the traditional sense. Those spaces are
locally geodesic: in our context, a counterpart of this notion can be defined as
follows.

\begin{defi}
  \label{def:local-geodesic}
  A path~$s$ is \emph{locally geodesic} if it is not dihomotopic to a path of
  the form~$s'\cc a\cc\ol a\cc s''$ or $s'\cc\ol a\cc a\cc s''$. It is
  \emph{geodesic} if every path homotopic to it has a greater or equal length.
\end{defi}

\noindent
By Lemma~\ref{lem:dihomotopy-direction}, since the dihomotopy class of a dipath
contains only dipaths, the situation of previous definition can never occur:

\begin{lem}
  Every dipath is locally geodesic.
\end{lem}

\noindent
Clearly, a geodesic path is locally geodesic, but the reciprocal is not
generally true.

\begin{exa}
  Consider again the hollow cube without bottom face introduced in
  Example~\ref{ex:cube-without-bottom}. The ``loop around the hole'', \ie the
  path $\sst{--0}\cc\sst{0-+}\cc\ol{\sst{+-0}}\cc\ol{\sst{0--}}$, is locally
  geodesic. However it is not geodesic, since it is homotopic to the identity
  path on the vertex $\sst{---}$, which is strictly shorter.
\end{exa}

\subsection{A 2-categorical approach to dihomotopy in precubical sets}
\label{sec:dihom-hom}
In this section, we reuse and extend the rewriting techniques developed by
Lafont~\cite{lafont2003towards} for presenting monoidal categories, and the
subsequent developments in the second author's PhD thesis~\cite{mimram:phd,
  mimram2011structure}, in order to show one of the main results of this
article: homotopy and dihomotopy relations coincide for directed paths in \NPC
precubical sets (Theorem~\ref{thm:hom-dihom}). This will be done by rewriting
homotopies between paths into a canonical form, which is always a dihomotopy
when the two paths are directed.

\subsubsection{The fundamental 2-category and 2-groupoid of a precubical set}
\label{sec:2-cat-gpd}
In the following, we suppose fixed a 2-precubical set~$C$ which
satisfies:
\begin{enumerate}
\item the at most one square closing property,
\item the cube property,
\item for any square $\ol{a}\cc a'\tile b\cc\ol{b'}$ as in~\eqref{eq:tile'}, we
  have $a=a'$ if and only if $b=b'$.
\end{enumerate}
In particular any \NPC precubical set (see Definition~\ref{def:2-pcs-npc})
satisfies those conditions, which is our main object of interest here, but the
results of this section hold in this more general context.
To such a precubical set, we can canonically associate two distinct 2-categories
respectively defined as follows, which are 2-dimensional refinements of the
categories introduced in Definition~\ref{def:fund-cat-gpd}. The precise way in
which those are generalizations is given in Lemma~\ref{lem:fund-cat-gpd-quot}.

In classical algebraic topology, the fundamental groupoid $\fundgpd(X)$ of a
topological space is constructed as the category with, as objects, the points of
$X$, and as morphisms, the 1-tracks, or paths modulo homotopy. A direct
generalization is that of the fundamental bicategory, where points are still
0-cells, paths are 1-cells and 2-tracks or homotopies relative end points modulo
higher homotopies are 2-cells~\cite{Bigroupoid}. Composition in dimension 2, \ie
vertical composition is associative and even groupoidal (every 2-cell has an
inverse), but composition in dimension 1 is a little less well-behaved:
composition is only associative up to some 2-cells. Hence the natural structure
of \emph{weak} 2-category, or bicategory. In dimension one
also, the fundamental bicategory has all its 1-cells being invertible (up to
coherent 2-cells again), hence is a bigroupoid.
For Hausdorff spaces, it is possible to construct directly a fundamental
2-groupoid (hence a strict version)~\cite{2groupoid}, which is biequivalent to
the fundamental bigroupoid.  We follow an equivalent approach here, in defining
a fundamental 2-category and fundamental 2-groupoid, but based on a
combinatorial structure, instead of a topological structure.

\begin{defi}
\label{fund2cat}
  The \textbf{fundamental 2-category} $\fundtcat(C)$ associated to~$C$ is the
  2-category whose
  \begin{itemize}
  \item 0-cells are the vertices of $C$,
  \item 1-cells are generated by (non-reversed) edges of $C$: 1-cells are
    dipaths in $C$ and composition is given by concatenation,
  \item 2-cells are freely generated by
    \[
    \gammag{a,b}{b',a'}
    \qcolon
    a\cc b\qTo b'\cc a'
    \]
    whenever $a,b,a',b'$ are transitions such that $a\cc b\tile b'\cc a'$, and
    quotiented by the smallest congruence (\wrt both compositions) such that
    \begin{align}
      (\gammag{b',c'}{c'',b''}\circ\id_{a''})\circ(\id_{a'}\cc\gammag{a',c}{c',a''})\circ(\gammag{a,b}{b',a'}\cc\id_c)
      &
      \qeq(\id_{c''}\cc\gammag{a''',b'''}{b'',a''})\circ(\gammag{a,c'''}{c'',a'''}\cc\id_{b'''})\circ(\id_a\cc\gammag{b,c}{c''',b'''})
      \nonumber
      \\
      \label{eq:fundtcat-sym}
      \gammag{b',a'}{a,b}\circ\gammag{a,b}{b',a'}&\qeq\id_{a\cc b}
    \end{align}
    for transitions such that all the involved morphisms are
    defined. Graphically,
    \begin{align*}
    \vxym{
      &\ar[r]^b\ar@{}[d]|\tile&\ar[dr]^c\ar@{}[dd]|\tile\\
      \ar[ur]^a\ar[dr]_{c''}\ar[r]|{b'}&\ar[ur]|{a'}\ar[dr]|{c'}&&\\
      &\ar[r]_{b''}\ar@{}[u]|\tile&\ar[ur]_{a''}&
    }
    &\qeq
    \vxym{
      &\ar[r]^b\ar[dr]|{c'''}\ar@{}[dd]|\tile&\ar[dr]^c\ar@{}[d]|\tile\\
      \ar[ur]^a\ar[dr]_{c''}&&\ar[r]|{b'''}&\\
      &\ar[r]_{b''}\ar[ur]|{a'''}&\ar[ur]_{a''}\ar@{}[u]|\tile&
    }
    \\
    \vxym{
      &\ar[dr]^b\ar@{}[d]|\tile\\
      \ar[ur]^a\ar[r]|{b'}\ar[dr]_{a}&\ar[r]|{a'}&\\
      &\ar[ur]_b\ar@{}[u]|\tile
    }
    &\qeq
    \vxym{
      \ar[r]^{a}&\ar[r]^{b}&\\
    }
    \end{align*}
  \end{itemize}
  The horizontal composition will be denoted as concatenation ($\cc$), in
  sequential order, whereas the vertical composition will be denoted as usual
  ($\circ$), in categorical order, thus following the usual convention for
  monoidal categories.
\end{defi}

\begin{rem}
  \label{rem:fundtcat-gpd}
  In the 2-category $\fundtcat(C)$, if $\gammag{a,b}{b',a'}$ is defined, then
  the transitions $a'$ and $b'$ are uniquely determined from~$a$ and~$b$
  because~$C$ satisfies the at most one square closing property. Moreover, if
  $\gammag{a,b}{b',a'}$ is defined then~$\gammag{a',b'}{a,b}$ is also defined,
  and is an inverse for $\gammag{a,b}{b',a'}$ by~\eqref{eq:fundtcat-sym}: in
  fact, the 2-category is a category enriched in groupoids.
\end{rem}

\begin{lem}
  \label{lem:dihom-2-cell}
  Two dipaths~$f,g:x\pathto y$ in~$C$ are dihomotopic if and only if there
  exists a 2-cell $\alpha:f\To g$ in $\fundtcat(C)$.
\end{lem}

\begin{defi}
  \label{def:fundtgpd}
%
%
  The \textbf{fundamental 2-groupoid} $\fundtgpd(C)$ associated to~$C$ is the
  2-category whose
  \begin{itemize}
  \item 0-cells are the vertices of $C$,
  \item 1-cells are generated by edges of~$C$ or their reverse: 1-cells are
    paths in $C$ and composition is given by concatenation,
  \item 2-cells are freely generated by
    \[
    \gammag{a,b}{b',a'}
    \qcolon
    a\cc b\qTo b'\cc a'
    \]
    whenever $a,b,a',b'$ are possibly reversed transitions such that
    $a\cc b\tile b'\cc a'$, and
    \[
    \begin{array}{r@{\qquad:\qquad}r@{\qquad\To\qquad}l}
      \eta_a
      &
      \id_x&a\cc\ol a
      \\
      \varepsilon_a
      &
      \ol a\cc a&\id_y
    \end{array}
    \]
    whenever $a:x\to y$ is a possibly reversed transition, and quotiented by the
    smallest congruence (\wrt both compositions) such that
    \begin{align}
      \label{eq:fundtgpd-yb}
      (\gammag{b',c'}{c'',b''}\circ\id_{a''})\circ(\id_{a'}\cc\gammag{a',c}{c',a''})\circ(\gammag{a,b}{b',a'}\cc\id_c)
      &
      \qeq(\id_{c''}\cc\gammag{a''',b'''}{b'',a''})\circ(\gammag{a,c'''}{c'',a'''}\cc\id_{b'''})\circ(\id_a\cc\gammag{b,c}{c''',b'''})
      \\
      \label{eq:fundtgpd-sym}
      \gammag{b',a'}{a,b}\circ\gammag{a,b}{b',a'}&\qeq\id_{a\cc b}
      \\
      \label{eq:fundtgpd-zigzag1}
      (\id_a\cc\varepsilon_a)\circ(\eta_a\cc\id_a)&\qeq\id_a
      \\
      \label{eq:fundtgpd-zigzag2}
      (\varepsilon_a\cc\id_{\ol a})\circ(\id_{\ol a}\cc\eta_a)&\qeq\id_{\ol a}
      \\
      \label{eq:fundtgpd-eps-eta-inv}
      \varepsilon_{\ol a}\circ\eta_a&\qeq\id_{\emptypath{}}
      \\
      \gammag{a,\ol a}{\ol{a'},a'}\circ\eta_a&\qeq\eta_{\ol{a'}}
      \\
      \varepsilon_{a'}\circ\gammag{a,\ol a}{\ol a',a'}&\qeq\varepsilon_{\ol a}
      \\
      (\gammag{a,b'}{b,a'}\cc\id_{\ol{a'}})\cc(\id_a\cc\gammag{\ol a,b}{b',\ol{a'}})\circ(\eta_a\cc\id_b)
      &\qeq
      \id_b\cc\eta_{a'}
      \\
      (\id_{a'}\cc\gammag{b',\ol a}{\ol{a'},b})\circ(\gammag{b,a}{a',b'}\cc\id_{\ol a})\circ(\id_b\cc\eta_a)
      &\qeq
      \eta_{a'}\cc\id_b
      \\
      (\varepsilon_{a'}\cc\id_b)\circ(\id_{\ol{a'}}\cc\gammag{b',a}{a',b})\circ(\gammag{b,\ol a}\cc\id_a)
      &\qeq
      \id_b\cc\varepsilon_a
      \\
      (\id_b\cc\varepsilon_{a'})\circ(\gammag{\ol a,b'}{b,\ol{a'}}\cc\id_{a'})\circ(\id_{\ol a}\cc\gammag{a,b}{b',a'})
      &\qeq
      \varepsilon_a\cc\id_b
    \end{align}
    for possibly reversed transitions $a,a',a'',b,b',b'',c,c',c''$ such that all
    the involved morphisms are defined.
  \end{itemize}
\end{defi}

\noindent
In string diagrammatic notation, the generators are drawn as
\[
\strid1{gamma}
\qquad\qquad
\strid1{eta}
\qquad\qquad
\strid1{eps}
\]
or often simply as
\[
\strid1{gamma_simple}
\qquad\qquad
\strid1{eta_simple}
\qquad\qquad
\strid1{eps_simple}
\]
and the relations of Definition~\ref{def:fundtgpd} can be pictured as
\begin{align*}
  \strid1{yb_l}&\qeq\strid1{yb_r}&\strid1{sym_l}&\qeq\strid1{sym_r}\\
  \strid1{zigzag1_l}&\qeq\strid1{zigzag1_r}&\strid1{zigzag2_l}&\qeq\strid1{zigzag2_r}\\
  \strid1{unidim}&\qeq\\
  \strid1{etasym_l}&\qeq\strid1{etasym_r}&\strid1{epssym_l}&\qeq\strid1{epssym_r}\\
  \strid1{etanat1_l}&\qeq\strid1{etanat1_r}&\strid1{etanat2_l}&\qeq\strid1{etanat2_r}\\
  \strid1{epsnat1_l}&\qeq\strid1{epsnat1_r}&\strid1{epsnat2_l}&\qeq\strid1{epsnat2_r}\\
\end{align*}


\noindent
By definition of the generating 2-cells, we have that

\begin{lem}
  \label{lem:hom-2-cell}
  Two paths~$f,g:x\pathto y$ in~$C$ are homotopic if and only if there exists a
  2-cell $\alpha:f\To g$ in $\fundtgpd(C)$.
\end{lem}

\begin{exa}
  \label{ex:cube-without-bottom-dihom}
  In the hollow cube without bottom~$C$ of Example~\ref{ex:cube-without-bottom},
  the fact that the two paths $f=\sst{--0}\cc\sst{0-+}$ and
  $g=\sst{0--}\cc\sst{+-0}$ are homotopic is witnessed by the following 2-cell
  $\phi:f\To g$ in $\fundtgpd(C)$ (drawn from left to right instead of top to
  bottom for space constraints):
  \[
  \phi\qeq
  \strid1{hollow-cube-hom}
  \]
  Of course, this precubical set does not satisfy our assumptions, but we could
  still define a fundamental 2-category, excepting that some members of the
  equations may not be defined. We allow ourselves to consider it in this
  example only, for illustrative purposes.
\end{exa}

\begin{prop}
  \label{prop:derivable-rel}
  The following relations are derivable for 2-cells in $\fundtgpd(C)$:
  \begin{align}
    \label{eq:rel-der1}
    (\gammag{b,a'}{a,b'}\cc\id_{\ol{a'}})\circ(\id_b\cc\eta_{a'})&\qeq(\id_a\cc\gammag{\ol a,b}{b',\ol{a'}})\circ(\eta_a\cc\id_b)\\
    \label{eq:rel-der2}
    (\id_{b'}\cc\varepsilon_{a'})\circ(\gammag{\ol a,b}{b',\ol{a'}}\cc\id_{a'})&\qeq(\varepsilon_a\cc\id_{b'})\circ(\id_{\ol a}\cc\gammag{b,a'}{a,b'})\\
    \label{eq:rel-der3}
    (\varepsilon_{a'}\cc\id_a)\circ(\id_{\ol{a'}}\cc\gammag{a',a}{a',a})\circ(\eta_{\ol{a'}}\cc\id_a)&\qeq\id_a
  \end{align}
  (when the involved morphisms are defined). Graphically,
  \begin{align*}
    \strid1{etaturn_l}&\qeq\strid1{etaturn_r}&
    \strid1{epsturn_l}&\qeq\strid1{epsturn_r}\\
    \strid1{etasymeps_l}&\qeq\strid1{etasymeps_r}
  \end{align*}
\end{prop}
\begin{proof}
  We provide graphical derivations for the first and third relations:
  \[
  \begin{array}{c}
    \strie1{.5}{etaturn1}\qeq\strie1{.5}{etaturn2}\qeq\strie1{.5}{etaturn3}\\
    \strie1{.5}{etasymeps1}\qeq\strie1{.5}{etasymeps2}\qeq\strie1{.5}{etasymeps3}\qeq\strie1{.5}{etasymeps4}
  \end{array}
  \]
  It can easily be checked that the intermediate morphisms are always defined
  when the first and last one are. For instance, the equational derivation of
  the first relation is
  \[
  (\gammag{b,a'}{a,b'}\cc\id_{\ol{a'}})\circ(\id_b\cc\eta_{a'})
  =
  (\id_a\cc\gammag{\ol a,b}{b',\ol{a'}})\circ(\id_a\cc\gammag{b',\ol{a'}}{\ol a,b})\circ(\gammag{b,a'}{a,b'}\cc\id_{\ol{a'}})\circ(\id_b\cc\eta_{a'})
  =
  (\id_a\cc\gammag{\ol a,b}{b',\ol{a'}})\circ(\eta_a\cc\id_b)
  \]
  Since
  $\gammag{b,a'}{a,b'}$ is supposed to be defined, we have
  $b\cc a'\tile a\cc b'$ in~$C$ and therefore the morphisms
  $\gammag{b',\ol{a'}}{\ol a,b}$ and $\gammag{\ol a, b}{b',a'}$ are also
  defined. The relation \eqref{eq:rel-der2} is similar to \eqref{eq:rel-der1}.
\end{proof}

\begin{rem}
  \label{rem:fundtcat-fundtgpd}
  The generators (for 0-, 1- and 2-cells) of $\fundtcat(C)$ are a subset of the
  generators of $\fundtgpd(C)$ and similarly for relations. There is thus a
  canonical functor $\fundtcat(C)\to\fundtgpd(C)$ exhibiting a quotient of
  $\fundtcat(C)$ as a subcategory of $\fundtgpd(C)$. In the following we study
  some of the properties of this functor.
\end{rem}

The interest of those 2-categories is that they respectively contain ``more
information'' than the fundamental category and groupoid, in the following
sense. Given a 2-category~$\C$, we write~$\widetilde\C$ for the category with
the 0-cells of $\C$ as objects, and morphisms are the 1-cells of~$\C$ quotiented
by the smallest equivalence relation identifying two 1-cells between which there
exists a 2-cell.

\begin{lem}
  \label{lem:fund-cat-gpd-quot}
  We have $\fundcat(C)\cong\widetilde{\fundtcat(C)}$ and
  $\fundgpd(C)\cong\widetilde{\fundtgpd(C)}$.
\end{lem}

\subsubsection{Rewriting homotopies as dihomotopies}
\label{sec:hom-to-dihom}
By definition, the 2-cells of $\fundtgpd(C)$ are equivalence classes of 2-cells
in the free 2-category generated by the 2-cells
$\gammag{a,b}{b',a'}$, $\eta_a$ and $\varepsilon_a$, quotiented by the
equivalence relation~$\equiv$ generated by the relations. An element of such an
equivalence class is called a \emph{formal 2-cell}. We say that a formal 2-cell
$\phi$ \emph{rewrites} to a 2-cell $\psi$, which we write $\phi\TO\psi$, when
$\psi$ can be obtained from~$\phi$ by iteratively replacing the left member of a
relation in some context by the right member of the relation in the same
context, where the relation is one of the eleven relations of
Definition~\ref{def:fundtgpd} or the three relations of
Proposition~\ref{prop:derivable-rel}. For instance the formal 2-cell on the left
rewrites to the formal 2-cell on the right using the relation
\eqref{eq:fundtgpd-sym}:
\[
\strie1{.7}{etaturn2}
\qquad\TO\qquad
\strie1{.7}{etaturn1}
\]
More formally, we can introduce a 3-category with the above free 2-category as
underlying 2-category, and 3-cells generated by the relations oriented from left
to right, and we would have $\phi\TO\psi$ precisely when there exists a 3-cell
from $\phi$ to $\psi$ in this 3-category. The interested reader can find a more
detailed presentation of this construction in~\cite{t3rt}. Notice that if $\phi$
and~$\psi$ are two formal 2-cells such that $\phi\TO\psi$ then $\phi$ and~$\psi$
are in the same equivalence class, moreover the following lemma ensures that
rewriting a formal 2-cell will always produce a (well-defined) formal 2-cell:

\begin{lem}
  \label{lem:rel-well-defined}
  In Definition~\ref{def:fundtgpd} and Proposition~\ref{prop:derivable-rel}, if
  the left member of a relation is well-defined then the right member is also
  well-defined.
\end{lem}
\begin{proof}
  We show that when the left member of a relation is well-defined then the right
  member is also well-defined. The property is verified for
  \eqref{eq:fundtgpd-yb} because~$C$ is supposed to satisfy the cube property,
  and the right members of other relations of Definition~\ref{def:fundtgpd} are
  always well-defined. For~\eqref{eq:rel-der1}, if the left-member is
  well-defined then $\gammag{b,a'}{a,b'}$ is, \ie $b\cc a'\tile a\cc b'$ in~$C$,
  and therefore $\gammag{\ol a,b}{b',\ol{a'}}$ is well-defined. The case
  of~\eqref{eq:rel-der2} is similar, and the right-member of~\eqref{eq:rel-der3}
  is always defined.
\end{proof}

\noindent
We call a \emph{slice} a formal 2-cell $\phi$ of the form
$\phi=\id_f\cc\alpha\cc\id_g$ where $f,g$ are 1-cells and $\alpha$ is a
generating 2-cell. A slice is thus a 2-cell constituted of a unique generator in
identity context:
\[
\strid1{slice}
\]
Using the laws of 2-categories it is easy to show the following lemma, which
will enable us to reason by induction on formal 2-cells:

\begin{lem}
  \label{lem:slices}
  Any formal 2-cell~$\phi$ can be expressed as a composite of slices: there
  exists slices $\phi_1,\ldots,\phi_n$ such that
  $\phi=\phi_n\circ\ldots\circ\phi_1$. Moreover, any two such decompositions
  have the same number of slices.
\end{lem}

\begin{defi}
  The number $n$ of slices in the decomposition of a formal 2-cell $\phi$ is
  called its \emph{length} and is denoted~$\length{\phi}$.
\end{defi}

\noindent
We first generalize the notation for the 2-cells $\gammag{a,b}{b',a'}$ as
follows:

\begin{defi}
  \label{def:gen-gamma}
  Given transitions $a,a'$ and 1-cells $f,f'$, we write
  $\gammag{a,f}{f',a'}:a\cc f\To f'\cc a'$ for the 2-cell defined by induction
  on the length of~$f$ by
  \begin{itemize}
  \item if $a:x\to y$ then
    \[
    \gammag{a,\emptypath{y}}{\emptypath{x},a}\qeq\id_a
    \]
  \item given transitions $a,b,a',b'$ such that $a\cc b\tile b'\cc a'$, we have
    \[
    \gammag{a,b\cc f}{b'\cc f',a''}
    \qeq
    (\gammag{a,b}{b'a'}\cc\id_f)\circ(\id_{b'}\cc\gammag{a',f}{f',a''})
    \]
    whenever $\gammag{a',f}{f',a''}$ is defined.
  \end{itemize}
  Graphically, we have
  \[
  \gammag{a,b_1\cc b_2\cc\ldots\cc b_n}{b_1'\cc b_2'\cc\ldots\cc b_n',a'}
  \qeq
  \strid1{gamman}
  \]
\end{defi}


\noindent
As explained in the beginning of the section, we are going to show that every
formal 2-cell can be rewritten into one of a particular form, called a canonical
form. For those, we will be able to show that they are actually homotopies when
the source and target are directed paths (Lemma~\ref{lem:cf-dihom}) and
conclude.

\begin{defi}
  \label{def:cf}
  A formal 2-cell~$\phi$ is a \emph{canonical form} when it is of the form
  $\id_{\id_x}$ for some 0-cell~$x$, what we write $Z_x$, or there exists a
  canonical form $\psi$ such that $\phi$ is of one of the four following forms:
  \begin{align*}
    G^{a,f,g}_{f',a'}\psi
    &\qeq
    (\gammag{a,f}{f',a'}\cc\id_g)\circ(\id_a\cc\psi)
    \\
    H^{f,a,g,h}_{g',a'}\psi
    &\qeq
    (\id_{f\cc a}\cc\gammag{\ol a,g}{g',\ol{a'}}\cc\id_h)\circ(\id_f\cc\eta_a\cc\id_{g\cc h})\circ\psi
    \\
    E^{a,f}\psi
    &\qeq
    (\varepsilon_a\cc\id_f)\circ(\id_{\ol a}\cc\psi)
  \end{align*}
  for some transitions $a$, $a'$ and morphisms $f$, $f'$, $g$, $g'$, $h$.
  Graphically,
  \begin{align*}
    Z_x&\qeq&H^{f,a,g,h}_{g',a'}\psi&\qeq\strid1{H}\\
    G^{a,f,g}_{f',a'}\psi&\qeq\strid1{G}&E^{a,f}\psi&\qeq\strid1{E}
  \end{align*}
  The operations~$G$, $H$ and $E$ on 2-cells are called \emph{operators}.
  Notice that the operator $E^{a,f}$ transforms a 2-cell $\psi:f\To\ol a\cc g$ into
  a 2-cell $E^{a,f}\psi:a\cc f\To g$. Similarly, the ``type'' for other operators is
  \[
  \begin{array}{r@{\qquad:\qquad}r@{\qTo}l@{\qquad\rightsquigarrow\qquad}r@{\qTo}l}
    G^{a,f,g}_{f',a'}&h&f\cc g&a\cc h&f'\cc a'\cc g\\
    H^{f,a,g,h}_{g',a'}&i&f\cc g\cc h&i&f\cc a\cc g'\cc\ol{a'}\cc h\\
    E^{a,f}&g&\ol a\cc f&a\cc g&f
  \end{array}
  \]
  We sometimes use the notation~$I^{a,f}$ for the operator $G^{a,\id,f}_{\id,a,f}$:
  \[
  I^{a,f}\psi
  \qeq
  G^{a,\id,f}_{\id,a,f}\psi
  \qeq
  \id_a\cc\psi
  \qeq
  \strid1{I}
  \]
  and simply write $I^a$ when $f$ is clear from the context.
\end{defi}

\begin{prop}
  \label{prop:cf}
  Every formal 2-cell~$\phi$ rewrites to a canonical form.
\end{prop}
\begin{proof}
  By induction on the length of~$\phi$. If $\length{\phi}=0$, then
  $\phi=\id_f=I^{a_1}\ccs I^{a_n}Z$ for some 1-cell~$f=a_1\ccs a_n$ and we
  conclude. Otherwise $\length{\phi}=n+1$, and by Lemma~\ref{lem:slices}, $\phi$
  admits a decomposition as $\phi=\sigma\circ\psi$ with $\length{\psi}=n$ and
  $\sigma$ is a slice, \ie $\sigma=\id_f\cc\alpha\cc\id_g$ for some generating
  2-cell $\alpha$ and 1-cells $f$ and $g$. By induction hypothesis, $\psi$ is
  equal to a canonical form, and we conclude by examining all the possible forms
  for the canonical form~$\psi$ and the slice~$\sigma$.
  \begin{itemize}
  \item Suppose that $\alpha=\eta_a$. Then
    $\phi=(\id_f\cc\eta_a\cc\id_g)\circ\psi=H^{f,a,\id,g}_{\id,a}\psi$.
  \item Suppose that $\alpha=\varepsilon_b$. We proceed by induction on $\psi$.
    \begin{itemize}
    \item Suppose that $\psi=G^{a,f'',g'}_{f',a'}\psi'$:
      \[
      \psi=\strid1{Gproof}
      \]
      Depending on~$f$ and~$g$, the following cases are possible.
      \begin{itemize}
      \item If $f'=f'_1\cc\ol b\cc b\cc f'_2$, $f=f'_1$ and $g=f'_2\cc a'\cc g'$,
        \[
        \phi
        =
        \strie1{.7}{cf-proof1a}
        =
        \strie1{.7}{cf-proof1b}
        \TO
        \strie1{.7}{cf-proof1c}
        =
        G\psi''
        \]
        where $\psi''$ is a canonical form obtained by induction.
      \item If $f'=f'_1\cc\ol{a'}$, $b=a'$, $f=f'_1$ and $g=g'$,
        \[
        \phi
        =
        \strie1{.7}{cf-proof2a}
        \TO
        \strie1{.7}{cf-proof2b}
        \TO
        \strie1{.7}{cf-proof2c}
        \TO
        \strie1{.7}{cf-proof2d}
        =
        E\psi''
        \]
        where $\psi''$ is a canonical form obtained by induction.
      \item If $g'=\ol{a'}\cc g'_2$, $b=\ol{a'}$, $f=f'$, $g=g'_2$,
        \[
        \phi
        =
        \strie1{.7}{cf-proof3a}
        \TO
        \strie1{.7}{cf-proof3b}
        \TO
        \strie1{.7}{cf-proof3c}
        =
        E\psi''
        \]
        where $\psi''$ is a canonical form obtained by induction.
      \item If $g'=g'_1\cc\ol b\cc b\cc g'_2$
        \[
        \phi
        =
        \strie1{.7}{cf-proof4a}
        =
        \strie1{.7}{cf-proof4b}
        \TO
        \strie1{.7}{cf-proof4c}
        =
        G\psi''
        \]
        where $\psi''$ is a canonical form obtained by induction.
      \end{itemize}
    \item The cases where $\psi=H^{f,a,g,h}_{g',a'}\psi'$ and
      $\psi=E^{a,f}\psi'$ can be handled similarly by case analysis.
    \end{itemize}
  \item The case where $\alpha=\gammag{a,b}{b',a'}$ can be handled similarly by
    case analysis.\qedhere
  \end{itemize}
\end{proof}


\noindent
A canonical form is not, in general, a normal form. The following lemma shows
that a canonical form of the form $XEHY$ (where $X$ and $Y$ are composites of
operators and we omit indices) can always be rewritten to a canonical form of
the form $XGY$ of $XHEY$, and $XGHY$ to $XHGY$. By severely abusing notations,
this can be summarized as
\[
EH\qTO G\text{ or }HE
\qquad\qquad\qquad
GH\qTO HG
\]
From these relations, it is easy to show that if a canonical form contains a $H$
(\ie is of the form $XHY$) then it can be rewritten to one which contains a $H$
in first position (\ie of the form $HX$) and if it contains an $E$ (\ie is of
the form $XEY$) then it can be rewritten to one of the form $XEY$ where $Y$ does
not contain an~$H$. This will be used in Lemma~\ref{lem:cf-dihom}.

\newcommand{\fakeitem}[1]{\textnormal{\hspace{3.5ex}(#1)}}
\begin{lem}
  \label{lem:reduced-cf}
  The following rewriting relations can be shown:
  \begin{align*}
    \intertext{\fakeitem1 for every morphism $\phi:i\to g\cc h$,}
    E^{a,g'\cc\ol{a'}\cc h}H^{\id,a,g,h}_{g',a'}\phi
    &\qTO
    G^{a,g,h}_{g',a'}\phi
    \\
    \strid1{EH1}&\qTO\strid1{EH2}
    \\
    \intertext{\fakeitem2 for every morphism $\phi:i\to b\cc f\cc g\cc h$,}
    E^{b,f\cc a\cc g'\cc\ol{a'}\cc h}H^{b\cc f,a,g,h}_{g',a'}\phi
    &\qTO
    H^{f,a,g,h}_{g',a'}E^{b,f\cc g\cc h}\phi
    \\
    \strid1{EHcom1}&\qeq\strid1{EHcom2}
    \\
    \intertext{\fakeitem3 for every morphism $\phi:i\To f_1\cc f_2\cc g\cc h$,}
    G^{b,f_1,f_2\cc a'\cc g'\cc\ol{a'}\cc h}_{f_1',b'}H^{f_1\cc f_2,a,g,h}_{g',a'}\phi
    &\qTO
    H^{f_1'\cc b'\cc f_2,a,g,h}_{g',a'}G^{b,f_1,f_2\cc g\cc h}_{f_1',b'}\phi
    \\
    \strid1{GHcom1}&\qeq\strid1{GHcom2}
    \\
    \intertext{\fakeitem4 for every morphism $\phi:i\To f\cc g_1\cc g_2\cc h$}
    G^{b,f\cc a\cc g_1'',g_2'\cc\ol{a''}\cc h}_{f'\cc a'\cc g_1',b'}H^{b\cc f,a,g_1\cc g_2,h}_{g_1''\cc g_2',a''}\phi
    &\qTO
    H^{f',a',g_2''\cc b''\cc g_2,h}_{g_1'\cc b'\cc g_2',a''}G^{b,f\cc g_1,g_2\cc h}_{f'\cc g_2'',b''}\phi
    \\
    \strid1{GHmid1}&\qTO\strid1{GHmid2}
    \\
    \intertext{\fakeitem5 for every morphism $\phi:i\To f\cc g\cc h_1\cc h_2$}
    G^{b,f\cc a\cc g''\cc\ol{a'''}\cc h_1,h_2}_{f'\cc a'\cc g'\cc\ol{a''}\cc h_1',b'}H^{b\cc f,a,g,h_1\cc h_2}_{g'',a'''}\phi
    &\qTO
    H^{f',a',g''',h_1'\cc b'\cc h_2}_{g',a''}G^{b,f\cc g\cc h_1,h_2}_{f'\cc g'''\cc h_1',b'}\phi
    \\
    \strid1{GHr1}&\qTO\strid1{GHr2}
  \end{align*}
  (in the last two cases the omitted indices can be inferred from the figure).
\end{lem}
\begin{proof}
  (1) is a direct application of relation~\eqref{eq:fundtgpd-zigzag2}, (2) and
  (3) are equalities (they follow from the exchange law in 2-categories), (4)
  and (5) can be shown by applying suitable sequence of rewriting, which can
  easily be guessed from the figures.
\end{proof}

\noindent
In the following, we sometimes omit the superscript and subscript indices of
operators for simplicity. We say that a 1-cell~$f$ \emph{contains a reversed
  transition} when it can be written in the form $f=f_1\cc\ol a\cc f_2$ where
$\ol a$ is a reversed transition: a path~$f$ is a dipath when it does not
contain a reversed transition. Two possibly reversed transitions $a$ and $a'$
are said to have the \emph{same direction} if they are both non-reversed or both
reversed.

\begin{lem}
  \label{lem:gamma-direction}
  In a 2-cell of the form $\gammag{a,b}{b',a'}$ the transitions $a$ and $a'$
  (\resp $b$ and $b'$) have the same direction. In a 2-cell of the form
  $\gammag{a,f}{f',a'}$ the transitions $a$ and $a'$ have the same direction,
  and $f$ contains a reversed transition if and only if $f'$ does.
\end{lem}
\begin{proof}
  The first point is immediate by definition of~$\gammag{a,b}{b',a'}$. This can
  easily be used to show the result on $\gammag{a,f}{f',a'}$ by induction on its
  definition (Definition~\ref{def:gen-gamma}).
\end{proof}

\begin{lem}
  \label{lem:cf-rev-st}
  The source and target of a formal 2-cell in canonical form satisfy the
  following properties.
  \begin{enumerate}
  \item A formal 2-cell of the form $H\phi$ necessarily contains a reversed
    transition in its target.
  \item $Z$ contains no reversed transition in its source or its target.
  \item $G\phi$ contains a reversed transition in its target if and only if
    either it contains a reversed transition in its source or $\phi$ contains a
    reversed transition in its target.
  \item A formal 2-cell of the form $EX$ where is consists of operators
    within~$\set{Z,E,G}$ necessarily contains a reversed transition in its
    source.
  \end{enumerate}
\end{lem}
\begin{proof}
  The properties can be shown as follows.
  \begin{enumerate}
  \item The target of a formal 2-cell $H^{f,a,g,h}_{a',g'}$ is
    $f\cc a\cc g'\cc\ol{a'}\cc h$ and thus contains a reversed transition since
    it contains both $a$ and $\ol{a'}$, and $a$ and $a'$ have the same direction
    (this follows easily from Lemma~\ref{lem:gamma-direction}).
  \item The source (\resp target) of~$Z$ is an empty path.
  \item Consider $G^{a,f,g}_{f',a'}\phi:a\cc h\To f'\cc a'\cc g$ with
    $\phi:h\To f\cc g$, where
    $G^{a,f,g}_{f',a'}\phi=(\gammag{a,f}{f',a'}\cc\id_g)\circ(\id_a\cc\phi)$. If
    the transition $a'$ is reversed then~$a$ is also reversed and we
    conclude. Otherwise, either~$f'$ or~$g$ contains a reversed transition. If
    it the case for~$g$ then $\phi$ contains a reversed transition in its
    target. Otherwise~$f'$ contains a reversed transition and thus also $f$ by
    previous lemma. The converse implication is similar.
  \item It is enough to show the property in the case where $X$ is of the form
    $X=GG\ldots GZ$ because if~$EX$ contains a reversed transition in its source
    then so does $EX'EX$. Consider $E^{a,f}\phi:\ol a\cc g\To f$ with
    $\phi:g\to a\cc f$, if $a$ not reversed then $E^{a,f}\phi$ contains the
    reversed transition $\ol a$ in its source. Otherwise $\phi$ contains a
    reversed transition in its target and is of the form $GG\ldots GZ$. By an
    easy induction, using the two previous cases, it thus contains a reversed
    transition in its source. The converse implication is similar.\qedhere
  \end{enumerate}
\end{proof}

\begin{lem}
  \label{lem:cf-dihom}
  Suppose given a formal 2-cell $\phi:f\To g$ whose source~$f$ and target~$g$
  are dipaths (\ie do not contain reversed transitions). Then $\phi$ rewrites to
  a formal 2-cell which is a composite of generators not involving generators of
  the form~$\eta_a$ and $\varepsilon_a$.
\end{lem}
\begin{proof}
  By Proposition~\ref{prop:cf}, $\phi$ rewrites to a canonical form. This
  canonical form corresponds to a formal 2-cell involving a generator of the
  form $\eta_a$ if and only if it contains an operator~$H$. By
  Lemma~\ref{lem:reduced-cf}, up to more rewriting, if the canonical form
  contains an operator~$H$ we can always suppose that there is one placed at
  leftmost position. This situation is impossible because a morphism of the form
  $H\psi$ necessarily contains a reversed transition in its target by
  Lemma~\ref{lem:cf-rev-st}. Similarly, the formal 2-cell involves a generator
  $\varepsilon_a$ if and only if its canonical form contains an operator~$E$. Up
  to rewriting more, we can always suppose that the canonical form is of the
  form $YEX$ where~$X$ consists of operators not involving~$H$. By
  Lemma~\ref{lem:cf-rev-st}, it thus contains a reversed transition in its
  source, contradicting the hypothesis.
\end{proof}

\begin{thm}
  \label{thm:hom-dihom}
  Two dipaths in~$C$ are homotopic if and only if they are dihomotopic.
\end{thm}
\begin{proof}
  We show the left-to-right implication, the converse being obvious. Suppose
  that~$f$ and~$g$ are two homotopic dipaths. By Lemma~\ref{lem:hom-2-cell},
  this means that there exists a 2-cell $\phi:f\To g$ in $\fundtgpd(C)$. This
  2-cell $\phi$ is, by definition, an equivalence class of formal 2-cells under
  the relations defining fundamental 2-groupoids. Choose an arbitrary formal
  2-cell~$\psi$ in this equivalence class. By Lemma~\ref{lem:cf-dihom}, $\psi$
  rewrite to a formal 2-cell $\psi'$, which is still in the equivalence
  class~$\phi$ by definition of rewriting, and does not involve
  generators~$\eta$ or~$\varepsilon$. Therefore, the 2-cell~$\phi$ is in the
  image of the canonical functor $\fundtcat(C)\to\fundtgpd(C)$ (see
  Remark~\ref{rem:fundtcat-fundtgpd}), and there exists a dihomotopy from~$f$
  to~$g$ by Lemma~\ref{lem:dihom-2-cell}.
\end{proof}

\noindent
This can be formulated more categorically as follows.


\begin{thm}
  \label{thm:fundt-embedding}
  The quotient functor $\fundcat(C)\to\fundgpd(C)$ is full and faithful.
\end{thm}



\noindent
On a related note, a characterization of the embeddability of small categories
into their groupoid completion is given in~\cite{johnstone2008embedding}, and
variants of the above theorem have also been investigated in more geometric
contexts, see the references given in Section~\ref{sec:dgreal}.

\subsection{Extensions}
In this section, we briefly comment on a few topics and extensions related to
the preceding approach. We omit proofs, they should be detailed in subsequent
works.

\subsubsection{The free compact closed category on a unidimensional object}
Given an \NPC precubical set~$C$, we have seen (Lemma~\ref{lem:hom-2-cell}) that
homotopy corresponds to the existence of a 2-cell in $\fundtgpd(C)$ and
therefore the relations we put on 2-cells in the definition of $\fundtgpd(C)$ do
not really matter as long as they are well-defined
(Lemma~\ref{lem:rel-well-defined}) and allow us to rewrite a homotopy between
dipaths to a dihomotopy (Theorem~\ref{thm:hom-dihom}). In this section, we
advocate that our axioms are not completely arbitrary by showing that our
construction encompasses the free compact closed category on an unidimensional
object as a particular case.

We study here the particular case of the 2-category $\C=\fundtgpd(C)$ where~$C$
is the precubical set with one vertex~$x$, one edge $a:x\to x$ and one square
$a\cc a\tile a\cc a$:
\[
C
\qeq
\vxym{
  &x\\
  x\ar[ur]^a&\tile&\ar[ul]_ax\\
  &\ar[ul]^ax\ar[ur]_a&
}
\]
This will illustrate that our axioms for the category $\fundtcat(C)$ constitute
a variant of the well-known notion of compact-closed category. Since the
2-category~$\C$ has only one 0-cell, we can consider it as a monoidal category,
with ``$\cc$'' as tensor product. It can be noticed that the monoidal
category~$\C$ is symmetric (the symmetry being generated by
$\gammag{a,a}{a,a}:a\cc a\to a\cc a$), and compact closed (with $\ol a$ being
dual to~$a$):

\begin{defi}
  A symmetric monoidal category $(\C,\cc,I)$ is \emph{compact closed} when for
  every object $A\in\C$, there exists an object $A^*\in\C$, called the
  \emph{dual} of $A$, together with two morphisms
  \[
  \eta_A:I\to A\cc A^*
  \qquad\text{and}\qquad
  \varepsilon_A:A^*\cc A\to A
  \]
  such that
  \[
  (\id_A\cc\varepsilon_A)\circ(\eta_A\cc\id_A)=\id_A
  \qquad\qquad
  (\varepsilon_A\cc\id_{A^*})\circ(\id_{A^*}\cc\eta_A)=\id_{A^*}
  \]
  An object~$A\in\C$ in such a category is \emph{unidimensional} when
  $\varepsilon_A\circ\gamma_{A,A^*}\circ\eta_A=\id_I$.
\end{defi}

\begin{exa}
  \newcommand{\Tr}{\operatorname{Tr}}
  Consider the symmetric monoidal category $\FDVect$ of finite-dimensional
  vector spaces over a fixed field~$\k$, equipped with the usual tensor product
  $\otimes$, the unit~$I$ being~$\k$ equipped with its canonical $\k$-vector
  space structure. This category is compact closed, the dual of a vector
  space~$A$ being the space~$A^*$ of linear functionals on~$A$. Notice that
  given a linear map $f:A\to B$, the morphism $\Tr(f):I\to I$ defined by
  $\Tr(f)=\varepsilon_{A^*}\circ(f\otimes\id_{A^*})\circ\eta_A$ is the linear
  map $x\mapsto\tr(f)x$ where $\tr(f)\in\k$ is the trace of~$f$ in the usual
  sense. In particular, given a space~$A$, we have $\tr(\id_A)=\dim(A)$ and
  $\Tr(\id_A)=\id_I$ precisely when $A$ is of dimension~$1$ (we suppose
  that~$\k$ is of null characteristic).
\end{exa}

\noindent
A converse to the preceding remark can be shown, and precisely formulated as
follows. The proof of this result is not entirely obvious, and can be obtained
as a variant of~\cite{kelly1980coherence, preller2007free}.

\begin{prop}
  The category~$\C$ is the free compact closed category containing a
  unidimensional object.
\end{prop}

\subsubsection{Convergence of the rewriting system}
We conjecture that this rewriting system is convergent, \ie both terminating and
confluent. We believe that Guiraud's techniques based on
derivations~\cite{guiraud2006termination} can be used in order to show the
termination of the rewriting system. Confluence is quite tedious to check.
Because of the ``Yang-Baxter'' rule~\eqref{eq:fundtgpd-yb} there is an infinite
number of critical pairs as discovered by Lafont~\cite{lafont2003towards}. We
could still be able to handle those families of critical pairs (as done by
Lafont), however this requires beforehand to establish the existence of
canonical forms as we did in previous section: it would therefore require
strictly more work than done here. As a byproduct of this result, we expect to
be able to show that in $\fundtgpd(C)$ (and therefore also in $\fundtcat(C)$ by
Theorem~\ref{thm:fundt-embedding}) there is at most one homotopy between any two
paths (of course, when~$C$ satisfies the cube property and other hypothesis).

\subsubsection{An axiomatization of homotopy between homotopies}
\label{sec:hom-ax}
As explained in detail in next section, one can interpret a precubical set~$C$
as a directed space~$\dgreal C$ by taking its geometric realization (see
Definition~\ref{def:dgreal}). It has been shown by
Fajstrup~\cite{fajstrup2005dipaths} that two paths in~$C$ are dihomotopic (\resp
homotopic) in~$C$ if and only if the corresponding paths in~$\dgreal{C}$ are
dihomotopic (\resp homotopic) in the geometric sense, \ie our algebraic notion
of (di)homotopy corresponds to the geometric one. We conjecture that this result
extends one dimension higher, \ie that two formal 2-cells in $\fundtgpd(C)$
(\resp $\fundtcat(C)$) are equal (\ie in the same equivalence class modulo the
relations) if and only if the corresponding homotopies (\resp dihomotopies)
in~$\dgreal C$ are homotopic, thus bringing a geometric justification for our
axiomatic definition of $\fundtgpd(C)$ and $\fundtcat(C)$.
%


\section{A (geo)metric approach to the cube property}
\label{sec:geo-metric}
The cube condition which is at the heart of the developments in previous
sections is quite reminiscent of Gromov's condition for characterizing
non-positively curved cubical complexes~\cite{gromov1987hyperbolic}. In order to
make the connection precise, one has to associate to every precubical set a
geodesic metric space (which is a cubical complex). Such a construction could be
performed abstractly as a geometric realization in the category of metric spaces
if this category were cocomplete... which is not the case. This leads us to
investigate a generalization of the notion of metric spaces (called
\emph{generalized metric spaces}~\cite{grandis2009directed} or \emph{Lawvere
  metric spaces}~\cite{lawvere1973metric} or \emph{hemi-metric
  spaces}~\cite{goubault2013non}) which form a cocomplete
category. Interestingly, these spaces also allow one to encode a notion of
``time direction'' in the metric, as first noticed by
Grandis~\cite{grandis2007fundamental, grandis2009directed}, thus enabling one to
formulate a metric semantics for concurrent programs, which is not only able to
encode the direction of time but also the duration of the elapsed time during an
execution.
The use of non-positive curvature has been used by Ghrist in order to study
configuration spaces and cubical complexes~\cite{Ghrist1, ghristgeometric,
  Ghrist2}.

In Section~\ref{sec:gms}, we recall the definition of generalized metric spaces
as well as associated basic definitions and properties: the properties of the
resulting category, the symmetric variant of generalized metric spaces (which
can be thought of as being ``undirected''), the topology induced by a metric,
and the various notions of paths. In Section~\ref{sec:pcs-greal}, we define the
realization of a precubical set as a metric space and show that the underlying
topology of the realization induced by the metric is the expected one: this
means that considering the metric bring more (as opposed to different)
information compared to the usual case. Finally, in Section~\ref{sec:npc}, we
recall the original definition of space with non-positive curvature, and show
that our definition corresponds with this one, through geometric realization.

The two main results are those in the end of Sections~\ref{sec:pcs-greal} and
\ref{sec:npc} (compatibility of the metric with usual topology and
correspondence of our axioms with the usual ones). We do not claim much
originality in this section, but we did our best to collect in a concise way
results which are disseminated in the literature or, worse, considered as
folklore. The reason why we have presented this here is to explain the
inspiration for our axiomatics on precubical sets, but none of the results in
this paper depend on the geometric interpretation of those axioms.

\subsection{Generalized metric spaces}
\label{sec:gms}
Recall that a \emph{metric space} $(X,d)$ consists of a set~$X$ equipped with a
metric $d:X\times X\to[0,\infty]$, \ie a function such that, given $x,y,z\in X$,
we have
\[
\begin{array}{rrl}
  (1)&\text{point equality:}&d(x,x)=0\\
  (2)&\text{triangle inequality:}&d(x,z)\leq d(x,y)+d(y,z)\\
  (3)&\text{finite distances:}&d(x,y)<\infty\\
  (4)&\text{separation:}&\text{$d(x,y)=d(y,x)=0$ implies $x=y$}\\
  (5)&\text{symmetry:}&d(x,y)=d(y,x)\\
\end{array}
\]
Generalized metric space are defined similarly, but only keeping the two first
conditions:

\begin{defi}
  \label{def:gms}
  \label{def:gmet}
  A \textbf{generalized metric space} $(X,d)$ consists of a set~$X$ equipped
  with a function $d:X\times X\to[0,\infty]$, called generalized \emph{metric}
  or \emph{distance}, such that, given $x,y,z\in X$, we have
  \[
  \begin{array}{rrl}
    (1)&\text{point equality:}&d(x,x)=0\\
    (2)&\text{triangle inequality:}&d(x,z)\leq d(x,y)+d(y,z)\\
  \end{array}
  \]
  A morphism $f:(X,d_X)\to(Y,d_Y)$ between two such spaces is a function
  $f:X\to Y$ which is \emph{nonexpansive} (does not increase distance): for
  every $x,y\in X$, $d_Y(f(x),f(y))\leq d_X(x,y)$. We write $\GMet$ for the
  \emph{category of generalized metric spaces} and their morphisms.
\end{defi}

\noindent
By a \emph{space} we will always mean an object in this category, unless we
explicitly state that the space is ``non-generalized''. It was observed by
Lawvere~\cite{lawvere1973metric} that it can alternatively be defined as the
category of categories and functors enriched over the posetal category
$\mcV=[0,\infty]$ (with an arrow $x\to y$ whenever $x\geq y$) equipped with the
monoidal structure induced by addition.
Our main interest in considering this category instead of the full subcategory
of (usual) metric spaces is double. Firstly, removing conditions (3) and (4)
makes the category much more well-behaved (the category of generalized metric
spaces has small colimits, which is not the case for metric spaces). Of course,
we generally work with separated ($T_0$) spaces
and we will have to show that colimits we obtain from such spaces are still
separated. Secondly, removing condition (5) allows us to encode a notion of
direction of the space: intuitively, a point~$x$ is really ``before'' a
point~$y$ whenever the distance $d(x,y)$ is finite and $d(y,x)=\infty$.

\begin{exa}
  \label{ex:ui}
  \label{ex:dui}
  \label{ex:dr}
  Given $a,b\in\R$, we write $\clint ab$ for the interval equipped with the
  usual metric given by $d(x,y)=|y-x|$. We write $\directed{\clint ab}$ for the
  same interval metrized by
  \[
  d(x,y)
  \qeq
  \begin{cases}
      y-x&\text{if $y\geq x$}\\
      \infty&\text{if $y<x$}
  \end{cases}
  \]
  In particular, we often write $I$ (\resp $\dui$) instead of $\clint01$ (\resp
  $\directed{\clint01}$) for the (\emph{directed}) \emph{standard
    interval}. Similarly, we write $\R$ (\resp $\dR$) for the (\emph{directed})
  \emph{real line}, which is also sometimes called the \emph{Sorgenfrey
    line}~\cite{sorgenfrey1947topological,goubault2013non}.
\end{exa}

\begin{exa}
  \label{ex:dcircle}
  The \emph{directed unit circle} $\dcircle$ is the set of complex points of the
  form $\ce^{\ci 2\pi\theta}$, with $\theta\in\R$, equipped with the distance
  $d(x,y)=\bigwedge\setof{\rho-\theta}{x=\ce^{\ci 2\pi\theta}, y=\ce^{\ci
      2\pi\rho}, \rho\geq\theta}$.
\end{exa}

\begin{defi}
  An \emph{isometry} $f:X\to Y$ is a distance-preserving morphism, \ie satisfies
  $d_Y(f(x),f(x'))=d_X(x,x')$ for every $x,x'\in X$.
\end{defi}

\begin{rem}
  Isomorphisms in~$\GMet$ are isometries.
\end{rem}

\begin{lem}
  Suppose given an isometry $f:X\to Y$. For every $x,y\in X$, we have
  $f(x)=f(y)$ implies $d(x,y)=0=d(y,x)$. In particular, when $X$ separated $f$
  is injective, and we write $f^{-1}(y)$ for the unique antecedent of $y\in Y$
  under~$f$.
\end{lem}

\noindent
We write~$d_0$ for the constant distance equal to~$0$ and $d_\infty$ for the
distance such that $d(x,x)=0$ and~$d(x,y)=\infty$ for $x\neq y$.

\subsubsection{Limits and colimits}
\label{sec:limits}
We now show that the category of generalized metric spaces has all limits and
colimits and describe how to explicitly construct those. These are also studied
in~\cite{grandis2009directed, goubault2013non}.

\begin{prop}
  \label{prop:gmet-limits}
  The category $\GMet$ is complete and cocomplete. Moreover, the forgetful
  functor $\GMet\to\Set$ admits both a left and a right adjoint and thus
  preserves limits and colimits.
\end{prop}
\begin{proof}
  Since the category~$\mcV$ is complete and cocomplete, with products (\resp
  coproducts) being given by supremum (\resp infimum), by general results about
  enriched categories~\cite{betti1983variation}, the category
  $\GMet=\VCat\mcV$ has small limits and colimits. The left (\resp right)
  adjoint to the forgetful functor $\GMet\to\Set$ is the functor which to a
  set~$X$ associates the generalized metric space $(X,d_0)$ (\resp
  $(X,d_\infty)$). The underlying set of a limit of metric spaces is thus the
  limit of the underlying sets and similarly for colimits.
\end{proof}

Limits and colimits can be constructed explicitly as follows.
It is well-known that it is enough to construct all (co)products and
(co)equalizers to construct all (co)limits.

\begin{lem}
  \label{lem:gmet-prod}
  Given a family $(X_i,d_i)_{i\in I}$ of metric spaces, its product (\resp
  coproduct) is the set $\prod_{i\in I}X_i$ (\resp $\coprod_{i\in I} X_i$),
  respectively equipped with the distances
  \[
  d_{\prod X_i}((x_i),(y_i))=\bigvee_{i\in I}d_i(x_i,y_i)
  \qquad\qquad
  d_{\coprod X_i}(x,y)=
  \begin{cases}
    d_i(x,y)&\text{if $x,y\in x_i$}\\
    \infty&\text{otherwise}
  \end{cases}
  \]
\end{lem}


\begin{lem}
  \label{lem:gmet-eq}
  Given a pair of morphisms $f,g:(X,d_X)\to (Y,d_Y)$, their equalizer is the
  equalizer set $Z=\setof{x\in X}{f(x)=g(x)}$ equipped with the restriction of
  $d_X$ as distance.
\end{lem}

\begin{lem}
  \label{lem:gmet-coeq}
  Given a pair of morphisms $f,g:(X,d_X)\to (Y,d_Y)$, the coequalizer of~$f$
  and~$g$ is the set $Y/{\approx}$, where~$\approx$ is the smallest equivalence
  relation such that $y\approx y'$ whenever there exists $x\in X$ with $y=f(x)$
  and $y'=g(x)$, equipped with the following distance. Given two points
  $x,y\in Y$, a ``chain'' from $x$ to $y$
  is a sequence of points $x_1,y_1,x_2,y_2,\ldots,x_n,y_n\in Y$ such that
  $x=x_1$, $y_i\approx x_{i+1}$, and $y_n=y$. The length of such a chain~$u$ is
  defined to be $l(u)=\sum_{i=1}^nd(x_i,y_i)$. The distance between two points
  $x,y\in Y/{\approx}$ is the infimum of the lengths of chains from a
  representative of~$x$ to a representative of~$y$ (in fact, it does not depend
  on the choice of representatives for~$x$ and~$y$).
\end{lem}

\begin{rem}
  \label{rem:coeq-d}
  With the notations of previous lemma, given $x,y\in Y$, the sequence $x,y$ is
  a chain from~$x$ to~$y$ and thus $d_{Y/{\approx}}(x,y)\leq d_Y(x,y)$.
\end{rem}

\begin{exa}
  The directed circle (see Example~\ref{ex:dcircle}) can be constructed
  from~$\directed{\clint01}$ as a colimit, by identifying the points~$0$
  and~$1$.
\end{exa}

\begin{exa}
  \label{ex:In}
  Consider the spaces $(I_n,d_n)$ indexed by $n\in\N$ with $I_n=[0,1]$ and
  distance $d_n(x,y)=|y-x|/n$. The colimit $I_\omega=\coprod_{n\in\N}/{\approx}$
  where $\approx$ identifies points $0$, \resp $1$, in various $I_n$ is the
  colimit of the sets equipped with the distance~$d_\omega$ such that given
  $x\in I_n$ and $y\in I_m$,
  \[
  d_\omega(x,y)\qeq
  \begin{cases}
    0&\text{if $x,y\in\set{0,1}$}\\
    |y-x|/n&\text{if $n=m$}\\
    (x/n+y/m)\vee((1-x)/n+(1-y)/m)&\text{otherwise}
  \end{cases}
  \]
  In particular, we have $d_\omega(0,1)=0$, which shows that coequalizers of
  separated spaces are not necessarily separated. Notice that, in the space
  $I_\omega$, if we ``cut in the middle'' all the intervals we obtain two
  ``star-shaped'' spaces which are both separated. The space $I_\omega$ can then
  be obtained as a pushout of the two spaces (over the discrete space with $\N$
  as points), showing that separated spaces are not either closed under
  pushouts.
\end{exa}

\begin{rem}
  Since the category~$\GMet$ is cocomplete, we can easily mimic the definitions
  of Section~\ref{sec:pcs-sem} in order to associate a generalized metric space
  to each program. We do not detail this here, because we will be able to reuse
  the precubical semantics and realize it as a space, see
  Section~\ref{sec:pcs-greal}, instead of starting all over again. In
  particular, an action in this semantics will typically be realized as a
  directed unit interval~$\dui$ (see Example~\ref{ex:dui}): this corresponds to
  the intuition that executing an action is a directed process and takes «~one
  unit of time~» (but of course, other duration choices could be made depending
  on the nature of the actions, as illustrated in next remark). Moreover,
  parallel composition is naturally interpreted as product of metric spaces: the
  time taken to execute two programs in parallel is the maximum of both
  execution times.
\end{rem}

\begin{rem}
  In order to illustrate why removing the separation axiom makes sense in this
  context, consider the action $\nop$, which is instantaneously executed and
  does nothing. This particular action is most naturally modeled as $(I,d)$,
  with $d(x,y)=0$ for $x\leq y$ and $d(y,x)=\infty$ otherwise (but other actions
  are still modeled as $\dui$). For instance, the geometric semantics of a
  program of the form $\ploop{A}$, where $A$ is some action, would be a space of
  the form
  \[
  \begin{tikzpicture}
    \draw (0,0) circle (.5);
    \gvert{(180:.5)} node[left] {$x$};
    \gvert{(0:.5)} node[right] {$y$};
    \draw (90:.5) node[above] {$\nop$};
    \draw (-90:.5) node[below] {$A$};
    \gvert{(0:.5)};
  \end{tikzpicture}
  \]
  In this space, we have $d(x,y)=0$ and $d(y,x)=1$; it is thus neither separated
  nor symmetric. Similarly, $\ploop\nop$ provides an example of a non-trivial
  space (non-contractible in particular) equipped with the constant null
  distance~$d_0$.
\end{rem}

Because our spaces are not supposed to be separated, it will sometimes be useful
to consider the following kind of quotient.

\begin{defi}
  Given a space~$X$, an equivalence relation~$\approx$ on~$X$ is called
  \emph{instantaneous} when $x\approx y$ implies $d(x,y)=0$ for every $x,y\in X$.
\end{defi}

\begin{lem}
  \label{lem:d0-quotient}
  Given a space~$X$ and a instantaneous equivalence relation~$\approx$, the
  quotient morphism $X\to X/{\approx}$ is an isometry.
\end{lem}
\begin{proof}
  The quotient space~$X/{\approx}$ can be computed as the coequalizer of
  ${\approx}\rightrightarrows X$, where the two arrows are the projections of
  the relation ${\approx}\subseteq X\times X$. Given $x,y\in X$ and a chain
  $x_0,y_0,\ldots,x_n,y_n$ with $x_0=x$, $y_n=y$, by the triangle inequality we
  have
  \begin{align*}
    d_X(x,y)
    &\quad\leq\quad
      d(x_0,y_0)+d(y_0,x_1)+\ldots+d(y_{n-1},x_n)+d(x_n,y_n)
    \\
    &\quad=\quad
      d(x_0,y_0)+d(x_1,y_1)+\ldots+d(x_{n-1},y_{n-1})+d(x_n,y_n)
  \end{align*}
  and therefore $d_X(x,y)\leq d_{X/{\approx}}(x,y)$. Conversely, since $x,y$ is
  a chain, we have that $d_{X/{\approx}}(x,y)\leq d_X(x,y)$.
\end{proof}

\noindent
In particular, given a generalized metric space~$X$, the free separated space
can be constructed as~$X/{\approx}$, where $\approx$ is the smallest equivalence
relation such that $d(x,y)=0$ or $d(y,x)=0$ implies $x\approx y$, for
$x,y\in X$. The quotient map will be an isometry when the relation~$\approx$ is
instantaneous, \ie when $d(x,y)=0$ implies $d(y,x)=0$. The colimits of
non-generalized metric spaces usually considered in literature can be obtained
from those described here by freely separating the resulting colimits.


\subsubsection{Symmetric metric spaces}
The category of spaces which are symmetric forms an interesting subcategory, in
which most of the usual properties of (non-generalized) metric spaces are still
valid. They are sometimes also called \emph{pseudo-metric spaces} in the
literature~\cite{wilson1931quasi, albert1941note, kelly1963bitopological}.

\begin{defi}
  We write $\SGMet$ for the full subcategory of~$\GMet$ whose object are
  symmetric generalized metric spaces, \ie spaces $(X,d)$ such that
  $d(x,y)=d(y,x)$ for every~$x,y\in X$.
\end{defi}

\noindent
The constructions of Proposition~\ref{prop:gmet-limits} are easily checked to
preserve the symmetry of spaces:

\begin{lem}
  The category $\SGMet$ is complete and cocomplete.
\end{lem}

\begin{rem}
  \label{rem:d-morphism}
  Notice that given a symmetric space $(X,d)$ and $x,x',y,y'\in X$, we have
  \[
  |d(x',y')-d(x,y)|
  \quad\leq\quad
  |d(x',y')-d(x,y')|+|d(x,y')-d(x,y)|
  \quad\leq\quad
  d(x,x')\vee d(y,y')
  \]
  the last step using the well-known ``reverse triangle inequality''. The
  distance~$d$ can thus be seen as a morphism $d:X\times X\to[0,\infty]$ in
  $\SGMet$.
\end{rem}

\noindent
Any generalized metric space can canonically be symmetrized in two ways as
follows.

\begin{prop}
  \label{prop:free-s-metric}
  The forgetful functor $\SGMet\into\GMet$ admits a left adjoint sending a space
  $(X,d)$ to the symmetric space $(X,\dsym{d})$ where the metric~$\dsym{d}$ is
  called the \emph{symmetric metric} generated by~$d$ and is defined by
  \[
  \dsym{d}(x,y)
  \qeq
  \bigwedge_{x=x_0,x_1,\ldots,x_{2n},x_{2n+1}=y}\sum_{i=0}^nd(x_{i+1},x_i)+d(x_{i+1},x_{i+2})
  \]
  It also admits a right adjoint sending $(X,d)$ to the space $(X,d^\vee)$
  equipped with the distance defined by $d^\vee(x,y)=d(x,y)\vee d(y,x)$. The
  forgetful functor thus preserves limits and colimits.
\end{prop}

\begin{exa}
  \label{ex:co-free-sym-d}
  For instance, consider the directed plane $\dR^2$, whose distance is given by
  $d((x_1,x_2),(y_1,y_2))=d_{\dR}(x_1,y_1)\vee d_{\dR}(x_2,y_2)$, and $d_{\dR}$
  is described in Example~\ref{ex:dr}. The left and right adjoint of the
  proposition respectively equip $\R^2$ with the distances~$\dsym d$ and $d^\vee$
  such that
  \[
  \dsym{d}((x_1,x_2),(y_1,y_2))
  \qeq
  \begin{cases}
    |y_1-x_1|\vee|y_2-x_2|&\text{if $(y_1-x_1)(y_2-x_2)\geq 0$}\\
    |y_1-x_1|+|y_2-x_2|&\text{if $(y_1-x_1)(y_2-x_2)\leq 0$}\\
  \end{cases}
  \]
  and
  \[
  d^\vee((x_1,x_2),(y_1,y_2))
  \qeq
  \begin{cases}
    |y_1-x_1|\vee|y_2-x_2|&\text{if $(y_1-x_1)(y_2-x_2)\geq 0$}\\
    \infty&\text{if $(y_1-x_1)(y_2-x_2)\leq 0$}\\
  \end{cases}
  \]
\end{exa}

\begin{rem}
  \label{rem:dsym-sym}
  For every points $x$ and $y$ in a generalized metric space $(X,d)$, we have
  $d(x,y)\geq\dsym{d}(x,y)$ (as witnessed by the unit of the monad on~$\GMet$
  induced by the adjunction).
  Moreover, if $d$ is a symmetric metric then $\dsym{d}=d$ (the comonad
  on~$\SGMet$ induced by the adjunction is the identity comonad).
\end{rem}

\begin{rem}
  We say a \emph{self-dual} category~$\C$ is a category equipped with an
  isomorphism $\C\cong\C^\op$ which is the identity on objects. The above
  construction for the left adjoint functor is a particular case of the
  construction of the free self-dual $\mcV$-category on a $\mcV$-category, and a
  variant of the more well-known construction of the enveloping $\mcV$-groupoid
  of a $\mcV$-category.
\end{rem}

\subsubsection{The underlying topological space of a metric space}
\label{sec:gmet-top}

We are now interested in equipping our spaces with a topology, \ie in
constructing a decent functor $\GMet\to\Top$ (in particular, we want directed
paths, see Definition~\ref{def:path}, to be continuous). In the case of
symmetric spaces, a satisfactory answer is provided by the usual functor
$\SGMet\to\Top$ sending a symmetric space $(X,d)$ to the topological space~$X$
equipped with the metric topology, which is generated by open balls
\[
B^\varepsilon(x)\qeq\setof{y\in X}{d(x,y)<\varepsilon}
\]
for $x\in X$ and $\varepsilon>0$. Since the maps in $\GMet$ are nonexpansive, it
can be easily checked that they are continuous \wrt to the metric topology and
the functor is well-defined.
For (general) generalized metric spaces, the situation is however not so
clear. Before presenting our answer, we would first like to explain why other
``intuitive'' options are not satisfactory. Given a space $(X,d)$ and $x\in X$,
one can construct \emph{past} and \emph{future open balls} of radius
$\varepsilon>0$, which are respectively defined by
\[
B_-^{\varepsilon}(x)=\setof{y\in X}{d(y,x)<\varepsilon}
\qquad\qquad\qquad
B_+^{\varepsilon}(x)=\setof{y\in X}{d(x,y)<\varepsilon}
\]
Considering the topology generated by future open balls is intuitive, but leads
to a topology which is too fine: for instance, the map $f:\dui\to\dui$ such that
$f(t)=0$ if $t<0.5$ and $f(t)=1$ otherwise would be continuous \wrt this
topology.
Another option could also consider the topology generated by sets of the form
$B_-^{\varepsilon}(x)\cup B_+^{\varepsilon}(x)$. However, an easy computation
shows that in $\dR^2$, the resulting topology is the discrete one. Namely, a
generating open set is drawn on the left below (continuous lines mean that the
border is included and dotted ones that the border is excluded). In the middle
right, an open set is shown (it is obtained by intersecting two generating open
sets as drawn on the middle left). Finally, by intersecting two such open sets,
we can obtain an open set reduced to a point, as shown on the right.
\[
\begin{tikzpicture}
  \filldraw[color=lightgray] (0,0) rectangle (-1,-1);
  \filldraw[color=lightgray] (0,0) rectangle (1,1);
  \draw (-1,0) -- (1,0);
  \draw (0,-1) -- (0,1);
  \draw[dotted] (-1,0) -- (-1,-1) -- (0,-1);
  \draw[dotted] (1,0) -- (1,1) -- (0,1);
\end{tikzpicture}
\qquad\qquad
\begin{tikzpicture}
  \filldraw[color=lightgray] (0,0) rectangle (-1,-1);
  \filldraw[color=lightgray] (0,0) rectangle (1,1);
  \filldraw[color=lightgray] (.75,.75) rectangle (-.25,-.25);
  \filldraw[color=lightgray] (.75,.75) rectangle (1.75,1.75);
  \filldraw[color=gray] (-.25,-.25) rectangle (0,0);
  \filldraw[color=gray] (0,0) rectangle (.75,.75);
  \filldraw[color=gray] (.75,.75) rectangle (1,1);
  \draw (-1,0) -- (1,0);
  \draw (0,-1) -- (0,1);
  \draw (-.25,.75) -- (1.75,.75);
  \draw (.75,-.25) -- (.75,1.75);
  \draw[dotted] (-1,0) -- (-1,-1) -- (0,-1);
  \draw[dotted] (1,0) -- (1,1) -- (0,1);
  \draw[dotted] (-.25,.75) -- (-.25,-.25) -- (.75,-.25);
  \draw[dotted] (1.75,.75) -- (1.75,1.75) -- (.75,1.75);
\end{tikzpicture}
\qquad\qquad
\begin{tikzpicture}
  \filldraw[color=lightgray] (-.25,-.25) rectangle (0,0);
  \filldraw[color=lightgray] (0,0) rectangle (.75,.75);
  \filldraw[color=lightgray] (.75,.75) rectangle (1,1);
  \draw (-.25,0) -- (.75,0);
  \draw (0,.75) -- (1,.75);
  \draw (0,-.25) -- (0,.75);
  \draw (.75,0) -- (.75,1);
  \draw[dotted] (-.25,0) -- (-.25,-.25) -- (0,-.25);
  \draw[dotted] (1,.75) -- (1,1) -- (.75,1);
\end{tikzpicture}
\qquad\qquad
\begin{tikzpicture}
  \filldraw[color=lightgray] (-.25,-.25) rectangle (0,0);
  \filldraw[color=lightgray] (0,0) rectangle (.75,.75);
  \filldraw[color=lightgray] (.75,.75) rectangle (1,1);
  \draw (-.25,0) -- (.75,0);
  \draw (0,.75) -- (1,.75);
  \draw (0,-.25) -- (0,.75);
  \draw (.75,0) -- (.75,1);
  \draw[dotted] (-.25,0) -- (-.25,-.25) -- (0,-.25);
  \draw[dotted] (1,.75) -- (1,1) -- (.75,1);
  \begin{scope}[shift={(.75,-.75)}]
    \filldraw[color=lightgray] (-.25,-.25) rectangle (0,0);
    \filldraw[color=lightgray] (0,0) rectangle (.75,.75);
    \filldraw[color=lightgray] (.75,.75) rectangle (1,1);
    \draw (-.25,0) -- (.75,0);
    \draw (0,.75) -- (1,.75);
    \draw (0,-.25) -- (0,.75);
    \draw (.75,0) -- (.75,1);
    \draw[dotted] (-.25,0) -- (-.25,-.25) -- (0,-.25);
    \draw[dotted] (1,.75) -- (1,1) -- (.75,1);
  \end{scope}
\end{tikzpicture}
\]
We hope to have convinced the reader that the canonical way of assigning a
topological space to a generalized metric space is the following 

\begin{defi}
  We define the \emph{forgetful functor} $\GMet\to\Top$ as the composite functor
  $\GMet\to\SGMet\to\Top$, where the first functor $\GMet\to\SGMet$ is the left
  adjoint to the forgetful functor constructed in
  Proposition~\ref{prop:free-s-metric} and the second functor $\SGMet\to\Top$ is
  the forgetful functor described at the beginning of this section.
\end{defi}

\begin{rem}
  Example~\ref{ex:co-free-sym-d} illustrates why the left adjoint (and
  not the right adjoint) is suitable here.
\end{rem}

\noindent
In the following, when we implicitly see a generalized metric space $(X,d)$ as a
topological space, it will always be in this way, \ie with topology generated by
the open balls $B^\varepsilon(x)=\setof{y\in X}{\dsym{d}(x,y)<\varepsilon}$.
%
  %
%
Morphisms of spaces being nonexpansive functions, the following can easily be
shown:

\begin{lem}
  \label{lem:gmet-mor-cont}
  A morphism $f:X\to Y$ of generalized metric spaces is continuous. In
  particular, by Remark~\ref{rem:d-morphism}, given a space $(X,d)$, the
  distance $\ol{d}:X\times X\to X$ is continuous.
\end{lem}

\begin{rem}
  Notice the distance is not necessarily continuous \wrt to the induced
  topology. For instance, in $\dR^2$ (see the figure on the left),
  \[
  \begin{tikzpicture}
    \draw[->] (-.2,0) -- (1,0);
    \draw[->] (0,-.2) -- (0,1);
    \gvert{(.2,.5)} node[above]{$x$};
    \gvert{(.8,.7)} node[right]{$y$};
    \gvert{(.8,.3)} node[right]{$y'$};
  \end{tikzpicture}
  \qquad\qquad\qquad\qquad
  \begin{tikzpicture}
    \draw (0,0) circle (.5);
    \gvert{(90:.5)} node[above] {$x$};
    \gvert{(140:.5)} node[above] {$y$};
    \gvert{(50:.5)} node[above] {$y'$};
  \end{tikzpicture}
  \]
  when~$y$ ``moves'' vertically to~$y'$, the value of $d(x,y)$ suddenly
  ``jumps'' from a finite value to $\infty$. Similarly in $\dcircle$ (figure on
  the right), when~$y$ moves to~$y'$, the value of $d(x,y)$ jumps from~$0$ to
  $2\pi$.
\end{rem}

\begin{prop}
  \label{prop:gmet-top-preservation}
  The functor $\GMet\to\Top$ preserves finite limits and small
  coproducts.
\end{prop}
\begin{proof}
  It can be directly checked that the functor preserves binary products,
  equalizers, and small coproducts, as described in Section~\ref{sec:limits},
  see also~\cite{goubault2013non}.
\end{proof}

\begin{rem}
  \label{rem:gmet-top-colimits}
  The functor does not preserve coequalizers.
  Namely, if we consider the colimit of Example~\ref{ex:In}, the colimit as
  topological spaces has points $0$ and $1$ separated in~$I_\omega$, whereas
  $d_\omega(0,1)=0$ and thus the points are not separated in the topological
  space associated to the colimit of generalized metric spaces.
\end{rem}

\begin{rem}
  \label{rem:gmet-top-limits}
  The functor does not preserve infinite products. Namely, given a
  family~$(X_i)$ of generalized metric spaces, the underlying topology of the
  generalized metric space $\prod_iX_i$ is the box topology, which is finer than
  the product topology.
\end{rem}

\subsubsection{Directed paths}
In this section, we define the notion of directed path in a generalized metric
space.

\begin{defi}
  \label{def:path}
  A \emph{path} in a space~$X$ is a continuous function $\gamma:I\to X$. Notice
  that~$\gamma$ is not required to be nonexpansive, \ie it is a path in the
  underlying topological space. The point $x=\gamma(0)$ (\resp $y=\gamma(1)$) is
  called the \emph{source} (\resp \emph{target}) of the path. We often write
  $\gamma:x\pathto y$ to indicate that $\gamma$ is a path from~$x$ to~$y$.
  The \emph{length}~$\length\gamma$ of a path~$\gamma$ is defined by
  \[
  \length\gamma\qeq\bigvee_{n\in\N}\bigvee_{0=t_0<t_1<\ldots<t_n=1}\sum_id(\gamma(t_i),\gamma(t_{i+1}))
  \]
  A path~$\gamma$ is called \emph{directed} (or a \emph{dipath}, or a
  \emph{rectifiable} path) when its length is finite.
\end{defi}

\begin{exa}
  In the directed circle $\dcircle$ (see Example~\ref{ex:dcircle}), directed
  paths are those which are ``turning counter-clockwise'', \ie maps of the form
  $t\mapsto\ce^{\ci\theta(t)}$ where $\theta:I\to~[0,2\pi]$ is increasing
  modulo~$2\pi$.
\end{exa}

\begin{rem}
  By the triangle inequality, we always have
  $\length\gamma\geq d(\gamma(0),\gamma(1))$ for an arbitrary path~$\gamma$.
\end{rem}


\begin{rem}
  By Lemma~\ref{lem:gmet-mor-cont}, given a path $\gamma:I\to X$ and a morphism
  $f:X\to Y$, the morphism $f\circ\gamma:I\to Y$ is also a path. Since $f$ is
  nonexpansive, the path $f\circ\gamma$ is directed when~$\gamma$ is: morphisms
  preserve the direction of paths.
\end{rem}


\begin{rem}
  Since a (non-generalized) metric space is symmetric, it does not contain any
  information about a ``direction of time'' and it is thus natural to expect
  that every path is directed in this case. However, this is not the case
  because it is well known that, in general, a path is not necessarily
  rectifiable (for instance the well-known topologist's sine curve in~$\R^2$
  defined by $\gamma(0)=0$ and $\gamma(t)=t\sin(1/t)$ for $0<t\leq 1$). The
  definition of directed path given in Definition~\ref{def:path} is thus not
  completely satisfactory yet, and we leave the investigation of a more general
  notion for future works; we expect that in the case of geometric semantics of
  concurrent programs, the current and the right notions of directed paths
  coincide.
  Notice that many simple generalizations of the notion of directedness that one
  could think of in order to overcome this problem do not work. For instance,
  one could declare that a path~$\gamma$ is directed if we have
  $d(\gamma(t_1),\gamma(t_2))<\infty$, for every $t_1,t_2\in I$ with
  $t_1<t_2$. However, with this definition, every path of the directed unit
  circle (see Example~\ref{ex:dcircle}) would be directed.
\end{rem}

\noindent
Of course, the usual notions of homotopy and dihomotopy directly extend to our
setting:

\begin{defi}
  A \emph{homotopy} between two paths $\gamma,\rho:I\to X$ is a continuous
  function $h:I\times I\to X$ such that $h(0,-)=\gamma$ and $h(1,-)=\rho$. Such
  a homotopy is a \emph{dihomotopy} when $h(t,-)$ is a directed path for every
  $t\in I$.
\end{defi}

Intuitively, a morphism $\directed{[0,a]}\to X$ is the same as a rectifiable
path in~$X$. However, because the maps take distance in account (they are
nonexpansive), this is only true up to an expected equivalence relation on paths:
if~$a$ is too small, there is no possible parametrisation of the path.
A \emph{partial reparametrization} is a continuous non-decreasing map
$\theta:I\to I$ and a \emph{reparametrization} is a surjective partial
reparametrization. A \emph{trace} is an equivalence class of paths under the
relation identifying two paths~$f$ and~$g$ whenever there exists a
reparametrization~$\theta$ such that $g=f\circ\theta$. It can easily be shown
that length is well-defined on traces and a trace is \emph{rectifiable} when its
length is finite. The following proposition shows that in every trace containing
a rectifiable path, there is a canonical one which corresponds to a morphism
in~$\GMet$. Its proof is a direct generalization of the one in the classical
case~\bridson{Prop.~I.1.20}{12}. Given $t,t'\in\R$ with $t<t'$, we write
$\iota_{[t,t']}:I\to[t,t']$ for the function such that
$\iota_{[t,t']}(u)=t+(t'-t)u$.

\begin{lem}
  Suppose that~$(X,d)$ is a separated space and $\gamma:I\to X$ a directed path
  of length~$a=\length\gamma$. The function $\lambda:I\to[0,a]$ such that
  $\lambda(t)=\length{\gamma\circ\iota_{[0,t]}}$ is well-defined, continuous and
  non-decreasing, and there exists a unique morphism $\tilde\gamma:[0,a]\to X$
  such that $\tilde\gamma\circ\lambda=\gamma$ and
  $\length{\tilde\gamma\circ\iota_{[0,t]}}=t$ for every $t\in[0,a]$.
\end{lem}

\begin{rem}
  In order to see why we need the separation hypothesis, consider a path in a
  space~$X$ equipped with the constant distance~$d_0$ equal to~$0$. Any path
  in~$X$ has length~$0$ and the lemma is clearly wrong.
\end{rem}

\noindent
As a direct corollary of previous lemma, we have:

\begin{prop}
  \label{prop:traces-maps}
  Given a separated space~$X$, rectifiable traces are in bijection with maps
  $\directed{[0,a]}\to X$ sending distance to length, with $a\geq 0$.
\end{prop}

Finally, we briefly mention here that we can recover traditional notions of
directed spaces from generalized metric spaces.
The notion of d-space has emerged as the ``standard'' model for topological
spaces equipped with a notion of time direction~\cite{grandis2009directed}. It
consists in a space together with the specification of which paths are to be
considered as directed.

\begin{defi}
  A \emph{d-space} $(X,dX)$ consists of a topological space~$X$ together with a
  set~$dX\subseteq X^I$ of paths in~$X$, whose elements are called
  \emph{d-paths}, which contains all constant paths and is closed under
  concatenation and reparametrization. We write $\dTop$ for the category of
  d-spaces with continuous functions preserving d-paths as morphisms.
\end{defi}

\begin{prop}
  The operation which to a metric space associates its underlying topological
  space together with the set of directed paths defines a functor
  $\GMet\to\dTop$.
\end{prop}

\noindent
A metric space is \emph{acyclic} when every path $f:x\pathto x$ with the same
source and target is constant. To such a space, we can associate the following.

\begin{defi}
  \label{def:pospace}
  A \emph{pospace} $(X,\leq)$ consists of a topological space~$X$ equipped with
  a partial order~$\leq$. We write $\POSpace$ for the category of pospace and
  non-decreasing continuous maps.
\end{defi}

\begin{rem}
  Some authors impose additional restrictions on pospaces such as the fact
  that~$\leq$ is a closed subset of~$X\times X$, or at least that the limit of
  an increasing sequence of points is its supremum.
\end{rem}

\begin{prop}
  The operation which to a metric space~$(X,d)$ associates its underlying
  topological space~$X$, equipped with the partial order~$\leq$ such that
  $x\leq y$ if and only if $d(x,y)<\infty$ for every $x,y\in X$, extends to a
  functor from the category of acyclic generalized metric spaces to the category
  of pospaces.
\end{prop}

\subsubsection{Geodesic and length spaces}
\label{sec:length-spaces}
In this section, we briefly turn our attention to length spaces, since all
geodesic spaces are such, and the CAT(0) condition we are interested in is
formulated on such spaces, by considering the ``size'' of triangles whose sides
are geodesics.

\begin{defi}
  \label{def:intrinsic}
  Given a space $(X,d)$, the \emph{intrinsic metric} $\dint{d}$ is defined by
  \[
  \dint{d}(x,y)
  \qeq
  \bigwedge_{\gamma:x\pathto y}\length\gamma
  \]
  A space equipped with its intrinsic metric is called a \emph{length space}.
\end{defi}

\begin{exa}
  Consider the circle
  $\ndcircle=\setof{\ce^{\ci 2\pi\theta}}{\theta\in\R}\subseteq\Cplx$ equipped
  with the distance induced by the euclidian distance on~$\Cplx$, \ie
  $d(x+\ci y, x'+\ci y')=\sqrt{(x'-x)^2+(y'-y)^2}$. This space is not a length
  space. The associated intrinsic metric is
  \[
  \dint{d}(\ce^{\ci\theta},\ce^{\ci\rho})
  \qeq
  \bigwedge\setof{\rho'-\theta',\theta'-\rho'}{\theta'=\theta, \rho'=\rho\mod 2\pi}
  \]
\end{exa}


\begin{defi}
  \label{def:geodesic}
  A directed path $\gamma:I\to X$ is a \emph{geodesic} when for every $t<t'$
  in~$I$ we have $d(\gamma(t),\gamma(t'))=\lambda(t'-t)$ with
  $\lambda=d(\gamma(0),\gamma(1))$. A space is \emph{geodesic} (\resp
  \emph{uniquely geodesic}) when between any two points at finite distance there
  exists a geodesic (\resp a unique geodesic).
\end{defi}


\begin{rem}
  Every geodesic space is a length space, but the converse is not true in
  general (for instance, the Hopf-Rinow theorem~\cite{hopf-rinow} provides
  sufficient conditions on spaces so that a length space satisfying those
  conditions is geodesic).
\end{rem}

\begin{exa}
  The space $\R^2$ equipped with the euclidian distance is geodesic. The
  subspace $\R^2\setminus\set{(0,0)}$ is a length space but is not geodesic
  since there is no path of length~$2=d(x,y)$ from $x=(0,-1)$ to $y=(0,1)$
  (however, there exists a path of length $2+\varepsilon$ from $x$ to $y$ for
  arbitrarily small $\varepsilon>0$).
\end{exa}

\noindent
In the case of a geodesic space, Proposition~\ref{prop:traces-maps} can be
reformulated as follows.

\begin{prop}
  Given a separated space~$X$, geodesic paths are in bijection with isometries
  $\directed{[0,a]}\to X$, with $a\geq 0$.
\end{prop}


\begin{lem}
  \label{lem:length-colimit}
  The subcategory of length spaces is closed under colimits.
\end{lem}
\begin{proof}
  See~\bridson{Lem.~I.5.20}{65}. It can be checked directly that the property of
  being a length space is preserved under taking coproducts and coequalizers.
\end{proof}

\subsection{Geometric realization of precubical sets}
\label{sec:pcs-greal}
We now consider two ways of associating a generalized metric space to a
precubical set, and study the properties of the resulting space. Since here we
are mainly concerned about precubical sets arising as the cubical semantics of
concurrent programs, we study mainly finite geometric precubical sets, but try
to provide more general hypothesis when these generalizations are easy to
make. Part of the current section is a reformulation -- and generalization -- in
the categorical language of properties studied in the framework of polyhedral
complexes~\cite{bridson2009metric}. The two main results of this section are
that metric and topological geometric realization provide comparable
constructions (Theorem~\ref{thm:geom-real-gmet-top}) and that realization of
\NPC precubical sets are CAT(0) cubical complexes (Theorem~\ref{thm:pcs-cat0}).

\subsubsection{Geometric realization}

Recall from Example~\ref{ex:ui} that we write $I=[0,1]$ for the \emph{standard
  interval}, and we write~$I^n$ for the \emph{standard (topological) $n$-cube}
obtained as the product of $n$ copies of $I$. We write $\delta^-=0$ and
$\delta^+=1$. Given $i$ with $0\leq i<n$, and $\epsilon=-$ (\resp $\epsilon=+$),
the set of points in $I^n$ whose $i$-th coordinate is $\delta^\epsilon$ is
isomorphic to $I^{n-1}$ and called the $i$-th back (\resp front) \emph{face} of
the cube. Now, consider the functor $I:\pcc\to\GMet$ such that the image of an
object $n\in\pcc$ is $I^n$, and the images of
morphisms~$\varepsilon^\epsilon_{i,n}$ are the morphisms
$I(\varepsilon^\epsilon_{i,n}):I^{n}\to I^{n+1}$ such that
\[
I(\varepsilon^\epsilon_{i,n})(x_0,\ldots,x_{n-1})
\qeq
(x_0,\ldots,x_{i-1},\delta^\epsilon,x_i,\ldots,x_n)
\]
\ie they are the inclusions of the faces of
a cube into the cube. A topological space can be obtained from a precubical
set~$C$ by taking a topological $n$-cube for each element of $C(n)$ an gluing
them according to faces.

\begin{defi}
  \label{def:greal}
  The \textbf{geometric realization} functor $\greal{-}:\hat\pcc\to\GMet$ is the
  functor obtained as the left Kan extension of $I:\pcc\to\GMet$ along the
  Yoneda embedding $y:\pcc\to\hat\pcc$. More explicitly, given a cubical set
  $C$, its geometric realization is the space
  \[
  \greal{C}
  \qeq
  \pa{\bigsqcup_{n\in\N}I^n\times C_n}/\approx
  \]
  where $\approx$ is the smallest equivalence relation such that
  $(x,c)\approx(y,d)$ whenever $d=\partial^\epsilon_{i,n}(c)$ for some indices
  $\epsilon,i,n$ and $I(\varepsilon^\epsilon_{i,n})(y)=x$.
\end{defi}

\begin{rem}
  One can define similarly a functor $I:\pcc\to\Top$ (by post-composing previous
  functor with the forgetful functor $\GMet\to\Top$), which induces, by left Kan
  extension along the Yoneda embedding $y:\pcc\to\hat\pcc$, a
  \textbf{topological geometric realization} functor
  $\greal-:\hat\pcc\to\Top$. Unless we add the adjective ``topological'',
  ``geometric realization'' will always be meant in generalized metric spaces.
\end{rem}

\noindent
By Proposition~\ref{prop:free-s-metric}, the forgetful functor $\SGMet\to\GMet$
preserves colimits, and therefore geometric realization commutes with it. And
similarly, it commutes with the symmetrization functor $\GMet\to\SGMet$, which
is a left adjoint. In the following, we will mainly focus on symmetric spaces
for simplicity.

For every $n$-cube $c\in C(n)$ there is a canonical morphism of metric spaces
$\iota_c:I^n\to\greal{C}$ such that the image of $x\in I^n$ is the equivalence
class of $(x,c)$. Formally, these morphisms can be obtained as the cocone arrows
of the colimit defining the geometric realization. They allow one to see the
cube $I^n$ as a ``subspace'' of $\greal{C}$. Notice that there is no a priori
reason why these morphisms should be isometries; we will however see below that,
they are isometries ``locally''. We say a cube~$c$ \emph{contains} a
point~$x\in|C|$ when $x$ is in the image of~$\iota_c$. When the precubical
set~$C$ has no self-intersection (in particular when~$C$ is geometric) the
function $\iota_c$ is injective, and we write $\iota_c^{-1}$ for its partial
inverse.  In this case, following~\bridson{Def.~7.8}{100}, given a point
$x\in\greal{C}$ and an $n$-cube $c\in C(n)$ such that $x$ occurs in the image of
$\iota_c$, we write
\[
\varepsilon(x,c)
\qeq
\bigwedge\setof{d_{I^n}(\iota^{-1}_c(x),K)}{\text{$K$ is a face of $I^n$ not containing~$x$}}
\]
with, by convention, $\varepsilon(x,c)=\infty$ when $c$ is a $0$-cube (above,
the distance between $\iota^{-1}_c(x)$ and $K$ is taken to be the infimum of
distances between $\iota^{-1}_c(x)$ and some point of $K$). We then define
\[
\varepsilon(x)
\qeq
\bigwedge\setof{\varepsilon(x,c)}{\text{$n\in\N$ and $c\in C(n)$ is such that $x$ belongs to the image of $\iota_{c}$}}
\]
It is easy to show that $\varepsilon(x)>0$ for any point~$x\in\greal{C}$,
see~\bridson{7.33}{112}. This constant, called the \emph{escape distance}
from~$x$, intuitively reflects the fact that a path (or a chain) starting
from~$x$ and going outside a simplex containing~$x$ will at least be of length
$\varepsilon(x)$. Otherwise said, the distance from~$x$ to a point~$y$
``near~$x$'' (\ie at distance less than $\varepsilon(x)$) will be the same
whether we consider the distance within the face or in the whole complex:

\begin{lem}
  \label{lem:local-distance}
  Suppose given a precubical set~$C$ with no self-intersection, and a point~$y$
  such that $d_{|C|}(x,y)<\varepsilon(x)$. Then any cube~$c$ which contains~$y$
  also contains~$x$ and we have
  $d_{\greal{C}}(x,y)=d_{I^n}(\iota_c^{-1}(x),\iota_c^{-1}(y))$.
\end{lem}
\begin{proof}
  See~\bridson{Lem.~7.9}{100}. Fix a ``chain'' $u=x_0,x_1,x_2,\ldots,x_n$ in
  $\greal{C}$ from $x=x_0$ to $y=x_n$, \ie a sequence of points such that $x_i$
  and $x_{i+1}$ both belong to $\iota_{c_i}$ for some given cube $c_i\in C(n_i)$
  (to be precise, the $c_i$ are part of the data defining the chain). Its length
  is by definition
  $l(u)=\sum_i d_{I^{n_i}}(\iota_{c_i}^{-1}(x_i),\iota_{c_i}^{-1}(x_{i+1}))$ and
  the distance $d_{\greal{C}}(x,y)$ is the infimum of lengths of such
  chains. The points~$x_1$ and~$x_2$ belong to the image of~$\iota_{c_1}$. If we
  suppose moreover that $l(u)<\varepsilon(x)$ then $x$ also belongs
  to~$\iota_{c_1}$, and by the triangle inequality the chain
  $u'=x_0,x_2,x_3,\ldots,x_n$ has smaller length. We can conclude by a simple
  induction.
\end{proof}

\begin{lem}
  \label{lem:greal-separated}
  Given a precubical set~$C$ with no self-intersection, its geometric
  realization~$\greal{C}$ is separated.
\end{lem}
\begin{proof}
  This is a corollary of Lemma~\ref{lem:local-distance}, since the spaces~$I^n$
  are separated.
\end{proof}


\begin{lem}
  The geometric realization of a precubical set is a length space.
\end{lem}
\begin{proof}
  The spaces $I^n$ are length spaces and, by Lemma~\ref{lem:length-colimit},
  length spaces are closed under colimits.
\end{proof}

\begin{lem}
  The geometric realization of a finite dimensional precubical set~$C$ is a
  complete space.
\end{lem}
\begin{proof}
  The ideas is as follows, see \bridson{Thm.~7.13}{101}, of which this proof is
  a variant, for details. Given a Cauchy sequence $(x_n)$ in~$\greal{C}$, we
  show that it contains a convergent subsequence. The points~$x_i$ belong to the
  image of $\iota_{c_i}$ for some $n_i$-cell $c_i\in C(n_i)$. Up to taking a
  subsequence, we can suppose that all the $n_i$ are equal and we write~$n$ for
  the corresponding integer: the points~$c_i$ are the images $\iota_{c_i}(x_i)$
  for some points $x_i\in I^n$. Up to taking a further subsequence, we can
  suppose that the sequence~$(x_i)$ converges to a point~$x_\infty$. Now, noting
  that the $\iota_{c_i}$ are nonexpansive, using an argument based on
  $\varepsilon(x_\infty)$ as in the proof of Lemma~\ref{lem:length-colimit}, one
  can show that the $c_i$ are all equal to some~$c_\infty$ for~$i$ big enough,
  and conclude that the subsequence converges to $\iota_{c_\infty}(x_\infty)$
  in~$\greal{C}$.
\end{proof}

  %

We have seen in Lemma~\ref{lem:local-distance} that the geometric realization of
a precubical set is locally isometric to a space~$I^n$. The proof used here does
not really depend on the choice of the functor~$I:\pcc\to\GMet$ used to specify
the realization of representable precubical sets, and would extend to other
choices. We however conjecture that, with this specific choice, this property is
more ``global'' in the sense that the canonical inclusions of standard cubes in
the realizations are in fact isometries. The intuition behind this is that
within a cube any two points are at distance at most one, and that a non-trivial
chain going outside a cube has to cross another cube and will thus be of length
more than one.

\begin{conjecture}
  \label{conj:canonical-isometry}
  Given a finite dimensional geometric precubical set~$C$, $n\in\N$ and
  $c\in C(n)$ the morphism $\iota_c:I^n\to\greal{C}$ is an isometry.
\end{conjecture}

\begin{rem}
  Notice that the above conjecture is not true if the precubical set is not
  supposed to be geometric. For instance, consider the following precubical
  set~$C$:
  \[
  \vxym{
    c'\ar@(ul,ur)^c
  }
  \]
  The morphism $\iota_c:I^1\to\greal{C}$ is clearly not an isometry.
  Moreover, the property strongly depends on the functor $I:\pcc\to\GMet$ we
  used. For instance, suppose that we had defined the geometric realization
  using the functor $I$ such that $I(n)$ is the set $I^n$ equipped with the
  euclidian metric (instead of the product metric), and consider the following
  precubical set~$C$:
  \[
  \vxym{
    x\ar[d]_g\ar[r]^f\ar@{}[dr]|\tile&\ar[d]^{g'}y_1\\
    y_2\ar[r]_{f'}\ar@/_2ex/[ur]|h&y
  }
  \]
  In the square, we have $d_{I^2}(y_1,y_2)=\sqrt 2$ but in the realization we
  have $d_{\greal{C}}(y_1,y_2)=1$. The above precubical set is not geometric,
  but it can easily be turned into a geometric one: instead of having one
  edge~$h$, consider a sequence of two edges $h_1\cc h_2$ (whose length is $2$),
  and instead of gluing it along the diagonal of a square, glue it along the
  diagonal of an hypercube of dimension~$5$, so that the length of the diagonal
  is $\sqrt{5}>2$.
\end{rem}

\subsubsection{The topology of the geometric realization}
\label{sec:El}
In previous section, we have introduced two possible geometric realizations: in
metric spaces and in topological spaces. We now show that the two agree for the
precubical sets of interest here. Notice that, by
Remark~\ref{rem:gmet-top-colimits}, we have no hope to show this using very
general arguments since the forgetful functor $\SGMet\to\Top$ does not preserve
colimits (nor even finite ones).
Consider a symmetric generalized metric space~$(X,d)$. We write
$U:\SGMet\to\Top$ for the forgetful functor described in
Section~\ref{sec:gmet-top}.

\begin{lem}
  \label{lem:gr-id-cont}
  Given an equivalence relation~$R$ on~$X$, the underlying sets of~$(UX)/R$ and
  $U(X/R)$ are the same and the identity function $f:(UX)/R\to U(X/R)$ is
  continuous.
\end{lem}
\begin{proof}
  The underlying sets of both~$(UX)/R$ and~$U(X/R)$ are the same because the
  forgetful functors $\Top\to\Set$ and $\SGMet\to\Set$ both admit a right
  adjoint (see Proposition~\ref{prop:gmet-limits}).
  Suppose given $V$ open in~$U(X/R)$: for every $x\in V$ there exists
  $\varepsilon>0$ such that $B^\varepsilon_{X/R}(x)\subseteq V$. Thus, for every
  $x\in f^{-1}(V)$, there exists $\varepsilon>0$ such that
  $f^{-1}(B^\varepsilon_{X/R}(f(x)))\subseteq f^{-1}(V)$. Given~$x\in X$, we
  have $B^\varepsilon_X(x)\subseteq B^\varepsilon_{X/R}(x)$ by
  Remark~\ref{rem:coeq-d}. Therefore, for every $x\in X$ such that
  $x\in\pi^{-1}(V)$, there exists $\varepsilon>0$ such that
  $B^\varepsilon_X(x)\subseteq\pi^{-1}(V)$, from which we can conclude that
  $\pi^{-1}(V)$ is open.
\end{proof}

\noindent
The following lemma is well-known~\cite[Thm.~4.17]{rudin1964principles}:

\begin{lem}
  \label{lem:bij-homeo}
  A continuous bijection~$f:X\to Y$ between topological spaces such that~$X$ is
  compact and~$Y$ is separated is an homeomorphism.
\end{lem}
\begin{proof}
  Given a closed set~$U$ of~$X$, $U$ is compact since~$X$ is, therefore $f(U)$
  is also compact, and therefore $f(U)$ is closed since~$Y$ is
  separated. Therefore $(f^{-1})^{-1}(U)=f(U)$ is closed and $f^{-1}$ is
  continuous since the preimage of a closed set is closed.
\end{proof}

\begin{prop}
  \label{prop:geom-real-gmet-top}
  The geometric realization of finite geometric precubical sets commutes with
  the forgetful functor $\SGMet\to\Top$.
\end{prop}
\begin{proof}
  Given a finite geometric precubical set~$C$, the colimit defining its
  geometric realization can be computed as a quotient of the space
  $X=\coprod_{n\in\N}\coprod_{x\in C(n)}I(n)$, by a relation that we
  denote~$R$. The space~$UX$ is compact, as a finite coproduct of compact spaces
  (the functor~$U$ commutes with coproducts, see
  Proposition~\ref{prop:gmet-top-preservation}), and thus the space $(UX)/R$ is
  compact too, as a quotient of a compact space. Since~$C$ is geometric, and
  thus without self-intersection, its geometric realization~$U(X/R)$ is
  separated by Lemma~\ref{lem:greal-separated}. Finally, the identity function
  $(UX)/R\to U(X/R)$ is continuous by Lemma~\ref{lem:gr-id-cont} and thus an
  homeomorphism by Lemma~\ref{lem:bij-homeo}.
\end{proof}

\noindent
Since the topology of a space is determined locally, previous proposition
can be generalized to the case of precubical sets which are only locally finite:

\begin{defi}
  A precubical set is \emph{locally finite} if for every vertex~$x$, only
  finitely many cubes have $x$ as iterated vertex.
\end{defi}

\begin{thm}
  \label{thm:geom-real-gmet-top}
  The geometric realization of locally finite geometric precubical sets commutes
  with the forgetful functor $\SGMet\to\Top$.
\end{thm}
\begin{proof}
  Suppose given a locally finite geometric precubical set~$C$ and a
  point~$x\in\greal{C}$. We write~$c$ for the carrier of~$x$, \ie the cube
  in~$C$ of lowest dimension which contains~$x$. Consider the precubical
  set~$C_c$defined as the smallest precubical subset of~$C$ which contains all
  the cubes having~$c$ as iterated face (the star of~$c$): its cubes are faces
  of maximal cubes with~$c$ as iterated face. Since~$C$ is supposed to be both
  locally finite and geometric, the precubical set~$C_c$ is easily shown to be
  finite and we can apply Proposition~\ref{prop:geom-real-gmet-top}. By
  Lemma~\ref{lem:local-distance}, a neighborhood of~$x$ in the geometric
  realization of~$C$ is isometric to a neighborhood of~$x$ in the geometric
  realization of~$C_c$ (because the lemma shows that the distance is the
  distance of standard cubes). Therefore, locally around~$x$, the metric
  topology is the colimit topology.
\end{proof}

\subsubsection{Cubical complexes}
Geometric precubical sets give rise to cubical complexes, in the following
sense.

\begin{defi}
  \label{def:cubical-complex}
  A \textbf{cubical complex} $K$ is a topological space of the form
  \[
  K\qeq\pa{\bigsqcup_{\lambda\in\Lambda}I^{n_\lambda}}/{\approx}
  \]
  where $\Lambda$ is a set, $(n_\lambda)_{\lambda\in\Lambda}$ is a family of
  integers, and $\approx$ is an equivalence relation, such that, writing
  $p_\lambda:I^{n_\lambda}\to K$ for the restriction of the quotient map
  $\bigsqcup_{\lambda\in\Lambda}I^{n_\lambda}\to K$, we have
  \begin{enumerate}
  \item for every $\lambda\in\Lambda$, the map $p_\lambda$ is injective,
  \item given $\lambda,\mu\in\Lambda$, if $p_\lambda(I^{n_\lambda})\cap
    p_\mu(I^{n_\mu})\neq\emptyset$ then there is an isometry from a face
    $J_\lambda$ of~$I^{n_\lambda}$ to a face $J_\mu$ of $I^{n_\mu}$ such that
    $p_\lambda(x)=p_\mu(y)$ if and only if $y=h_{\lambda,\mu}(x)$.
  \end{enumerate}
  A \emph{directed cubical complex} is defined similarly, as a d-space of the
  form $K=\pa{\bigsqcup_{\lambda\in\Lambda}I^{n_\lambda}}/{\approx}$.
\end{defi}

\begin{prop}
  \label{prop:geometric-pcs-real}
  The geometric realization of a geometric precubical set is a cubical complex.
\end{prop}

\noindent
Moreover, the metric of the geometric realization of a geometric precubical set
is precisely the intrinsic metric (as defined in~\bridson{Sect.~I.7.33}{112}).

\begin{rem}
  Notice that the geometric realization of precubical complexes of
  Example~\ref{ex:pcs-geom} which are not geometric are not precubical
  complexes.
\end{rem}


\begin{rem}
  The converse of the above lemma is not true. For instance, one can construct a
  Möbius strip as a complex, and it cannot be obtained as the geometric
  realization of a geometric precubical set because there is a ``mismatch of
  direction'' of the edges. One can actually show that directed geometric
  realizations of geometric precubical sets are precisely directed cubical
  complexes.
\end{rem}




\noindent
Finally, we mention the following useful fact, whose proof can be found
in~\bridson{7.33}{112}.
%

\begin{lem}
  A finite dimensional precubical complex is geodesic.
\end{lem}

\noindent
To sum up the above properties shown on geometric realizations of geometric
precubical sets, we have:

\begin{prop}
  \label{prop:fd-pcs}
  Finite dimensional precubical complexes are complete geodesic length spaces.
\end{prop}

\subsubsection{Directed geometric realization}
\label{sec:dgreal}
Interestingly, all the above properties can easily be checked to hold for a
variant of the geometric realization which is really ``directed'', and makes
much more sense from a concurrency point of view. In Definition~\ref{def:greal}
of geometric realization, we can replace the functor $I:\pcc\to\GMet$
sending~$n$ to $I^n$ by the functor $\dui:\pcc\to\GMet$ which send~$n$ to
$\dui^n$ and thus obtain a variant of the notion of geometric realization where
not every path is directed.

\begin{defi}
  \label{def:dgreal}
  The \textbf{directed geometric realization} $\dgreal{C}$ of a precubical complex
  is the functor obtained as the left Kan extension of $\dui:\pcc\to\GMet$ along
  the Yoneda embedding $y:\pcc\to\hat\pcc$.
\end{defi}

\begin{prop}
  The directed geometric realization of a finite dimensional geometric
  precubical set is a complete geodesic length space.
\end{prop}

Given a precubical set~$C$, we can consider its directed geometric
realization~$\dgreal{C}$. A vertex in~$C$ can be seen as a morphism $Y0\to C$
(where $Y$ denotes the Yoneda functor, see Section~\ref{sec:pcs}) and thus
induces, by functoriality of the realization, a point~$\dgreal{x}$ of
$\dgreal{C}$, which is given by the image of the morphism
$\dgreal{Y0}\to\dgreal{C}$. Similarly, a path (\resp dipath)~$f$ in~$C$ induces
a path (\resp dipath)~$\dgreal{f}$ in $\dgreal{C}$. It is easy to show that
given two dipaths $f,g:x\pathto y$ in~$C$, if they are homotopic (\resp
dihomotopic) then the dipaths $\dgreal{f}$ and $\dgreal{g}$ are also homotopic
(\resp dihomotopic). A converse property was shown by Fajstrup for locally
finite precubical sets~\cite{fajstrup2005dipaths} (more exactly, this result was
formulated for geometric realization of precubical sets as d-spaces, but the
proof adapts straightforwardly to the case of generalized metric spaces).

\begin{prop}
  Given $C$ geometric and locally finite, a geometric and locally finite
  precubical set, the homotopy (\resp dihomotopy) classes of dipaths from~$x$
  to~$y$ in~$C$ are in bijection with homotopy (\resp dihomotopy) classes of
  dipaths from $\dgreal{x}$ to $\dgreal{y}$ in~$\dgreal{C}$.
\end{prop}


From the previous proposition and Theorem~\ref{thm:hom-dihom}, one can deduce
that, for dipaths between any two points which are realizations of vertices
of~$C$, homotopy coincides with dihomotopy in~$\dgreal{C}$. Variants of this
properties have been shown in the literature: in the case of geometric semantics
of ``simple'' programs with mutexes~\cite{datc}, in the case of \NPC cube
complexes~\cite{Ghrist2} (see also next section).  A similar observation has
been made for spaces which are hypercontinuous lattices equipped with their
Lawson topology, such that their underlying space has connected CW
type~\cite{criterion}. In that case also, increasing paths that are homotopic
are in fact dihomotopic.
%


\subsection{Non-positively curved spaces}
\label{sec:npc}

\subsubsection{CAT(0) spaces}
We now recall the notion of non-positively curved space (also called CAT(0)
space), with the aim of showing that this notion is a topological counterpart of
the notion introduced for precubical sets (Definition~\ref{def:pcs-npc}). We
chose not to define this earlier in the paper in order to emphasize that this is
only an inspiration for our axiomatics on precubical sets, but the geometric
interpretation is not needed in order to work with non-positively curved
precubical sets (excepting, maybe, to forge intuitions).

This notion generalizes hyperbolic geometry by formalizing the observation that,
in a non-positively curved space, triangles with geodesic sides appear to be
thinner than in usual, flat, spaces. Its importance is due to the numerous
applications this notion has allowed. We only mention basic notions here, and
refer the reader to standard texts for a more detailed presentation of the
subject~\cite{gromov1987hyperbolic, bridson2009metric, ghys1990groupes}.
In the rest of this section, for simplicity, we only consider symmetric
generalized metric spaces, since the topology of a space does not depend on its
direction. Interesting possible generalizations in the directed case are
mentioned in Section~\ref{sec:conclusion}.

\begin{defi}
  \label{def:geodesic-triangle}
  A \emph{geodesic triangle} $\Delta(x,y,z)$ in a geodesic metric space~$X$
  consists of three points $x$, $y$, $z$ and three geodesics joining any pair of
  two. A \emph{comparison triangle} for a geodesic triangle $\Delta(x,y,z)$
  consists of an isometry $\ol-:\Delta(x,y,z)\to\R^2$ whose image is a geodesic
  triangle $\ol\Delta(\ol x,\ol y,\ol z)$, where~$\R^2$ is equipped with the
  usual euclidian distance $d_{\R^2}$.
\end{defi}

\noindent
We now recall the definition of non-positively curved space based on the
comparison axiom of Cartan, Alexandrof and Topogonov. This is the origin of the
name CAT(0), where the 0 refers to the fact that strictly positive or negative
curvature can also be defined in a similar way, but we will refrain from
introducing those in full generality here (see below).

\begin{defi}
  \label{def:cat0}
  A geodesic triangle $\Delta(x,y,z)$ is \emph{CAT(0)} when there exists a
  comparison triangle $\ol\Delta(\ol x,\ol y,\ol z)$ such that for any two
  points $p,q\in\Delta(x,y,z)$, we have $d(p,q)\leq d_{\R^2}(\ol p,\ol q)$. A
  geodesic metric space is \emph{CAT(0)} when every geodesic triangle is CAT(0),
  and \emph{locally CAT(0)} or \emph{non-positively curved} (\NPC) when every
  point admits a neighborhood which is CAT(0).
\end{defi}

\noindent
These spaces enjoy many nice properties such as being (locally) uniquely geodesic
and contractible. 
We refer the reader to~\cite{bridson2009metric} for more details about those.


Other notions of curvature can be defined if we consider other ``model spaces''
instead of~$\R^2$ in which we take our comparison triangles. This amounts, in
Definitions~\ref{def:geodesic-triangle} and~\ref{def:cat0}, to consider geodesic
triangles $\ol-:\Delta(x,y,z)\to M_k$ where $M_k$ is called a model space. For
$k=0$, we have $M_k=\R^2$ and we recover the above definitions; for $k>0$, $M_k$
is a 2\nbd{}sphere (the smaller $k$ is, the bigger its radius is); for $k>0$,
$M_k$ is a hyperbolic space. Compared to usual model spaces $M_0$, triangles in
model spaces $M_k$ are fatter when $k>0$ and slimmer when $k<0$, this getting
more accentuated when $|k|$ becomes large.
%
In particular, for $k=1$, we consider the standard 2\nbd{}sphere~$\sphere$
equipped with the distance~$d$ such that $d(x,y)\in[0,\pi]$ is defined by
$\cos(d(x,y))=\langle x,y\rangle$. A space is CAT(1) when for every
triangle~$\Delta(x,y,z)$ whose diameter is less than $2\pi$, there exists a
comparison triangle $\ol\Delta(\ol x,\ol y,\ol z)$ in $\sphere$, such that we
have $d(p,q)\leq d_{\sphere}(\ol p,\ol q)$ for every points
$p,q\in\Delta(x,y,z)$.

\subsubsection{The Gromov link condition}
\label{sec:Gromov-condition}

In this section we recall the characterization given by
Gromov~\cite{gromov1987hyperbolic} of CAT(0) cubical complexes, see
also~\bridson{II.5.4}{206} and~\bridson{II.5.20}{212}. We first introduce some
necessary definitions.

\begin{defi}[\bridson{Section 7.14}{102}]
  Given a geodesic metric space~$X$ and a point $x\in X$, the \emph{link}
  $\link_X(x)$ of $x$ in $X$ is the set of geodesics $\gamma:[0,a]\to X$ such
  that $\gamma(0)=x$, equipped with the compact-open topology, quotiented by the
  relation identifying two paths which coincide on an interval of the for
  $[0,\varepsilon[$ for some $\varepsilon>0$, \ie the ``directions'' from~$x$
  pointing into~$X$. This space can be metrized using the angle between two such
  directions.
\end{defi}

\noindent
This notion of link can be formally related to the one introduced for precubical
sets in Definition~\ref{def:pcs-link} by observing that the former is a
geometric realization of the latter (where the standard $n$-simplex is realized
as a spherical simplex with edges of length $\pi/2$).


\begin{defi}
  An abstract simplicial complex is \emph{flag} if whenever the complex contains
  the 1-skeleton of a simplex, it also contains the simplex: given vertices
  $x_1,\ldots,x_k$ such that $\set{x_i,x_j}$ are simplices for every indices
  $i,j$, the set $\set{x_1,\ldots,x_k}$ is also a simplex.
\end{defi}

\noindent
Again, this notion corresponds to the one of Definition~\ref{def:flag} through
the geometric realization.




Finally, we can explain in which way our conditions for non-positively curved
precubical sets (Definition~\ref{def:pcs-npc}) correspond to the traditional
geometric one.

\begin{thm}
  \label{thm:pcs-cat0}
  Given a finite dimensional geometric precubical set~$C$, the following are
  equivalent:
  \begin{enumerate}[align=left]
  \item[(i)] $\greal{C}$ is non-positively curved
  \item[(ii)] in $\greal{C}$ the link of every point is CAT(1)
  \item[(iii)] in $\greal{C}$ the link of every vertex is a flag complex
  \item[(iv)] in $C$ the link of every vertex is a flag complex
  \item[(v)] $C$ satisfies the cube condition
  \end{enumerate}
  Moreover, the following are equivalent:
  \begin{enumerate}[align=left]
  \item[(i')] $\greal{C}$ is CAT(0)
  \item[(ii')] $\greal{C}$ is uniquely geodesic
  \item[(iii')] $\greal{C}$ is non-positively curved and simply connected
  \item[(iv')] $C$ satisfies the cube condition and is simply connected
  \end{enumerate}
\end{thm}
\begin{proof}
  First, notice that because~$C$ is supposed to be finite dimensional, its
  realization is geodesic by Proposition~\ref{prop:fd-pcs}. The equivalence
  between (i), (ii) and (iii) is due to Gromov~\cite{gromov1987hyperbolic}, see
  also~\bridson{Thm. 5.20}{212}. The equivalence between (iii) and (iv) is
  immediate and the equivalence between (iv) and (v) was shown in
  Theorem~\ref{thm:pcs-npc-flag}.
  The equivalence between (i'), (ii') and (iii') is due to
  Gromov~\bridson{Thm.~5.5}{207}. The equivalence between (iii') and (iv')
  follows from the equivalence between (i) and (v).
\end{proof}


\section{Conclusion and future work}
\label{sec:conclusion}
We showed in this article that numerous concepts from different fields of
mathematics and computer science are closely related, non-positively curved
spaces and 
rewriting systems with the cube property in particular, \etc{}
These are important observations since, in many ways, non-positively curved
spaces are ``simpler'' to work with than general spaces: topological invariants
are much simpler, \eg the homology and homotopy groups have no torsion, and
their universal covering space are contractible \cite{bridson2009metric}.

There are many possible extensions to this work. First, in the case of a
non-positively curved precubical set, as used in the geometric semantics of
programs in Section~\ref{sec:pcs}, we believe that the coreflexion of the
adjunction between (some form of labeled) precubical sets and (labeled) prime
event structures of~\cite{goubault2012formal} is directly linked to the
universal dicovering in the sense of~\cite{datc}. This would enlighten the
relationship between these two models, in case we are dealing with concurrent
programs with mutual exclusion primitives.

Our original idea was that for non-positively curved spaces~$X$, dihomotopy of
dipaths can be characterized by the first relative homology group $H_1(X,A)$
where~$A$ is the subspace of~$X$ formed of the initial and final points of the
dipaths we are considering. By the work of Steiner~\cite{omegacatchaincomplex},
we know that the homotopy 2-categories, for some particular loop-free spaces,
are in one-to-one correspondence with some particular augmented directed
complexes, and that homotopic dipaths would map onto homologous elements of the
corresponding chain complex. In the case of non-positively curved spaces, this
implies that dihomotopic dipaths would be mapped onto homologous elements of the
underlying chain complex, in a not too distant manner, by Hurewicz theorem, \ie
up to abelianization. Unfortunately, it seems that we need to further shrink
down the class of spaces so that dihomotopy of dipaths can be fully
characterized by this first relative homology group, and there are probably
subtle combinatorial properties to be discovered.

We also expect to extend the results of this article to programs with resources
of arbitrary capacity, where we still have some control of the homotopy type of
the spaces generated by the semantics of such programs. This is, after all, part
of an endeavor to better understand the geometries that can appear in the study
of rewriting, concurrent and distributed systems, where other primitives
(test-and-set and all read-modify-write statements) will generate different
types of geometries.  The relationship between metrics and (di-)homotopies,
unravelled in this paper, may also pave the way to considering timed concurrent
models, studied up to deformation (and amenable to state-space reduction
techniques~\cite{datc}).  This has actually already been used in a different
context, in robotics, by Ghrist and collaborators, to design efficient
algorithms, see \eg \cite{Ghrist2}.



\end{document}